\documentclass{article}

\usepackage[utf8]{inputenc}
\usepackage{amssymb}
\usepackage{amsmath}
\usepackage{amsthm}
\usepackage{amsfonts}
\usepackage{graphicx}
\usepackage{tikz}
\usepackage{hyperref}
\usepackage{xspace}
\usepackage{tikz}
\usetikzlibrary{calc,trees,positioning,arrows,chains,shapes.geometric,%
    decorations.pathreplacing,decorations.pathmorphing,shapes,%
    matrix,shapes.symbols,patterns}

\usetikzlibrary{decorations.pathmorphing}
\usetikzlibrary{shapes.misc} 
\usetikzlibrary{backgrounds}
\usepackage{transparent}

\newtheorem{theorem}{Theorem}
\newtheorem{definition}{Definition}
\newtheorem{lemma}{Lemma}
\newtheorem{corollary}{Corollary}

\newtheorem{observation}{Observation}

\newcommand{\termasm}[1]{\mathcal{A}_{\Box}[{#1}]}
\newcommand{\pathassembly}[1]{\mathrm{asm}{(#1)}}
\newcommand{\asm}[1]{\pathassembly{#1}}

\newcommand{\draws}[2]{\draw[fill=black] (#1.16,#2.16) rectangle (#1.84,#2.84);}
\newcommand{\drawt}[2]{\draw[fill=white] (#1.16,#2.16) rectangle (#1.84,#2.84);}
\newcommand{\drawlg}[2]{\draw[fill=green!50!white] (#1.16,#2.16) rectangle (#1.84,#2.84);}
\newcommand{\drawg}[2]{\draw[fill=green!50!black] (#1.16,#2.16) rectangle (#1.84,#2.84);}
\newcommand{\drawgre}[2]{\draw[fill=green!75!black] (#1.16,#2.16) rectangle (#1.84,#2.84);}
\newcommand{\drawgreen}[2]{\draw[fill=green] (#1.16,#2.16) rectangle (#1.84,#2.84);}
\newcommand{\drawo}[2]{\draw[fill=orange!50!white] (#1.16,#2.16) rectangle (#1.84,#2.84);}
\newcommand{\drawp}[2]{\draw[fill=purple!50!white] (#1.16,#2.16) rectangle (#1.84,#2.84);}
\newcommand{\drawb}[2]{\draw[fill=blue!50!white] (#1.16,#2.16) rectangle (#1.84,#2.84);}
\newcommand{\drawblue}[2]{\draw[fill=blue] (#1.16,#2.16) rectangle (#1.84,#2.84);}
\newcommand{\drawr}[2]{\draw[fill=red!50!white] (#1.16,#2.16) rectangle (#1.84,#2.84);}
\newcommand{\drawlr}[2]{\draw[fill=red!25!white] (#1.16,#2.16) rectangle (#1.84,#2.84);}
\newcommand{\drawre}[2]{\draw[fill=red!75!black] (#1.16,#2.16) rectangle (#1.84,#2.84);}
\newcommand{\drawred}[2]{\draw[fill=red!90!white] (#1.16,#2.16) rectangle (#1.84,#2.84);}
\newcommand{\drawgray}[2]{\draw[fill=black!50!white] (#1.16,#2.16) rectangle (#1.84,#2.84);}
\newcommand{\drawdgray}[2]{\draw[fill=black!75!white] (#1.16,#2.16) rectangle (#1.84,#2.84);}
\newcommand{\drawwhite}[2]{\fill [fill=white] (#1.16,#2.16) rectangle (#1.84,#2.84);}

\newcommand{\drawyellow}[3]{\draw[fill=yellow!#3!black] (#1.16,#2.16) rectangle (#1.84,#2.84);}
\newcommand{\drawpurple}[3]{\draw[fill=purple!#3!white] (#1.16,#2.16) rectangle (#1.84,#2.84);}

\newcommand{\Z}{\mathbb{Z}}
\newcommand{\N}{\mathbb{N}}
\newcommand{\R}{\ensuremath{\mathbb{R}}}
\newcommand{\resp}{respectively\xspace}
\newcommand{\vect}{\protect\overrightarrow}
\newcommand{\prodasm}[1]{\mathcal{A}[{#1}]}
\newcommand{\calT}{\mathcal{T}}
\newcommand{\prodT}{\prodasm{\mathcal{T}}}
\newcommand{\prodpaths}[1]{{\bf{P}}[{#1}]}
\newcommand{\prodpathsT}{\prodpaths{\mathcal{T}}}
\newcommand{\pos}[1]{\mathrm{pos}(#1)}
\newcommand\type[1]{\mathrm{type}(#1)}
\newcommand{\reverse}[1]{\ensuremath{{#1}^\leftarrow}}
\newcommand{\xcoord}[1]{\mathrm{x}_{#1} }
\newcommand{\ycoord}[1]{\mathrm{y}_{#1} }
\newcommand\range[3]{{#1},{#2},\ldots,{#3}}
\newcommand{\gs}[2]{\ensuremath{\Big[#1,#2\Big]}}

\newcommand{\olq}{{\ensuremath{\overline{q}}}}
\newcommand\torture[4]{{#1}_{{#3}+1 + (({#2}-{#3}-1)\mod({#4}-{#3}))} + \left\lfloor \frac{{#2}-{#3}-1}{{#4}-{#3}}\right\rfloor\vect{{#1}_{#3}{#1}_{#4}}}
\newcommand\embed[1]{{\ensuremath{\frak{E}[#1] }}}
\newcommand{\concat}[1]{\mathrm{concat}\!\left(#1\right)}
 
\newcommand{\uniterm}{\alpha}
\newcommand{\inuniterm}{\prodpaths{\mathcal{\uniterm}}}
\newcommand\bound{\ensuremath{(8|T|)^{4|T|+1}(5|\sigma|+6)}\xspace}
\newcommand\ass{\beta}

\newcommand{\pump}[1]{(#1)^*}
\newcommand{\bipump}[1]{{}^*(#1)^*}

\newcommand{\leftside}[1]{L(#1)}
\newcommand{\rightside}[1]{R(#1)}

\newcommand{\domain}[1]{dom(#1)}
\newcommand{\funpump}[2]{f(#1,#2)}
\newcommand{\height}{\mathrm{h}}

\title{On the directed tile assembly systems at temperature $1$.\thanks{Supported by European Research Council (ERC) award number 772766 and Science foundation Ireland (SFI) grant 18/ERCS/5746 (this manuscript reflects only the authors' view and the ERC is not responsible for any use that may be made of the information it contains).}} % Some of this work was supported by, and carried out at, Inria, Paris, France.}}

\author{Pierre-\'Etienne Meunier\thanks{Coturnix. \href{pe@coturnix.fr}{pe@coturnix.fr}}\\
  \and
Damien Regnault\thanks{IBISC, Université Évry, Université Paris-Saclay, 91025, Evry, France. \href{mailto:damien.regnault@univ-evry.fr}{damien.regnault@univ-evry.fr}}\\
}
\date{}

\begin{document}
\numberwithin{figure}{section} 
\maketitle

%TODO
%figure partie sans redondance (pour plus tard)
%figure argument MT1
%figure argument décalage
%figure M2 (pour plus tard)

\section{Introduction}

%présentation de la conjecture
We show here that a model called directed self-assembly at temperature 1 is unable to do complex computations like the ones of a Turing machine. Since this model can be seen as a generalization of finite automata to $2D$ languages, a logical approach is to proceed in two steps. The first one is to develop a $2D$ pumping lemma and the second one is to use this pumping lemma to classify the different types of possible computation.

%présentation du modèle TROP DETAILLE?
In the directed tile assembly system at temperature $1$, tiles are positioned on a 2D discrete grid. The tiles are squares with some glues on their four sides. Two tiles can bind together if they are neighbors and have matching glues on their abuttal sides. There is only a finite number of tile types. Several tiles bounded together are called an assembly. A sequence of tiles such that each tile binds to the previous one is called a path. We consider initially a specific assembly called the seed and we study how this assembly grows when some new tiles bind on it. An assembly is called terminal when no more tile can bind to it. In the directed case considered here, a given set of tile types and a seed will always produce the same terminal assembly. 

%ETAT DE L'ART
In \cite{STOC2019}, the authors demonstrate a pumping lemma: for any path growing from the seed which is lengthy enough, some part of it can be copy-pasted in order to grow an infinite periodic path (note that this lemma states that the path could also be fragile which can occurs only in the non-directed case). The bound proven for this pumpable lemma is a tower of exponential according to the number of tile types and the size of the seed (a previous version of this result is still available on arXiv \cite{pumpabilityLargeBound}). Beforehand, Doty et al \cite{Doty-2011}, assuming the existence of a pumping lemma, classify the different types of terminal assembly. 
Thus the combination of these two papers solves the directed temperature 1 conjecture $\ldots$ but in an imperfect way. Indeed, since \cite{Doty-2011} is anterior to \cite{STOC2019}, the authors assumed a different and stronger pumping lemma where there is no mention of the seed: any path of a given length in the terminal assembly is pumpable. Nevertheless, all the demonstrations made in \cite{Doty-2011} still hold with the pumping lemma of \cite{STOC2019}.

%pourquoi la classification de doty fonctionne encore -> casser les peignes
The classification of Doty et al \cite{Doty-2011} leads to four kinds of terminal assembly: finite, infinite with/without comb, periodic with/without comb and bi-periodic. A comb is a periodic path growing on another periodic path. Combs are the most complex structure of the classification and cannot be dealt directly with the pumping lemma of \cite{STOC2019} since they seem to grow arbitrarily far away from the seed. Nevertheless, since a comb grows on a periodic path, there exists an infinity of copies of the comb growing on the same periodic path. Then, pumping one comb is enough to deal with all of its copies. Moreover, the authors of \cite{Doty-2011} have shown that these copies cannot intersect. Finally, if a comb grows far away enough from the seed, it will have to intersect one of its copy before intersecting with the seed which is not possible. Hence, the pumping lemma of \cite{STOC2019} can be used on a comb close to the seed. 

%contribution
In this paper, we harmonize the notations between \cite{Doty-2011} and \cite{STOC2019} in order to clearly solve the directed temperature 1 conjecture, the lemma \ref{lem:magicindice} explains in details how to break combs with the reasoning of the previous paragraph. We are also able to show that the bi-periodic structures cannot reuse a tile type and thus we give an optimal description of these structures.  For a more detailed presentation of this model and the implications of this result, see the introduction of \cite{STOC2019}.

\section{Definitions (from \cite{STOC2019}) and main theorems}
\label{definitions}

%repris du papier précédent 
As usual, let $\mathbb{R}$ be the set of real numbers, let $\mathbb{Z}$ be the integers, and let $\mathbb{N}$ be the natural numbers including 0.  
The domain of a function $f$ is denoted $\domain{f}$, and its range (or image) is denoted $f(\domain{f})$.

%rajouter main theorem

\subsection{Abstract tile assembly model}\label{sec:atam}

The abstract tile assembly model was introduced by Winfree~\cite{Winf98}. In this paper we study a restriction of the abstract tile assembly model called the directed temperature~1 abstract tile assembly model, or noncooperative abstract tile assembly model. For definitions of the full model, as well as intuitive explanations, see for example~\cite{RotWin00,Roth01}.

A \emph{tile type} is a unit square with four sides,
each consisting of a glue \emph{type} and a nonnegative integer \emph{strength}. Let  $T$  be a a finite set of tile types.
%In any set of tile types used in this paper, we assume the existence of a well-defined total ordering which we call the {\em canonical ordering} of the tile set. 
The sides of a tile type are respectively called  north, east, south, and west, as shown  in the following picture:
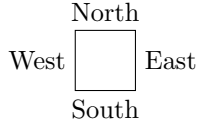
\begin{center}
\vspace{-1ex}
\begin{tikzpicture}[scale=0.8]
\draw(0,0)rectangle(1,1);
\draw(0,0.5)node[anchor=east]{West};
\draw(1,0.5)node[anchor=west]{East};
\draw(0.5,0)node[anchor=north]{South};
\draw(0.5,1)node[anchor=south]{North};
\end{tikzpicture}
\vspace{-1ex}\end{center}

An \emph{assembly} is a partial function $\alpha:\mathbb{Z}^2\dashrightarrow T$ where $T$ is a set of tile types and the domain of $\alpha$ (denoted $\domain{\alpha}$) is connected.\footnote{Intuitively, an assembly is a positioning of unit-sized tiles, each from some set of tile types $T$, so that their centers are placed on (some of) the elements of the discrete plane $\mathbb{Z}^2$ and such that those elements of $\mathbb{Z}^2$ form a connected set of points.} 
We let $\mathcal{A}^T$ denote the set of all assemblies over the set of tile types $T$. 
In this paper, two tile types in an assembly are said to  {\em bind} (or \emph{interact}, or are
\emph{stably attached}), if the glue types on their abutting sides are
equal, and have strength $\geq 1$.  An assembly $\alpha$ induces an undirected
weighted \emph{binding graph} $G_\alpha=(V,E)$, where $V=\domain{\alpha}$, and
there is an edge $\{ a,b \} \in E$ if and only if the tiles at positions $a$ and $b$ interact, and
this edge is weighted by the glue strength of that interaction.  The
assembly is said to be $\tau$-stable if every cut of $G_{\alpha}$ has weight at
least $\tau$.

A \emph{tile assembly system} is a triple $\mathcal{T}=(T,\sigma,\tau)$,
where $T$ is a finite set of tile types, $\sigma$ is a $\tau$-stable assembly called the \emph{seed}, and
$\tau \in \mathbb{N}$ is the \emph{temperature}.
Throughout this paper,  $\tau=1$.

Given two $\tau$-stable assemblies $\alpha$ and $\beta$, we say that $\alpha$ is a
\emph{subassembly} of $\beta$, and write $\alpha\sqsubseteq\beta$, if
$\domain{\alpha}\subseteq \domain{\beta}$ and for all $p\in \domain{\alpha}$,
$\alpha(p)=\beta(p)$.
We also write
$\alpha\rightarrow_1^{\mathcal{T}}\beta$ if we can obtain $\beta$ from
$\alpha$ by the binding of a single tile type, that is:  $\alpha\sqsubseteq \beta$, $|\domain{\beta}\setminus\domain{\alpha}|=1$ and the tile type at the position $\domain{\beta}\setminus\domain{\alpha}$ stably binds to $\alpha$ at that position.  We say that~$\gamma$ is
\emph{producible} from $\alpha$, and write
$\alpha\rightarrow^{\mathcal{T}}\gamma$ if there is a (possibly empty)
sequence $\alpha_1,\alpha_2,\ldots,\alpha_n$ where $n \in \N \cup \{ \infty \} $, $\alpha= \alpha_1$ and $\alpha_n =\gamma$, such that
$\alpha_1\rightarrow_1^{\mathcal{T}}\alpha_2\rightarrow_1^{\mathcal{T}}\ldots\rightarrow_1^{\mathcal{T}}\alpha_n$. 
A sequence of $n\in\mathbb{Z}^+ \cup \{\infty\}$ assemblies
$\alpha_0,\alpha_1,\ldots$ over $\mathcal{A}^T$ is a
\emph{$\mathcal{T}$-assembly sequence} if, for all $1 \leq i < n$,
$\alpha_{i-1} \to_1^\mathcal{T} \alpha_{i}$.

%rajout def unon
Given two $\tau$-stable assemblies $\alpha$ and $\beta$, the union of $\alpha$ and $\beta$, write $\alpha\cup\beta$,  is an assembly defined if and only if and for all $p\in \domain{\alpha}\cap\domain{\beta}$, $\alpha(p)=\beta(p)$ and either at least one tile of $\alpha$ binds with a tile of $\beta$ or $\domain{\alpha}\cap\domain{\beta}\neq \emptyset$. Then, for all $p\in \domain{\alpha}$, we have $(\alpha\cup \beta)(p)=\alpha(p)$ and for all $p\in \domain{\beta}$, we have $(\alpha \cup \beta)(p)=\beta(p)$.

The set of \emph{productions}, or \emph{producible assemblies}, of a tile assembly system $\mathcal{T}=(T,\sigma,\tau)$ is the set of all assemblies producible
from the seed assembly $\sigma$ and is written~$\prodasm{\mathcal{T}}$. An assembly $\alpha$ is called \emph{terminal} if there is no $\beta$ such that $\alpha\rightarrow_1^{\mathcal{T}}\beta$. The set of all terminal assemblies of $\mathcal{T}$ is denoted~$\termasm{\mathcal{T}}$. 
%rajout directed et translation
If there is a unique terminal assembly, \emph{i.e.} $|\termasm{\mathcal{T}}|=1$, then $\mathcal{T}$ is \emph{directed}. Along the article, this unique terminal assembly is denoted $\uniterm$. Note that we do not assume that the seed could be reduced to a single tile as explained in details in appendix \ref{app:seed}.

%We introduce the definition needed to describe the different kinds of terminal assembly and state the main theorems. 
The translation of an assembly $\alpha$ by a vector $\vect{v}$, written $\alpha+\vect{v}$, is the assembly $\beta$ defined for all $(x,y)\in(\domain\alpha+\vect{v})$ as $\beta(x,y)=\alpha((x,y)-\vect{v})$. An assembly $\alpha$ is \emph{$\vect{v}$-periodic} if and only if $\vect{v}$ is not the null vector $\vect{0}$ and $\alpha+\vect{v}=\alpha$. The periodicity of the terminal assembly determine its complexity: the three main kinds of terminal assembly can be identified according to this parameter. An assembly $\alpha$ is \emph{bi-periodic} if there exist two non-collinear vectors $\vect{u}$ and $\vect{v}$ such that $\alpha$ is $\vect{u}$-periodic and $\vect{v}$-periodic. An assembly is \emph{simply periodic} if it is not bi-periodic and if there exists a vector $\vect{v}$ such that $\alpha$ is $\vect{v}$-periodic. Otherwise, an assembly is aperiodic. We define formally the complexity of an assembly and state the three main theorems.

\begin{definition}
The complexity of a finite assembly is $0$. For $i\geq 1$, the complexity of an assembly $\alpha$ is $i$ if 
$\alpha=\bigcup_{\ell \in \N} (\beta+\ell\vect{v})$ where the complexity of $\beta$ is $i-1$,
or if it is a finite union of assembly of level less than $i$.
\end{definition}

\begin{theorem}[Description of the bi-periodic terminal assemblies]
\label{main:biperiodic}
Consider a directed tile assembly system $\mathcal{T}=(T,\sigma,1)$ whose terminal assembly $\uniterm$ is bi-periodic. Then there exist two non-collinear vectors $\vect{u}$ and $\vect{v}$ and an assembly $\ass$ of complexity $0$ whose size is bounded by $|T|^2$ such that $$\uniterm=\bigcup_{\ell,\ell' \in \Z} (\ass+\ell\vect{u}+\ell'\vect{v}).$$
\end{theorem}

\begin{theorem}[Description of the simply periodic terminal assemblies]
\label{main:periodic}
Consider a directed tile assembly system $\mathcal{T}=(T,\sigma,1)$ whose terminal assembly $\uniterm$ is simply periodic. Then there exists a vector $\vect{v}$ and an assembly $\ass$ of complexity $1$ such that $$\uniterm=\bigcup_{\ell\in \Z} (\ass+\ell\vect{v}).$$
\end{theorem}

\begin{theorem}[Description of the aperiodic terminal assemblies]
\label{main:aperiodic}
Consider a directed tile assembly system $\mathcal{T}=(T,\sigma,1)$ whose terminal assembly $\uniterm$ is aperiodic, then the complexity of $\uniterm$ is $2$.
\end{theorem}

An assembly of complexity $0$ is finite while assemblies of higher complexity are infinite. According to these three theorems we only need to study assemblies of complexity less than $2$. Such assemblies could be represented by semi-linear sets. This technique was introduced in \cite{Doty-2011} and the Observation $4.2$ in the arXiv version of \cite{Doty-2011} argues that such assemblies are not able of complex computations.

\begin{definition}[linear and semilinear sets]
A set $X \subset \Z^2$ is \emph{linear} if there exists $p\in \Z^2$ and two vectors $\vect{u}$ and $\vect{v}$ such that $$X=\bigcup_{\ell,\ell' \in \N}\left\{p+\ell\vect{u}+\ell'\vect{v}\right\}.$$
A semilinear set is a finite union of linear sets.
\end{definition}

\begin{observation}
The domain of an assembly of complexity less than $2$ is a semilinear set.
\end{observation}

\begin{proof}
An assembly $\alpha^{(0)}$ of complexity $0$ is a finite assembly, then its domain can be described by a finite union of $|\domain{\alpha^{(0)}}|$ linear sets where $\vect{u}$ and $\vect{v}$ are both the null vector $\vect{0}$. If an assembly $\alpha^{(1)}$ is defined by as $\bigcup_{\ell \in \N} (\alpha^{(0)}+\ell\vect{u'})$ where the complexity of $\alpha^{(0)}$ is $0$, then for any linear set in the description of $\alpha^{(0)}$ we can set $\vect{u}=\vect{u'}$ instead of the null vector. The size of the semilinear set used to described $\domain{\alpha^{(1)}}$ is the same as the one used to describe $\domain{\alpha^{(0)}}$. If an assembly of complexity $1$ is the finite union of assembly of complexity less than $1$ then the domain of such an assembly is the finite union of the semilinear sets used to describe the domains of these assemblies. Similarly, If an assembly $\alpha^{(2)}$ is defined as $\bigcup_{\ell \in \N} (\alpha^{(1)}+\ell\vect{v'})$ where the complexity of $\alpha^{(1)}$ is $1$, then for any linear set in the description of $\alpha^{(1)}$ we can set $\vect{v}=\vect{v'}$ instead of the null vector. The size of the semilinear set used to described $\domain{\alpha^{(2)}}$ is the same as the one used to describe $\domain{\alpha^{(1)}}$. If an assembly of complexity $2$ is the finite union of assemblies of complexity less than $2$, then the domain of such an assembly is the finite union of the semilinear sets used to describe the domain of these assemblies.

%A $2$-level assembly $A^2$ is defined by a couple $(A^1,\vect{v'})$ where the complexity of $A^1$ is $1$, then for any linear set in the description of $A^1$ we can set $\vect{v}=\vect{v'}$ instead of the null vector. The size to describe $A^2$ is the same as the one of $A^1$. An assembly of complexity $2$ is the finite union of $0$-level assemblies, $1$-level assemblies and $2$-level assemblies, the description of such an assembly is the finite union of the linear sets used to describe these assemblies. 
\end{proof}

\subsection{Paths}\label{sec:defs-paths}

This section introduces quite a number of key definitions and concepts that will be used extensively throughout the paper. 

Let $T$ be a set of tile types. 
A {\em tile} is a pair $((x,y),t)$ where $(x,y) \in \mathbb{Z}^2$ is a position and $t\in T$ is a tile type. 
Intuitively, a path is a finite or one-way-infinite simple (non-self-intersecting) sequence of tiles placed on points of $\mathbb{Z}^2$ so that each tile in the sequence interacts with the previous one, or more precisely: 

\begin{definition}[Path]\label{def:path}
  A {\em path} is a (finite or infinite) sequence  $P = P_0 P_1 P_2 \ldots$  of tiles   $P_i = ((x_i,y_i),t_i) \in \mathbb{Z}^2 \times T$, such that:
\begin{itemize}
\item for all $P_j$ and $P_{j+1}$ defined on $P$ it is the case that~$t_{j}$ and~$t_{j+1}$ interact, and
\item for all $P_j,P_k$ such that $j\neq k$ it is the case that $ (x_j,y_j) \neq (x_k,y_k)$.
\end{itemize}
\end{definition}

By definition, paths are simple (or self-avoiding), and this fact will be repeatedly used throughout the paper.
For a tile $P_i$ on some path $P$, its x-coordinate is denoted~$\xcoord{P_i}$ and its y-coordinate is denoted~$\ycoord{P_i}$. 
The \emph{concatenation} of two paths $P$ and $Q$ is the concatenation $PQ$ of these two paths as sequences, and is a path if and only if (1) the last tile of $P$ interacts with the first tile of $Q$ and (2)  $P$ and $Q$ do not intersect each other.

For a path $P = P_0 \ldots  P_i P_{i+1} \ldots P_j  \ldots $, we define the notation $P_{i,i+1,\ldots,j} = P_i P_{i+1} \ldots P_j$, i.e.\  ``the subpath of $P$ between indices $i$ and $j$, inclusive''.
Whenever $P$ is finite, i.e. $P = P_0P_1P_2\ldots P_{n-1}$ for some $n\in\mathbb{N}$, $n$ is termed the {\em length} of $P$ and denoted by $|P|$. In the special case of a subpath where $i=0$, we say that $P_{0,1,\ldots,j}$ is a prefix of $P$ and when $j=|P|-1$, we say that $P_{i,\ldots, |P|-1}$ is a suffix of $P$.
For any path $P = P_0 P_1 P_2, \ldots$ and integer $i\geq 0$, we write $\pos{P_i} \in \mathbb{Z}^2$, or  $(x_{P_i},y_{P_i}) \in \mathbb{Z}^2$, for the position of $P_i$ and $\type{P_i}$ for the tile type of $P_i$. Hence if  $P_i = ((x_i,y_i),t_i) $ then $\pos{P_i} =  (\xcoord{P_i},\ycoord{P_i}) = (x_i,y_i) $ and $\type{P_i} = t_i$.
A ``\emph{position of}  $P$'' is an element of $\mathbb{Z}^2$ that appears in $P$ (and therefore appears exactly once), and an \emph{index} $i$ of a path $P$ of length $n\in \mathbb{N}$ is a natural number  $i \in \{0,1,\ldots,n-1\}$.
For a path $P=P_0P_1P_2\ldots$ we write $\pos{P}$ to mean ``the sequence of positions of tiles along $P$'', which is $\pos{P}=\pos{P_0}\pos{P_1}\pos{P_2}\ldots\;$. For a finite path $P=P_0P_1P_2\ldots P_{|P|-1}$, we define $\reverse{P}$ as the path $P_{|P|-1}P_{|P|-2}\ldots P_0$. The vertical height of a path $P$ is defined as $\max\{|y_{P_i}-y_{P_j}|: 0\leq i \leq j \leq |P|-1\}$ and its horizontal width is $\max\{|x_{P_i}-x_{P_j}|: 0\leq i \leq j \leq |P|-1\}$.

Although a path is not an assembly, we know that each adjacent pair of tiles in the path sequence interact implying that the set of path positions forms a connected set in $\Z^2$ and hence every path uniquely represents an assembly containing exactly the tiles of the path, more formally:
for a path $P = P_0 P_1 P_2 \ldots$ we define the set of tiles  $\pathassembly{P} = \{ P_0, P_1, P_2, \ldots\}$ which we observe is  an assembly\footnote{I.e.  $\pathassembly{P}$ is  a partial function from $\Z^2$ to tile types, and is defined on a connected set.} and  we call $\pathassembly{P}$ a {\em path assembly}.
%A GARDER 
Given a tile assembly system $\calT = (T,\sigma,1)$ the path $P$ is a {\em  producible path of $\calT$} if 
$\asm{P}$ does not intersect\footnote{Formally, non-intersection of a path $P = P_0 P_1, \ldots $ and a seed assembly $\sigma$ is defined as: $\forall t$ such that  $t \in \sigma$, $\nexists i $ such that $\pos{P_i} = \pos{t}$.} the seed $\sigma$
and the assembly $(\asm{P} \cup \sigma ) $ is producible by $\calT$, i.e.\ $(\asm{P} \cup \sigma ) \in \prodT$, and $P_0$ interacts with a tile of $\sigma$. Consider an assembly $\alpha$ (resp. a path $Q$), as a convenient abuse of notation we sometimes write $\sigma \cup P$ (resp. $P \cup Q$)  as a shorthand for $\sigma \cup \asm{P}$  (resp. $\asm{P} \cup \asm{Q}$).
Given  a tile assembly system $\calT = (T,\sigma,1)$
we define  the set of producible paths of $\calT$ to be:\footnote{Intuitively, although producible paths are not assemblies, any  producible path $P$ has the nice property that it encodes an unambiguous description of how to grow $\asm{P}$ from the seed $\sigma$, in path  ($P$) order, to produce  the assembly $  \asm{P}\cup \sigma$.}
$$\prodpathsT = \{ P  \mid P \textrm{ is a path that does not intersect } \sigma \textrm{ and } (\asm{P}\cup\sigma) \in \prodT \} $$
%NEW version
Given a directed tile assembly system $\calT = (T,\sigma,1)$ and its unique terminal assembly $\uniterm$, the path $P$ is {\em  a path of $\uniterm$} if 
$\asm{P}$ is a subassembly of $\uniterm$.
%Given  a tile assembly system $\calT = (T,\sigma,1)$
We define  the set of paths of $\uniterm$ to be:
$$\inuniterm = \{ P  \mid P \textrm{ is a path and $\asm{P}$ is a subassembly of } \uniterm \} $$

So far, we have defined paths of tiles (Definition~\ref{def:path}). In our proofs, we will also reason about (untiled) {\em binding paths} in the binding graph of an assembly.

\begin{definition}[Binding path]\label{def:binding graph}
Let $G=(V,E)$ be a binding graph. 
A {\em binding path} $q$ in $G$ is a sequence $q_{\range{0}{1}{|q|-1}}$ of vertices from $V$ such that 
for all $i \in \{\range{0}{1}{|q|-2} \}$,
$\{ q_i, q_{i+1} \} \in E$ ($q$~is connected) and no vertex appears twice in $q$ ($q$ is simple).
\end{definition}

The following observation was proven in \cite{STOC2019}.
\begin{observation}
Let $\calT = (T,\sigma,1)$ be a tile assembly system and let $\alpha \in \prodT$.
For any tile $((x,y),t) \in \alpha$ either $((x,y),t)$ is a tile of $\sigma$ or else there is a finite producible path $P \in \prodpathsT$ such that for some $j \in \N$ contains $P_j = ((x,y),t)$.
\end{observation}

\begin{observation}
Let $\calT = (T,\sigma,1)$ be a directed tile assembly system whose terminal assembly is $\uniterm$.
For any tile $A \in \alpha$ there exists a finite producible assembly $\ass$ such that $A$ is a tile of $\ass$.
\end{observation}

\begin{proof}
By the previous observation, if $A$ is a tile of $\sigma$ then $\ass=\sigma$ satisfies the observation. Otherwise, there exists a finite producible path $P \in \prodpathsT$ such that for some $j \in \N$ contains $P_j = A$. Then $P \cup \sigma$ satisfies the observation.
\end{proof}

%MODIFIER
\begin{observation}
Let $\calT = (T,\sigma,1)$ be a directed tile assembly system whose terminal assembly is $\uniterm$.
For any tiles $((x,y),t) \in \alpha$ and $((x',y'),t') \in \alpha$  there is a path $P \in \inuniterm$  such that for some $P_0 = ((x,y),t)$ and $P_{|P|-1} = ((x',y'),t')$.
%Let $\calT = (T,\sigma,1)$ be a tile assembly system and let $\alpha \in \prodT$.
%For any tile $((x,y),t) \in \alpha$ either $((x,y),t)$ is a tile of $\sigma$ or else there is a producible path $P \in \prodpathsT$  that for some $j \in \N$ contains $P_j = ((x,y),t)$.
\end{observation}
\begin{proof}
Since $\domain\alpha$ is a connected subset of $\Z^2$  there is an integer $n\geq 0$ and a binding path $p_{0,1,\ldots ,n}$ in the binding graph of $\uniterm$ where $p_0=(x,y)$ and $p_n = (x',y')$. We can then define $P $ as the path:
  $$P = P_{0,1,\ldots,n} = (p_{0},\alpha(p_{0}))(p_{1},\uniterm(p_1))\ldots(p_n,\uniterm(p_{n}))$$
  By definition of binding graph, for all $i \in \{0,i,\ldots,n-1 \}$, the tiles 
  $(p_{i},\alpha(p_{i}))$ and $(p_{i+1},\alpha(p_{i+1}))$ on $P$ are adjacent in $\mathbb{Z}^2$ and interact on their abutting sides, meaning that $P\in\inuniterm$, thus proving the statement. 
  \end{proof}

For $A,B\in\mathbb{Z}^2$, we define $\vect{AB} = B - A$ to be the vector from $A$ to $B$, and  for two \emph{tiles} $P_i = ((x_i,y_i),t_i)$ and $P_j = ((x_j,y_j),t_j)$ we define $\vect{P_i P_j} = \pos{P_j} - \pos{P_i}$ to mean the vector from $\pos{P_i}=(x_i,y_i)$ to $\pos{P_j}=(x_j,y_j)$.
The translation of a path $P$ by a vector $\vect{v} \in \mathbb{Z}^2$, written $P+\vect{v}$, is  the path $Q$ such that $|P|=|Q|$
and for all indices $i \in \{0,1,\ldots,|P|-1 \}$,
$\pos{Q_i}=\pos{P_i}+\vect{v}$ 
and 
$\type{Q_i}=\type{P_i}$. 
%We always use parentheses for scoping of translations when necessary, i.e.\ $P(Q+\vect{v})$ is the sequence containing the path $P$ followed by the translation of the entire path $Q$ by vector $\vect{v}$. 

Let $P$ be a path, let $i\in\{1,2,\ldots,|P|-2\}$, and let $A\neq P_{i+1}$ be a tile such that $P_{0,1,\ldots,i}A$ is a path.
Let also $\rho$ be the clockwise rotation matrix defined as $\rho= \bigl(\begin{smallmatrix}0&1 \\ -1&0\end{smallmatrix} \bigr)$, and let
$\tau = (\rho\vect{P_iP_{i-1}}, \rho^2\vect{P_iP_{i-1}}, \rho^3\vect{P_iP_{i-1}})$ (intuitively, $\tau$ is the vector of possible steps after $P_i$, ordered clockwise).
We say that $P_{0,1,\ldots,i}A$ \emph{turns right} (\resp \emph{turns left}) from $P_{0,1,\ldots,i+1}$ if $\vect{P_{i}A}$ appears after (\resp before) $\vect{P_{i} P_{i+1}}$ in $\tau$.

\subsection{Intersections}

If two paths, or two assemblies, or a path and an assembly, share a common position we say that they {\em intersect} at that position. Furthermore, we say that two paths, or two assemblies, or a path and an assembly,  {\em agree} on a position if they both place the same tile type at that position and {\em conflict} if they place a different tile type at that position.
We say that a path $P$ is {\em fragile} to mean that there is a producible assembly $\alpha$ that conflicts with $P$. Intuitively, if we grow $\alpha$ first, then there is at least one tile that $P$ cannot place:

\begin{definition}[Fragile]\label{def:fragile}
  Let $\calT = (T,\sigma,1)$ be a tile assembly system and $P\in\prodpathsT$. We say that $P$ is fragile if there exists a producible assembly $\alpha \in \prodasm{\calT}$ and a position $(x,y)\in(\domain\alpha\cap\domain{\asm P})$
  such that $\alpha((x,y))\neq\asm{P}((x,y))$.\footnote{
    Here, it might be the case that $\alpha$ and $P$ conflict at only one position by placing two different tile types $t$ and $t'$, but that $t$ and $t'$ may place the same glues along $P$. In this case $P$ is not producible when starting from the assembly  $\alpha$ because one of the tiles along the positions of $P$ is of the wrong type.
  }
\end{definition}

% NOUVEAU

In a directed tile assembly system, there exists not fragile path, thus if a path conflicts with a producible assembly then this path is not a path of the unique terminal assembly.

\begin{observation}
Let $\calT = (T,\sigma,1)$ be a directed tile assembly system whose terminal assembly is $\uniterm$, a producible assembly $\ass$ and a path $P$ which intersects with $\ass$. If all intersections are agreements then $P\in \inuniterm$, otherwise $P\not \in \inuniterm$.
\end{observation}

Consider two paths $P$ and $Q$, we say that $Q$ \emph{grows} on $P$ at index $i$, if the only intersection between $Q$ and $P$ occurs at $\pos{Q_0}=\pos{P_i}$ and is an agreement. Also, $Q$ is an \emph{arc} of $P$ between indices $i \neq j$ if and only if there are exactly two intersections between $Q$ and $P$ which occurs at $\pos{Q_0}=\pos{P_i}$ and $\pos{Q_{|Q|-1}}=\pos{P_j}$ and both are agreement, moreover for the special case where $|Q|=2$, we should also have $j\neq i+1$ and $j\neq i-1$. The width of an arc $Q$ of $P$ is defined by $|j-i|$. Note that, if $P\in \inuniterm$, then the path $Q$ or the arc $A$ do not necessarily belong to $\inuniterm$ since they can conflict with the seed. 

\subsection{Pumping a path}\label{sec:Pumping a path}

Next, for a path $P$, we define sequences of points and tile types (not necessarily a path) called the \emph{pumping of $P$} or the \emph{bi-pumping of $P$}:
\begin{definition}[Pumpings of $P$]
  \label{def:pumpingP}
Let $\calT = (T,\sigma,1)$ be a tile assembly system and a path $P$ such that $\type{P_0}=\type{P_{|P|-1}}$.
We say that the \emph{``pumping of $P$''}, denoted by $\pump{P}$, is the infinite sequence $\olq$ of elements from $\Z^2\times T$ defined by:

$$\olq_k = P_{k \mod (|P|-1)} + \left\lfloor \frac{k}{|P|-1}  \right\rfloor \vect{P_0P_{|P|-1}} \textrm{ for } k \in \N.$$

Whereas, we say that the \emph{``bi-pumping of $P$''}, denoted by $\bipump{P}$, is the bi-infinite sequence $\olq$ of elements from $\Z^2\times T$ defined by:

$$\olq_k = P_{k \mod (|P|-1)} + \left\lfloor \frac{k}{|P|-1}  \right\rfloor \vect{P_0P_{|P|-1}} \textrm{ for } k \in \Z.$$
\end{definition}

We will always consider cases where $\olq$ is self-avoiding and that in particular, for any $s< t $, if the path $P+s\vect{P_0P_{|P|-1}}$ intersects with the path $P+t\vect{P_iP_j}$, then $t=s+1$ and the only intersection is an agreement between $P_0+t\vect{P_0P_{|P|-1}}$ and $P_{|P|-1}+s\vect{P_0P_{|P|-1}}$. A sufficient condition for this is that the only intersection between $P$ and $P+\vect{P_0P_{|P|-1}}$ is an agreement between $P_0+\vect{P_0P_{|P|-1}}$ and $P_{|P|-1}$ (see Lemma~\ref{lem:precious}). If this condition is satisfied then $P$ is called a \emph{good candidate} and $\pump{P}$ and $\bipump{P}$ are both paths. Remark that, for all $k \in \N$ (resp. $k \in \Z$), we have $\pump{P}_{k+|P|-1}=\pump{P}_{k}+\vect{P_0P_{|P|-1}}$ (resp. $\bipump{P}_{k+|P|-1}=\bipump{P}_{k}+\vect{P_0P_{|P|-1}})$.

\begin{definition}[Pumpable path]\label{def:pumpable path}
Let $\calT = (T,\sigma,1)$ be a directed tile assembly system and let $\uniterm$ be its unique terminal assembly.  We say that a good candidate $P$ is {\em pumpable} if $\pump{P} \in \inuniterm$ and {\em bi-pumpable} if $\bipump{P} \in \inuniterm$. A good candidate which is pumpable but not bi-pumpable is called simply pumpable.
\end{definition}

This definition of pumping does not take into account the position of the seed. In order to use the pumping lemma of \cite{STOC2019}, we need a special definition of pumping for producible paths. For a producible path $P$ and two indices $i,j$ on $P$, we define a sequence of points and tile types (not necessarily a path) called the \emph{pumping of $P$ between $i$ and $j$}:
\begin{definition}[Pumping of a producible path $P$ between $i$ and $j$]
  \label{def:pumpingPbetweeniandj}
Let $\calT = (T,\sigma,1)$ be a tile assembly system and $P\in\prodpathsT$.
We say that the \emph{``pumping of $P$ between $i$ and $j$''} is the sequence $\olq$ of elements from $\Z^2\times T$ defined by:

\begin{equation*}
\olq_k =
\begin{cases}
 P_k &\qquad \textrm{for } 0\leq k \leq i \\
 \torture P k i j &\qquad \textrm{for }  i < k  ,
\end{cases}
\end{equation*}
\end{definition}

Intuitively, $\olq$ is the concatenation of a finite path $P_{0,1,\ldots,i}$ and an infinite periodic sequence of tile types and positions (possibly intersecting $\sigma\cup P_{0,1,\ldots,i}$, and possibly intersecting itself). The following definition gives the notion of pumpable path used in the pumping lemma \ref{theorem:pumping:original} which follows. 

\begin{definition}[Producible pumpable path]\label{def:pumpable path}
  Let $\calT = (T,\sigma,1)$ be a tile assembly system.  We say that a producible path $P \in\prodpathsT$, is {\em infinitely pumpable}, if there are two integers $i<j$ such that the pumping of $P$ between $i$ and $j$ is  an infinite  producible path, i.e.\ formally: $ \olq \in\prodpaths{\mathcal{T}}$.

  In this case, we say that the \emph{pumping vector} of $\olq$ is $\vect{P_iP_j}$, and that $P$ is \emph{pumpable with pumping vector $\vect{P_iP_j}$}.
\end{definition}

The following pumping lemma is the main result of \cite{STOC2019}. 

\begin{theorem}
\label{theorem:pumping:original}
Let $\mathcal T=(T,\sigma,1)$ be any tile assembly system and let $P$ be a path producible by  $\mathcal T$. 
If $P$ has vertical height or horizontal width at least \bound, then $P$ is infinitely pumpable or fragile.
\end{theorem}

In our context, we will use the following version of this theorem where there are no fragile path and  the bound is replaced by a generic function $\funpump{x}{y}$ where $x$ is the number of tile types and $y$ is the size of the seed (since the bound computed in \cite{STOC2019} is probably not optimal and could be improved independently of the results presented here). Note that, unlike in \cite{Doty-2011}, we do not consider that the size of the seed could be reduced $1$ as explained in appendix \ref{app:seed}.

\begin{theorem}
\label{theorem:pumping}
Let $\mathcal T=(T,\sigma,1)$ be any directed tile assembly system and let $P$ be a path producible by  $\mathcal T$. 
If $P$ has vertical height or horizontal width at least $\funpump{|T|}{|\sigma|}$, then $P$ is infinitely pumpable.
\end{theorem}

\subsection{2D plane}

%\subsubsection{Column, glue column, row, glue row}
%When referring to sets of positions, we use the term ``\emph{the column $x$}'' for some fixed $x\in\Z$ to mean the set $\{(x, y)|y\in\Z\}$, and the term ``\emph{the row $y$}'' for some fixed $y\in\Z$ to mean the set $\{(x, y)| x\in\Z\}$. 
%The \emph{glue column $x$}, for some fixed $x\in\Z$, is the set of 2D half-integer positions $\{ (x+0.5,y) \mid y \in \Z \}$.
%The \emph{glue row $y$},       for some fixed $y\in\Z$, is the set of 2D half-integer positions $\{ (x,y+0.5) \mid x \in \Z \}$.
%(Equivalently, by changing the notation but not the underlying idea, the \emph{glue column $x$} is the set of edges of the grid graph of $\Z^2$ between column $x$ and column $x+1$, and the \emph{glue row $y$} is the set of edges of the grid graph of $\Z^2$ between row $x$ and row $x+1$.)

\subsubsection{Curves}
\label{subsubsec:defs:curves}
A curve $c : I \rightarrow \mathbb{R}^2$ is a function from an interval $I \subset \mathbb{R}$ to $\mathbb{R}^2$, where $I$ is one of a closed, open, or half-open. 
All the curves in this paper are polygonal, i.e. unions of line segments and rays.

For a finite path $P$, we call the \emph{embedding $\embed{P}$ of $P$} the curve defined for all $t\in[0,|P|-1]\subset\R$ by:
$$\embed{P}(t) = \pos{P_{\lfloor t\rfloor}} + (t-\lfloor t\rfloor)\vect{P_{\lfloor t\rfloor}P_{\lfloor t\rfloor+1}}$$
Similarly, for a finite binding path $p$, the \emph{embedding $\embed p$ of $p$} is the curve defined for all
$t\in[0,|p|-1]\subset\R$ by:
$$\embed{p}(t) = p_{\lfloor t\rfloor} + (t-\lfloor t\rfloor)\vect{p_{\lfloor t\rfloor}p_{\lfloor t\rfloor+1}}$$

The ray of vector $\vect v$ {\em from} (or, {\em that starts at}) point $A\in\mathbb{R}$ is defined as the curve $r:[0, +\infty[\rightarrow \R^2$ such that $r(t) = A + t\vect v$.
%The \emph{vertical ray from a point $A$ to the south (\resp to the north)} is the ray of vector $(0,-1)$ (\resp $(0,1)$) from $A$, and the \emph{horizontal ray from a point $A$ to the west (\resp to the east)} is the ray of vector $(-1, 0)$ (\resp $(1, 0)$) from~$A$.

If $C$ is a curve defined on some real interval of the form $[a, b]$ or $]a, b]$, and $D$ is a curve defined on some real interval of the form $[c, d]$ or $[c, d[$, and moreover $C(b) = D(c)$, then the \emph{concatenation $\concat{ C,  D}$ of $ C$ and $D$} is the curve defined on $\domain{ C}\cup(\domain{ D} - (c-b))$~by:\footnote{$\domain{ D} - (c-b)$ means $[b, d-(c-b)]$ if $\domain{ D}=[c, d]$, and $[b, d-(c-b)[$ if $\domain{ D}=[c, d[$}
$$\concat{ C,  D}\!(t) = \begin{cases} C(t)\text{ if }t\leq b\\ D(t+(c-b))\text{ otherwise}\end{cases}$$

A curve $c$ is said to be \emph{simple} or \emph{self-avoiding} if all its points are distinct, i.e. if for all $x, y\in\domain c$, $c(x) = c(y)\Rightarrow x=y$.

The reverse $\reverse c$ of a curve $c$ defined on some interval $[a, b]$ (\resp $[a, b[$, $]a, b]$, $]a, b[$\,) is the curve defined on $[-b, -a]$ (\resp $]\!-b,-a]$, $[-b, -a[$, $]\!-b,-a[$\,) as $\reverse c(t) = c(-t)$.

If $A=(x_a,y_a)\in\R^2$ and $B=(x_b,y_b)\in\R^2$, the \emph{segment $\gs A B$} is defined to be the curve $s:[0,1]\rightarrow\R^2$ such that for all $t\in[0,1]$,  $s(t) = ((1-t)x_a+t x_b, (1-t)y_a+t y_b)$. We sometimes abuse notation and write $\gs{A}{B}$ even if $A$ or $B$ (or both) is a \emph{tile}, in which case we mean the position of that tile instead of the tile itself.

For a curve  $c:\mathbb{R} \rightarrow \mathbb{R}^2$ we write $c(\mathbb{R})$ to denote the range of $c$ (whenever we use this notation the curve $c$ has all of $\mathbb{R}$ as its domain). When it is clear from the context, we sometimes write $c$ to mean $c(\mathbb{R})$, for example 

%For a path or binding path, $p_{0,1,\ldots,k}$ of length $\geq 1$, for $0 \leq i < k$ the notation $\hp{p}{i}{i+1}$ denotes the midpoint of the unit-length line segment $\embed{p_{i,i+1}}$. For a path $P$, we have $\hp p i {i+1} = \pos{\glu P i}$, hence this notation is especially useful for binding paths, since they do not have glues.

\subsubsection{Cutting the plane with curves; left and right turns}\label{subsubsec:defs:cuting}
In this paper we use finite and periodic infinite polygonal curves to cut the $\mathbb{R}^2$ plane into two pieces. 
The finite polygonal curves we use consist of a finite number of concatenations of vertical and horizontal segments of length 1. 
If the curve is simple and closed we may apply the Jordan Curve Theorem
to cut the plane into connected components. 
\begin{theorem}[Jordan Curve Theorem]
  \label{thm:jordan}
  Let $c$ be a simple closed curve, then $c$ cuts $\R^2$ into two connected components.\end{theorem}
Here, we have stated the theorem in its general form, although for our results the (easier to prove) polygonal version suffices. Moreover, one of two components is finite and we call the restriction of $\Z$ to this finite area, the \emph{interior} of the cycle. The interior also includes the positions of the curve which are in $\Z$.

%%%REFAIRE à partir de là
The second kind of curve we use is periodic ones. Indeed, we will consider a good candidate path $P$ and we cut the plan using the embedding of $\bipump{P}$. For such infinite polygonal curves we also state and prove a slightly different version of the polygonal Jordan Curve Theorem, as Theorem~\ref{thm:infinite-jordan}.

\begin{theorem}[Jordan Curve Theorem]
  \label{thm:jordan}
  Let $c$ be a periodic bi-infinite simple curve, then $c$ cuts $\R^2$ into two connected components.\end{theorem}

% composed of one or two infinite rays, along with a finite number of length 0.5 or length 1 segments. 

In Section~\ref{appB:curves} we define what it means for one curve to turn left or right from another, as well as left hand side and right hand side of a cut of the real plane.  We define the \emph{left side} $\leftside{\bipump{P}}$ and \emph{right side} $\rightside{\bipump{P}}$ of a good candidate path $P$ as the restriction of $\Z$ to the left hand side and right side of the curve obtained by $\bipump{P}$. Note that both the left side and right side contains $\pos{\bipump{P}}$.
Consider two good candidate paths $P$ and $Q$, we say $P$ is greater than $Q$, denoted by $P\geq Q$, if and only if $\bipump{P}$ is inside  the left side of $\bipump{Q}$, \emph{i.e.} $\pos{\bipump{P}}\subset \bipump{\leftside{Q}}$. Moreover, if $\bipump{P}\neq \bipump{Q}$, we say that $P$ is strictly greater than $Q$, denoted by $P>Q$.  If we consider a good candidate path $P$ and a vector $\vect{v}$ such that $\bipump{P}$ and $\bipump{P}+\vect{v}$ do not intersect  and  $P>P+\vect{v}$ then there exists $\ell$ such that for any $(x,y)\in Z^2$, $(x,y) +\ell\vect{v} \in \rightside{\bipump{P}}\cap \leftside{\bipump{P}+\vect{v}}$ (see \ref{appB:curves} for more details). 

Consider an arc $A$ or a path $P$ which grow on $\bipump{P}$ then their positions are  either all in the left side of $P$ or all in the right side of $P$. Moreover if the width of the arc $A$ is as least $|P|$ then the arc $A$ must intersect with $A+\vect{P_0{P_{|P-1|}}}$ and this intersect does not occurs at its extremities, \emph{i.e.} it occurs at $\pos{A_j}$ with $0<j<|A|-1$.

\begin{observation}
\label{obs:grow}
Consider a directed tile assembly system $\mathcal{T}=(T,\sigma,1)$ whose terminal assembly is $\uniterm$, a simply pumpable path $P$ in $\inuniterm$ and a path $Q$ which grows on $\pump{P}$ at index $i \geq |P|-1$ then either $Q$ intersects with $Q+\vect{P_0P_{|P|-1}}$ or for all $\ell \in \N$, $Q+\ell\vect{P_0P_{|P|-1}}$ grows on $\pump{P}$.
\end{observation}

\begin{proof}
If some $\ell\in\N$, the path $Q+\ell\vect{P_0P_{|P|-1}}$ does not grow on $\pump{P}$ then there exists $0<j\leq |Q|-1$ such that $\pos{Q_j}+\ell\vect{P_0P_{|P|-1}} \in \pos{\pump{P}}$. Since $Q$ grows on $\pump{P}$ then $\pos{Q_j} \in \bipump{P}$. The binding path of $Q_{0,1,\ldots, j}$ is the binding path of an arc of $\bipump{P}$ whose width is at least $|P|$ (since $i \geq |P|-1$) and thus $Q_{0,1,\ldots, j}$ intersects with $Q_{0,1,\ldots, j}+\vect{P_0P_{|P|-1}}$.
\end{proof}

\section{Proof of the main theorems}

From now on, all tile assembly systems considered here are directed.

\subsection{Study of the bi-pumpable paths}

\subsubsection{Role of the seed}

In this subsection, corollary \ref{cor:periodic} and lemma \ref{lem:exist:bipump} show the equivalence between $\alpha$ is $\vect{v}$-periodic and there exists a bi-pumpable path $P$ where $\vect{v}=\vect{P_0P_{|P|-1}}$. To prove this result, we show that the seed can be reduced to any single tile of $\bipump{P}$ (the arguments used in this proof are the ones used in lemma $4.8$ of \cite{Doty-2011}). Later, this lemma will allow us to grow and pump paths easily.

\begin{lemma}
\label{lem:bipump:seed}
Let $\mathcal T=(T,\sigma,1)$ be a directed tile assembly system and $\uniterm$ its unique terminal assembly. If a path $P\in \inuniterm$ is bi-pumpable then for any $i \in \mathbb{Z}$, $(T,\bipump{P}_i,1)$ is directed and its terminal assembly is $\uniterm$.
%If $P$ is bi-pumpable in $\alpha$ and there exists $i \in \mathbb{Z}$ and a path $Q$ which starts in $\bipump{P}_i$ then $Q$ is producible in $\alpha$.
\end{lemma}

\begin{proof}
Since $\bipump{P}$ is in $\inuniterm$ then we can consider a finite producible assembly $\ass$ of $(T,\sigma,1)$ such that $\bipump{P}_i$ is a tile of $\ass$. By hypothesis, $\bipump{P}\cup \ass$ is producible by $(T,\sigma,1)$. Consider a path $R$ producible by $(T,\bipump{P}_i,1)$, if $R$ does not conflict with $\ass$ or $\bipump{P}$ then $(R\cup \bipump{P}\cup \ass)$ is producible by $(T,\sigma,1)$ and $R$ is in $\inuniterm$. For the sake of contradiction suppose that such a conflict exists. We assume that $R_{|R|-1}$ is the only conflict between $R$ and either $\ass$ or $\bipump{P}$ (by eventually replacing $R$ with one of its prefix), see Figure \ref{fig;bipump:seed} for an illustration of the following cases:
\begin{itemize} 
\item If this conflict occurs with $\bipump{P}_j$ for $j>i$ (the case $j<i$ is symmetric), since $\ass$ and $R$ are finite and $\vect{P_0P_{|P|-1}}$ is not null, there exists $\ell\in \N$ such that $R+\ell\vect{P_0P_{|P|-1}}$ does not intersect with $\ass$. Remark that $R_0+\ell\vect{P_0P_{|P|-1}}=\bipump{P}_{i+\ell(|P|-1)}$ and that $R_{|R|-1}+\ell\vect{P_0P_{|P|-1}}$ is in conflicts with $\bipump{P}_{j+\ell(|P|-1)}$. By definition of $\ass$ and $P$, the assembly $\gamma=\ass \cup \bipump{P}_{i,\ldots,i+\ell(|P|-1)}$ is producible by $(T,\sigma,1)$. By definition of $\ell$, the tile $R_{|R|-1}+\ell\vect{P_0P_{|P|-1}}$ is not a tile of $\ass$ and since $j>i$, we have $i+\ell(|P|-1)<j+\ell(|P|-1)$ thus $R_{|R|-1}+\ell\vect{P_0P_{|P|-1}}$ is not a tile of $\gamma$. Then the assembly $\gamma \cup (R+\ell\vect{P_0P_{|P|-1}})$ is producible by $(T,\sigma,1)$ and is in conflict with $\bipump{P}\in \inuniterm$ which is a contradiction.
\item Otherwise, this conflict occurs with $\ass$. Since $\ass$ and $R$ are finite and $\vect{P_0P_{|P|-1}}$ is not null, there exists $\ell \in \mathbb{N}$ such that $\ass+\ell\vect{P_0P_{|P|-1}}$ and $R+\ell\vect{P_0P_{|P|-1}}$ do no intersect with $\ass$. Since $\bipump{P}$ is $\vect{P_0P_{|P|-1}}$-periodic then $\bipump{P}$ does not conflict with $\ass+\ell\vect{P_0P_{|P|-1}}$ or $R+\ell\vect{P_0P_{|P|-1}}$. Then the two assemblies $(\ass \cup \bipump{P} \cup (\ass+\ell\vect{P_0P_{|P|-1}}))$ and $(\ass \cup \bipump{P} \cup (R+\ell\vect{P_0P_{|P|-1}}))$ are both producible by $(T,\sigma,1)$, but these two assemblies are in conflict which contradicts the fact that $(T,\sigma,1)$ is directed.
\end{itemize} 
Thus, any path producible by $(T,\bipump{P}_i,1)$ is producible by $(T,\sigma,1)$. Then creating a conflict in $(T,\bipump{P}_i,1)$ would create a conflict in $(T,\sigma,1)$ and since $(T,\sigma,1)$ is directed then $(T,\bipump{P}_i,1)$ is directed. To conclude, since $(T,\bipump{P}_i,1)$ is directed and $\uniterm$ is a terminal assembly which contains $\bipump{P}_i$, then $\uniterm$ is the unique terminal assembly of $(T,\bipump{P}_i,1)$.
\end{proof}

\begin{figure}
\centering

\begin{minipage}{0.4\linewidth}
\begin{tikzpicture}[x=0.33cm,y=0.33cm]

\draw[very thick] (0,0.5) -| (1.5,3.5)  -| (3.5,2.5) -| (6.5,5.5)  -| (8.5,4.5) -| (11.5,7.5)  -| (13.5,6.5) -| (15,6.5);
\draw (11.4, 8.5) node {$\bipump{P}_i$};

\draws{7}{9}
\draws{7}{10}
\draws{8}{9}
\draws{8}{10}
\draws{9}{10}
\draws{9}{11}
\draws{10}{10}
\draws{10}{11}
\draws{11}{10}
\draws{11}{11}

\drawt{0}{0}
\drawt{1}{0}
\drawt{1}{1}
\drawt{1}{2}
\drawt{1}{3}
\drawt{2}{3}
\drawt{3}{3}
\drawt{3}{2}
\drawt{3}{2}
\drawt{4}{2}
\drawt{5}{2}
\drawt{6}{2}
\drawt{6}{3}
\drawt{6}{4}
\drawt{6}{5}
\drawt{7}{5}
\drawt{8}{5}
\drawt{8}{4}
\drawt{9}{4}
\drawt{10}{4}
\drawt{11}{4}
\drawt{11}{5}
\drawt{11}{6}
\drawt{11}{7}
\drawt{12}{7}
\drawt{13}{7}
\drawt{13}{6}
\drawt{14}{6}

\draw[very thick] (7.5,10.5) -| (11.5,11.5)  -| (9.5,11.5);
\draw[very thick] (7.5,10.5) |- (8.5,9.5)  -| (8.5,9);
\draw[very thick,blue] (8.5,9) |- (11.5,4.5)  -| (11.5,7.5);

\drawb{8}{8}
\drawb{8}{7}
\drawb{8}{6}
\drawb{8}{5}
\drawb{8}{4}
\drawb{9}{4}
\drawb{10}{4}
\drawb{11}{4}
\drawb{11}{5}
\drawb{11}{6}
\drawb{11}{7}

\path [dotted, draw, thin] (0,0) grid[step=0.33cm] (15,13);
\end{tikzpicture}

a) $\bipump{P}$ is in white, the seed is in black and $\ass$ is the union of the seed and the blue path.
\end{minipage} \hfill
\begin{minipage}{0.4\linewidth}
\begin{tikzpicture}[x=0.33cm,y=0.33cm]

\draw (11.4, 6.3) node {$\bipump{P}_i$};

\draw[very thick,purple!50!white] (11.5,7.5) |- (13.5,10.5)  -| (13.5,7.5);
\draw[very thick,red!50!white] (11.5,7.5) -| (7.5,5.5);

\draw (6.5, 6.5) node {$R$};
\draw (12.5, 11.5) node {$R'$};

\drawt{11}{7}
\drawr{10}{7}
\drawr{9}{7}
\drawr{8}{7}
\drawr{7}{7}
\drawr{7}{6}
\drawr{7}{5}

\drawp{11}{8}
\drawp{11}{9}
\drawp{11}{10}
\drawp{12}{10}
\drawp{13}{10}
\drawp{13}{9}
\drawp{13}{8}
\drawp{13}{7}

%manque R et R'

\path [dotted, draw, thin] (0,0) grid[step=0.33cm] (15,13);
\end{tikzpicture}

b) Two paths $R$ and $R'$ which are producible by $(T,\bipump{P}_i,1)$.
\end{minipage} \hfill 

\vspace{+1em}

\begin{tikzpicture}[x=0.33cm,y=0.33cm]

\draw[very thick] (3,3.5)  -| (3.5,2.5) -| (6.5,5.5)  -| (8.5,4.5) -| (11.5,7.5)  -| (13.5,6.5) -| (16.5,9.5) -| (18.5,8.5) -| (21.5,11.5) -| (23.5,10.5) -| (26.5,13.5)  -| (28,13.5);

\draws{7}{9}
\draws{7}{10}
\draws{8}{9}
\draws{8}{10}
\draws{9}{10}
\draws{9}{11}
\draws{10}{10}
\draws{10}{11}
\draws{11}{10}
\draws{11}{11}

\drawt{3}{3}
\drawt{3}{2}
\drawt{3}{2}
\drawt{4}{2}
\drawt{5}{2}
\drawt{6}{2}
\drawt{6}{3}
\drawt{6}{4}
\drawt{6}{5}
\drawt{7}{5}
\drawt{8}{5}
\drawt{8}{4}
\drawt{9}{4}
\drawt{10}{4}
\drawt{11}{4}
\drawt{11}{5}
\drawt{11}{6}
\drawt{11}{7}
\drawt{12}{7}
\drawt{13}{7}
\drawt{13}{6}
\drawt{14}{6}
\drawt{15}{6}
\drawt{16}{6}
\drawt{16}{7}
\drawt{16}{8}
\drawt{16}{9}
\drawt{17}{9}
\drawt{18}{9}
\drawt{18}{8}
\drawt{19}{8}
\drawt{20}{8}
\drawt{21}{8}
\drawt{21}{9}
\drawt{21}{10}
\drawt{21}{11}
\drawt{22}{11}
\drawt{23}{11}
\drawt{23}{10}
\drawt{24}{10}
\drawt{25}{10}
\drawt{26}{10}
\drawt{26}{11}
\drawt{26}{12}
\drawt{26}{13}
\drawt{27}{13}

\draw (11.4, 8.5) node {$\bipump{P}_i$};

\draw[very thick] (7.5,10.5) -| (11.5,11.5)  -| (9.5,11.5);
\draw[very thick] (7.5,10.5) |- (8.5,9.5)  -| (8.5,9);
\draw[very thick,blue] (8.5,9) |- (11.5,4.5)  -| (11.5,7.5);
\draw[very thick,blue] (18.5,13) |- (21.5,8.5)  -| (21.5,11.5);

\draw[very thick] (18.5,13) |- (17.5,13.5) |- (21,14.5);

\draw[very thick,purple!50!white] (21.5,11.5) |- (21.5,14);
\draw[very thick,purple!50!white] (26.5,13.5) |- (28.5,16.5)  -| (28.5,13.5);
\draw[very thick,red!50!white] (6.5,5.5) -| (2.5,3.5);

\draws{17}{13}
\draws{17}{14}
\draws{18}{13}
\draws{18}{14}
\draws{19}{14}
\draws{20}{14}

\drawb{8}{8}
\drawb{8}{7}
\drawb{8}{6}
\drawb{8}{5}
\drawb{8}{4}
\drawb{9}{4}
\drawb{10}{4}
\drawb{11}{4}
\drawb{11}{5}
\drawb{11}{6}
\drawb{11}{7}

\drawb{18}{12}
\drawb{18}{11}
\drawb{18}{10}
\drawb{18}{9}
\drawb{18}{8}
\drawb{19}{8}
\drawb{20}{8}
\drawb{21}{8}
\drawb{21}{9}
\drawb{21}{10}
\drawb{21}{11}

\drawr{6}{5}
\drawr{5}{5}
\drawr{4}{5}
\drawr{3}{5}
\drawr{2}{5}
\drawr{2}{4}
\drawr{2}{3}

\drawp{21}{12}
\drawp{21}{13}

\drawp{26}{13}
\drawp{26}{14}
\drawp{26}{15}
\drawp{26}{16}
\drawp{27}{16}
\drawp{28}{16}
\drawp{28}{15}
\drawp{28}{14}
\drawp{28}{13}

\draw (21, 16.5) node {conflict with $\ass$};
\draw (25, 1.5) node {conflicts with $\bipump{P}$};
\draw[->] (28.5,2.5) -- (28.5,13.5);
\draw[->] (20,1.5) -| (2.5,3.5);
\draw[->] (21.5,16) -- (21.5,14.5);

\draw (1.5, 4.5) node {$R$};
\draw (22.5, 13.5) node {$R'$};
\draw (27.5, 17.5) node {$R'$};

\draw[red,->] (11.5,7.5) -- (6.5,5.5);
\draw[purple,->] (11.5,7.6) -- (21.5,11.6);
\draw[purple,->] (11.5,7.4) -- (26.5,13.4);

\path [dotted, draw, thin] (0,0) grid[step=0.33cm] (30,20);
\end{tikzpicture}

c) The three different ways to create a conflict in $(T,\sigma,1)$ by using translations of $R$ and $R'$.

\caption{Illustrations of proof of lemma \ref{lem:bipump:seed}.}
\label{fig;bipump:seed}
\end{figure}

As a corollary of this result, any path $Q$ which grows on $\bipump{P}$ is in $\inuniterm$. Moreover, since for any $\ell \in \Z$, $\bipump{P}+\ell\vect{P_0P_{|P|-1}}=\bipump{P}$ then $\bipump{Q}+\ell\vect{P_0P_{|P|-1}}$ also grows in $\alpha$ and is in $\inuniterm$ which leads to the following corollary: 

\begin{corollary}
\label{cor:periodic}
Let $\mathcal T=(T,\sigma,1)$ be a directed tile assembly system and $\uniterm$ its unique terminal assembly. If $P\in \inuniterm$ is bi-pumpable then $\alpha$ is $\vect{P_{0}P_{|P|-1}}$-periodic.
\end{corollary}

\begin{lemma}
\label{lem:exist:bipump}
Let $\mathcal T=(T,\sigma,1)$ be a directed tile assembly system and $\uniterm$ its unique terminal assembly. If $\uniterm$ is periodic then there exists a path $P\in \inuniterm$ which is bi-pumpable.
\end{lemma}

\begin{proof}
By hypothesis, there exists a non null vector $\vect{v}$, such that $\uniterm$ is $\vect{v}$-periodic. Consider a tile $A$ of $\uniterm$ then $A+\vect{v}$ is also tile of $\uniterm$ and there exists a finite path $P\in \inuniterm$ such that $P_0=A$ and $P_{|P|-1}=A+\vect{v}$. Consider the shortest path $Q$ of $\inuniterm$ such that $Q_0+s\vect{v}=Q_{|Q|-1}$ for some $s\in \{-1,1\}$. The path $Q$ is correctly defined since $P$ satisfies this criterion. Since $\uniterm$ is $\vect{v}$-periodic, then for $\ell \in \Z$, $Q+\ell\vect{v}$ is in $\inuniterm$. If $Q$ and $Q+\vect{v}$ intersect only at their extremities then $Q$ is a good candidate and $\bipump{Q}$ is correctly defined, is in $\inuniterm$ and thus $Q$ is bi-pumpable. Otherwise, there exists another intersection between $Q$ and $Q+\vect{v}$, \emph{i.e.} there exists $0\leq i, j \leq |Q|-1$ such that $Q_i=Q_j+\vect{v}$ and if $j > i $ (resp. $i > j$) then $(i,j)\neq (0,|Q|-1)$ (resp. $(j,i)\neq (0,|Q|-1)$). Thus, $Q_{i,\ldots, j}$ (resp. $Q_{j,\ldots, i}$) contradicts the minimality of $Q$.
\end{proof}

\subsubsection{Paths without redundancy}

Now, we show that we can restrict our study to bi-pumpable paths whose two extremities are the only tile to share a common tile type.

\begin{definition}
A path $P$ is without redundancy if for all $0\leq i <j \leq |P|-1$, $\type{P_i}=\type{P_j}$ implies that $i=0$ and $j=|P|-1$. 
\end{definition}

Of course the length of a path without redundancy is bounded by $|T|+1$. Moreover, if two tiles of a path shares the same tiles type then a path without redundancy can be extracted from it. 

\begin{lemma}
\label{lem:redun:exists}
For any path $P$ where there exists $0\leq i < j \leq |P|-1$ such that $\type{P_{i}}=\type{P_{j}}$, there exists $i\leq i' < j' \leq j$ such that $P_{i',i'+1,\ldots, j'}$ is without redundancy. 
\end{lemma}

\begin{proof}
Consider $j'=\min\{i<j'\leq j:$ there exists $i \leq i'<j'$ such that $\type{P_{i'}}=\type{P_{j'}}\}$. Since $\type{P_{i}}=\type{P_{j}}$, $j'$ is correctly defined and by its definition there exists $i\leq i'<j'$ such that  $P_{i',i'+1,\ldots, j'}$ is without redundancy.
\end{proof}

The two following lemmas show that paths without redundancy are easy to pump in both direction.

\begin{lemma}
\label{lem:redun:good}
Let $\mathcal T=(T,\sigma,1)$ be a directed tile assembly system and $\uniterm$ its unique terminal assembly. For any path without redundancy $P\in\inuniterm$  if $P+\vect{P_0P_{|P|-1}}$ is in $\inuniterm$ then $P$ is a good candidate. 
\end{lemma}

\begin{proof}
If  $P$ and $P+\vect{P_0P_{|P|-1}}$ intersect then there exist $0\leq i,j \leq |P|-1$ such that $P_i=P_j+\vect{P_0P_{|P|-1}}$. Since $\vect{P_0P_{|P|-1}}$ is not null then $i \neq j$. Moreover, since $P$ and $P+\vect{P_0P_{|P|-1}}$ are both in $\inuniterm$ then $\type{P_i}=\type{P_j}$ and thus $i=|P|-1$ and $j=0$.
\end{proof}

The intuition of the following lemma is that a bi-pumpable path without redundancy could be extracted from a path $Q$ if both extremities of $Q$ belong to bi-pumpable paths and have the same tile types. 

\begin{lemma}
\label{lem:redun:pump}
Let $\mathcal T=(T,\sigma,1)$ be a directed tile assembly system, $\uniterm$ its unique terminal assembly and two paths $P,P'$ of $\inuniterm$ which are bi-pumpable. If there exists a path $Q$ such that $\type{Q_0}=\type{Q_{|Q|-1}}$ and $Q_0$ (resp. $Q_{|Q|-1}$) is a tile of $\bipump{P}$ (resp. $\bipump{P'}$) then there exist $0\leq i < j \leq |Q|-1$ such that the path $Q_{i,\ldots,j}$ is a bi-pumpable path without redundancy of $\inuniterm$.
\end{lemma}

\begin{proof}
Since $\type{Q_0}=\type{Q_{|Q|-1}}$ then by lemma \ref{lem:redun:exists}, there exists $0\leq i < j \leq |Q|-1$ such that the path $R=Q_{i,\ldots,j}$ is without redundancy. Note that we just need to prove than $R$ is in $\inuniterm$ and pumpable because $\reverse{R}$ satisfies the same hypothesis as $R$ and would also be pumpable. If both $R$ and $\reverse{R}$ are pumpable then $R$ is bi-pumpable.

If there is no conflict between $Q$ and $R+\ell\vect{R_0R_{|R|-1}}$ for any $\ell \in \N$, then in particular, there are no conflict between $Q$ and $R+\vect{R_0R_{|R|-1}}$ and thus the assembly $Q\cup(R+\vect{R_0R_{|R|-1}})$ is correctly defined. Then $Q\cup(R+\vect{R_0R_{|R|-1}})$ is producible by $(T,Q_0,1)$.  By lemma \ref{lem:bipump:seed} and since $Q_0$ is a tile of $\bipump{P}$, the tile assembly system $(T,Q_0,1)$ is directed and its terminal assembly is also $\alpha$, then $R$ and $R+\vect{R_0R_{|R|-1}}$ are both in $\inuniterm$ and by lemma \ref{lem:redun:good}, $R$ is a good candidate. Then $\pump{R}$ is correctly defined. Remind that we are in a case where $\pump{R}$ does not conflict with $Q$ then $Q\cup \pump{R}$ is producible by $(T,Q_0,1)$ and $\pump{R}$ is in $\inuniterm$.

Otherwise, suppose for the sake of contradiction that there is a conflict between $Q$ and $R+\ell\vect{R_0R_{|R|-1}}$ for some $\ell \in \N$, let $$\ell=\min\{\ell'\in \N: R+\ell'\vect{R_0R_{|R|-1}}\text{ conflicts with }Q\} \text{ and let }$$  $$m=\min\{0\leq k \leq |R|-1: R_{k}+\ell\vect{R_0R_{|R|-1}}\text{ conflicts with }Q\}$$ and $0\leq k \leq |Q|-1$ such that $\pos{Q_{k}}=\pos{R_{m}}$. If $\ell=1$, we defined $R'$ as $(R_{0,\ldots,m}+\vect{P_0P_{|P|-1}})$ and since $R'_0=R_0+\vect{P_0P_{|P|-1}}=R_{|R|-1}=Q_j$, we have $m>0$ and the conflict between $R'$ and $Q$ does not occur at $\pos{Q_j}$ since $R'$ is simple, \emph{i.e.} $k\neq j$. If $\ell>1$ as explained in the previous paragraph $\pump{R}$ is correctly defined and we defined $R'$ as $\pump{R}_{|R|-1,\ldots,(\ell-1)(|R|-1)+m}$ in this case. We also have $R'_0=Q_j$ and $R'_{|R'|-1}$ conflicts with $Q$ and this conflicts does not occur at $\pos{Q_j}$ since $R'$ is simple, \emph{i.e.} $k\neq j$. In both cases, $k\neq j$ and there is exactly one conflict between $R'$ and $Q$ with $Q_j\neq R'_{|R'|-1}$.
If $k>j$ (resp. $k<j$) then the assembly $Q_{0,\ldots, j} \cup R'$ (resp. $Q_{j,\ldots, |Q|-1} \cup R'$) is producible by $(T,Q_0,1)$ (resp. $(T,Q_{|Q|-1},1)$). By lemma \ref{lem:bipump:seed} and since $Q_0$ (resp. $Q_{|Q|-1}$) is a tile of $\bipump{P}$ (resp. $\bipump{P'}$), the tile assembly system $(T,Q_0,1)$ (resp. $(T,Q_{|Q|-1},1)$) is directed and its terminal assembly is also $\uniterm$, then $R'$ is in $\inuniterm$ but it conflicts with $Q$ which is producible by $(T,Q_0,1)$ (resp. $(T,Q_{|Q|-1},1)$) and thus also in $\inuniterm$. 
%Similarly, if $k'<j$ then the assembly  $$Q_{j,\ldots, |Q|-1} \bigcup_{0\leq m \leq \ell-1} (R+m\vect{R_0R_{|R+1|}} \cup R_{[0,k]}+\ell\vect{R_{0\ldots k}}$$ is producible by $(T,Q_{|Q|},1)$. By lemma \ref{lem:bipump:seed} and since $Q_{|Q|-1}$ is a tile of $\bipump{P'}$, $(T,Q_{|Q|-1},1)$ is directed and its terminal assembly is also $\uniterm$, then $R_{k'}+\ell\vect{R_0R_{|R|-1}}$ is a tile of $\uniterm$ which is a contradiction since this tile conflicts with $Q$ which is producible by $(T,Q_{|Q|-1},1)$ and thus in $\inuniterm$.

% from $Q_0$ (resp. $Q_{|Q|-1}$) it is possible to bind $Q_1,Q_2, \ldots, Q_{i}$ (resp. $Q_{|Q|-1},Q_{|Q|-2}, \ldots, Q_{j}$) and then $\pump{R}_0,\pump{R}_1,\ldots,\pump{R}_{\ell|R|-1+k'}$ (resp. $\pump{R}_{|R|-1},\pump{R}_{|R|-2},\ldots,\pump{R}_{\ell|R|-1+k'}$). By lemma \ref{lem:bipump:seed}, $(T,Q_0,1)$ (resp. $(T,Q_{|Q|-1},1)$) is directed and its terminal assembly is also $\alpha$, then $R_{k'}+\ell\vect{R_0R_{|R|-1}}$ is in $\alpha$ which is a contradiction since $Q$ is in $\alpha$.
\end{proof}

Note that, we can use the previous lemma with $P=P'$ which implies that if there is a bi-pumpable path, there is a bi-pumpable path without redundancy.

\begin{corollary}
\label{cor:exists:without}
Let $\mathcal T=(T,\sigma,1)$ be a directed tile assembly system and $\uniterm$ be its unique terminal assembly. If there is a bi-pumpable path $P \in \alpha$ then there is a bi-pumpable path $Q$ without redundancy in $\alpha$. 
\end{corollary}

\begin{proof}
If a path $P$ of $\inuniterm$ is bi-pumpable then by definition of bi-pumpable, we have $\type{P_0}=\type{P_{|P|-1}}$. Then, by lemma \ref{lem:redun:exists}, there exists $0\leq i <j \leq |P|-1$ such that $Q=P_{i,\ldots,j}$ is without redundancy. Since both extremities of $Q$ agree with $P$ (and are in $\bipump{P}$), then by lemma \ref{lem:redun:pump}, $Q$ is bi-pumpable.
\end{proof}

The following lemma will be useful to partitioned the discrete plan in the next section (as explained in section \ref{subsubsec:defs:cuting}). Indeed, if a vector $\vect{v}$ is not collinear with $\vect{P_0P_{|P|-1}}$ where $P$ is a bi-pumpable path without redundancy then $P$ and $P+\vect{v}$ cannot intersect.
 
\begin{lemma}
\label{lem:redun:nointer}
Let $\mathcal T=(T,\sigma,1)$ be a directed tile assembly system and $\uniterm$ be its unique terminal assembly. Consider a path $P$ without redundancy and some vector $\vect{v} \in \Z^2$ such that both $P$ and $P+\vect{v}$ are in $\inuniterm$ and bi-pumpable. If $\bipump{P}$ and $\bipump{P}+\vect{v}$ intersect, we have $\vect{v}=\ell\vect{P_0P_{|P|-1}}$ for some $\ell \in \Z$ and $\bipump{P}=\bipump{P}+\vect{v}$. 
\end{lemma}

\begin{proof}
If $\bipump{P}$ and $\bipump{P}+\vect{v}$ intersect then this intersection is not a conflict since $\bipump{P}$ and $\bipump{P}+\vect{v}$ are both in $\inuniterm$. Thus there exists $i,j \in \Z$ such that $\bipump{P}_i=\bipump{P}_j+\vect{v}$. Since $P$ is without redundancy then for all $k \in \Z$ such that $\type{\bipump{P}_k}=\type{\bipump{P}_i}$, there exists $\ell$ such that $P_k=P_i+\ell\vect{P_0P_{|P|-1}}$ and thus there exists $\ell \in \Z$ such that $\vect{v}=\ell\vect{P_0P_{|P|-1}}$ for some $\ell \in \Z$. Finally, since $\bipump{P}$ is $\vect{P_0P_{|P|-1}}$-periodic then $\bipump{P}=\bipump{P}+\ell\vect{P_0P_{|P|-1}}=\bipump{P}+\vect{v}$.
\end{proof}

\subsubsection{Reusing a tile type}

Suppose that a tile $A$ of $\alpha$ shares a common tile type with a tile used in a bi-pumpable path $P$ then either $A$ is a tile of $\bipump{P}$ or not. In the second case, we show in this section that $\alpha$ is bi-periodic. This remark will allow us to describe efficiently the bi-periodic terminal assemblies and prove one of the main theorem. We proceed in four steps, first we show that $\bipump{P}+\vect{P_iA}$ can grow from $A$ (lemma \ref{biperiodic:translation}), then we show that $\vect{P_iA}$ cannot be collinear with $\vect{P_0P_{|P|-1}}$ (lemma \ref{biperiodic:notcol} and lemma \ref{biperiodic:sameperiod}), afterwards we can find a bi-pumpable $Q$ without redundancy such that $\vect{Q_0Q_{|Q|-1}}$ is not collinear with $\vect{P_0P_{|P|-1}}$ (see lemma \ref{biperiodic:otherpath}) and we conclude by proving one of the main theorem \ref{the:MT1bis} (using lemma \ref{alpha:type1}) by assembling a simple cycle, doing any possible binding in the interior of this cycle and concluding that $\alpha$ is obtained by tiling the $2D$ plane with translations of this assembly.

\begin{lemma}
\label{biperiodic:translation}
Let $\mathcal T=(T,\sigma,1)$ be a directed tile assembly system, $\uniterm$ its unique terminal assembly and a bi-pumpable path $P\in \inuniterm$ without redondancy. If there exist $i \in Z$ and a tile $A$ in $\alpha$ such that $type(A)=type(P_i)$ then $\bipump{P}+\vect{P_iA}$ is in $\inuniterm$. 
\end{lemma}

\begin{proof}
First, if there exists $j\in \Z$ such that $A$ is the tile $\bipump{P}_j$ then $\bipump{P}+\vect{P_iA}$ is producible by $(T,\bipump{P}_j,1)$  and by lemma \ref{lem:bipump:seed}, $(T,\bipump{P}_j,1)$ is directed and its terminal assembly is also $\alpha$. Then $\bipump{P}+\vect{P_iA}$ is in $\inuniterm$ in this case. Otherwise, since $A$ is a tile of $\alpha$ and $\bipump{P}$ is in $\inuniterm$, there exists a path $Q\in \inuniterm$ and $j\in \Z$ such that $Q$ grows on $\bipump{P}$ at index $j$ and $Q_{|Q|-1}=A$. If $Q$ does not conflict with $\bipump{P}+\vect{P_iA}$ then $Q\cup (\bipump{P}+\vect{P_iA})$ is an assembly producible by $(T,\bipump{P}_j,1)$ and by lemma \ref{lem:bipump:seed}, $(T,\bipump{P}_j,1)$ is directed and its terminal assembly is also $\alpha$. Thus $\bipump{P}+\vect{P_iA}$ is in $\inuniterm$. Otherwise, for the sake of contradiction suppose that $Q$ conflicts with $\bipump{P}+\vect{P_iA}$ then $Q-\vect{P_iA}$ conflicts with $\bipump{P}$. Moreover, since $Q_{|Q|-1}-\vect{P_iA}=A-\vect{P_iA}=P_i$ then $Q-\vect{P_iA}$ is producible by $(T,\bipump{P}_i,1)$ and by lemma \ref{lem:bipump:seed}, $(T,\bipump{P}_i,1)$ is directed and its terminal assembly is also $\alpha$. In this case, $Q-\vect{P_iA}$ is in $\inuniterm$ which is a contradiction since it conflicts with $\bipump{P}\in \inuniterm$.
\end{proof}

\begin{lemma}
\label{biperiodic:notcol}
Let $\mathcal T=(T,\sigma,1)$ be a directed tile assembly system, $\uniterm$ its unique terminal assembly and a bi-pumpable path $P\in \inuniterm$ without redondancy. If there exists $0\leq i \leq |P|-1$ and a tile $A$ of $\uniterm$ such that $type(A)=type(P_i)$ and $\vect{P_iA}$ is collinear with $\vect{P_0P_{|P|-1}}$ then $A$ is a tile of $\bipump{P}$.
\end{lemma}

\begin{proof}
See Figure \ref{fig:geo} for an illustation of this reasoning and appendix \ref{appB:curves} for more technical details. Consider the line $\ell^+$ (resp. $\ell^-$) of direction $\vect{P_0P_{|P|-1}}$ included into $\mathcal{L}({\bipump{P}})$ (resp. $\mathcal{R}({\bipump{P}})$) which is tangent to $\embed{\bipump{P}}$. Then there exists $0\leq j \leq |P|-1$ (resp. $0\leq k \leq |P|-1$), such that $\ell^+$ (resp. $\ell^-)$ passes by $\pos{P_j}$ (resp. $\pos{P_k}$). Let $v^+$ (resp. $v^-$) be $\vect{P_0P_{|P|-1}}$ rotated by $\pi/2$ (resp. $-\pi/2$). Then the ray $r^+$ (resp. $r^-$) starting in $\pos{P_j}$ (resp. $\pos{P_k}$) of direction $v^+$ (resp. $v^-$) intersects $\bipump{P}$ only at $\pos{P_j}$ (resp. $\pos{P_k}$) and is included into $\mathcal{L}(\bipump{P})$ (resp. $\mathcal{R}(\bipump{P})$). By hypothesis, $P_i+\vect{P_iA}=A$ is a tile of $\uniterm$ and by lemma \ref{biperiodic:translation}, $\bipump{P}+\vect{P_iA}$ is in $\inuniterm$. Thus $P_j+\vect{P_iA}$ and $P_k+\vect{P_iA}$ are both tile of $\uniterm$. Since $\vect{P_0P_{|P|-1}}$ and $\vect{P_iA}$ are collinear, the ray $r^++\vect{P_iA}$ (resp. $r^-+\vect{P_iA}$) still belongs to $\mathcal{L}(\bipump{P})$ (resp. $\mathcal{R}(\bipump{P})$) and may intersect $\bipump{P}$ only at $\pos{P_j}+\vect{P_iA}$ (resp. $\pos{P_k}+\vect{P_iA}$) thus $\pos{P_j}+\vect{P_iA}$ is in $\leftside{\bipump{P}}$ (resp. $\pos{P_k}+\vect{P_iA}$ is in $\rightside{\bipump{P}}$). Then there is an intersection between $\bipump{P}$ and $\bipump{P}+\vect{P_iA}$ which means by lemma \ref{lem:redun:nointer} that $\bipump{P}=\bipump{P}+\vect{P_iA}$. Finally, $A=P_i+\vect{P_iA}$ is thus a tile of $\bipump{P}$. 

\begin{figure}
\centering
\begin{tikzpicture}[x=0.3cm,y=0.3cm]

%2 zone

\fill[red!10!white] (0,0) |- (36,21) |- (30.5,16.5) |- (28.5,15.5) |- (25.5,14.5) |- (24.5,12.5) |- (18.5,10.5) |-(16.5,9.5) |- (13.5,8.5) |- (12.5,6.5) |- (6.5,4.5) |- (4.5,3.5) |- (1.5,2.5) |- (0,0);
\fill[blue!10!white] (0,0) -| (36,18) |- (30.5,16.5) |- (28.5,15.5) |- (25.5,14.5) |- (24.5,12.5) |- (18.5,10.5) |-(16.5,9.5) |- (13.5,8.5) |- (12.5,6.5) |- (6.5,4.5) |- (4.5,3.5) |- (1.5,2.5) |- (0.5,0.5) |- (0,0);

%courbe
\draw[very thick] (12.5,6.5) -| (13.5,8.5)  -| (16.5,9.5) -| (18.5,10.5)  -| (24.5,12.5);
\draw[thick,dotted] (0.5,0) -| (0.5,0.5) -| (1.5,2.5)  -| (4.5,3.5) -| (6.5,4.5)  -| (12.5,6.5);
\draw[thick,dotted] (24.5,12.5) -| (25.5,14.5)  -| (28.5,15.5) -| (30.5,16.5)  -| (36,16.5);

%l+ l-
\draw[red] (0,1.75) -- (36,19.75);
\draw[blue] (3.75,0) -- (36,16.25);
\draw (1,3) node {$\ell^+$};
\draw (8, 1) node {$\ell^-$};

%nom des lignes

\draw[thick,->] (7,8)--(19,14);
\draw[thick,->] (7,8)--(1,20);
\draw (15, 14) node {$\vect{P_0P_{|P|-1}}$};
\draw (29.5, 11.25) node {$\vect{P_iA}$};
\draw (5.5, 14) node {$\vect{v^+}$};

%notes
\draw (10.8,6) node {$P_0$};
\draw (13.5,9.5) node {$P_j$};
\draw (26.8,12.8) node {$P_{|P|-1}$};
\draw (25.5,10) node {$P_{k}$};

%chemin P
\drawt{12}{6}
\drawt{13}{6}
\drawt{13}{7}
\drawt{13}{8}
\drawt{14}{8}
\drawt{15}{8}
\drawt{16}{8}
\drawt{16}{9}
\drawt{17}{9}
\drawt{18}{9}
\drawt{18}{10}
\drawt{19}{10}
\drawt{20}{10}
\drawt{21}{10}
\drawt{22}{10}
\drawt{23}{10}
\drawt{24}{10}
\drawt{24}{11}
\drawt{24}{12}

%r^+ r^- vector PiA
\draw[dashed,blue] (24.5,10.5) -- (30.75,0);
\draw[dashed,blue] (32.5,14.5) -- (36,7.5);
\draw[dashed,red] (13.5,8.5) -- (7.25,21);
\draw[dashed,red] (21.5,12.5) -- (17.25,21);
\draw[thick,->] (13.5,8.5) -- (21.5,12.5);
\draw[thick,->] (24.5,10.5) -- (32.5,14.5);
\end{tikzpicture}
\caption{Illustration of the proof of lemma \ref{biperiodic:notcol}. We consider a path $P$ such that that $\bipump{P}$ cuts the $2D$ plane in two parts. The translation of $P_j$ by $\vect{P_iA}$ is in the left side (in red) and the translation of $P_k$ by $\vect{P_iA}$ is in the right side (in blue). Then $\bipump{P}$ intersects with its translation by $\vect{P_iA}$.}
\label{fig:geo}
\end{figure}

\end{proof}

\begin{lemma}
\label{biperiodic:sameperiod}
Let $\mathcal T=(T,\sigma,1)$ be a directed tile assembly system, $\uniterm$ its unique terminal assembly and two bi-pumpable paths $P,Q\in \inuniterm$ without redondancy. If $\vect{P_0P_{|P|-1}}$ is collinear with $\vect{Q_0Q_{|Q|-1}}$ then $\vect{P_0P_{|P|-1}}=s\vect{Q_0Q_{|Q|-1}}$ for some $s\in\{-1,1\}$.
\end{lemma}

\begin{proof}
By corollary \ref{cor:periodic}, $\uniterm$ is $\vect{Q_0Q_{|Q|-1}}$-periodic which means that $P_0+\vect{Q_0Q_{|Q|-1}}$ is a tile of $\uniterm$. By lemma \ref{biperiodic:notcol}, since $\vect{P_0P_{|P|-1}}$ and $\vect{Q_0Q_{|Q|-1}}$ are collinear and since $P$ is without redundancy then $P_0+\vect{Q_0Q_{|Q|-1}}$ is in $\bipump{P}$. Since $\vect{P_0P_{|P|-1}}$ and $\vect{Q_0Q_{|Q|-1}}$ are not null then $\vect{Q_0Q_{|Q|-1}}=\ell\vect{P_0P_{|P|-1}}$ for some $\ell \in \Z^*$. Similarly, $\uniterm$ is $\vect{P_0P_{|P|-1}}$-periodic and $\vect{P_0P_{|P|-1}}=\ell'\vect{Q_0Q_{|Q|-1}}$ for some $\ell' \in \Z^*$. Then either $\vect{P_0P_{|P|-1}}=\vect{Q_0Q_{|Q|-1}}$ or $\vect{P_0P_{|P|-1}}=-\vect{Q_0Q_{|Q|-1}}$.
\end{proof}

\begin{lemma}
\label{biperiodic:otherpath}
Let $\mathcal T=(T,\sigma,1)$ be a directed tile assembly system, $\uniterm$ its unique terminal assembly and a bi-pumpable path $P\in \inuniterm$ without redundancy.
If there is a tile $A$ of $\uniterm$ such that $type(A)=type(P_i)$ and $\vect{P_iA}$ is not collinear with $\vect{P_0P_{|P|-1}}$ then there exists a bi-pumpable path $Q$ without redundancy in $\inuniterm$ such that $\vect{P_0P_{|P|-1}}$ is not collinear with $\vect{Q_0Q_{|Q|-1}}$.
\end{lemma}

\begin{proof}
See Figure \ref{fig:trans} for an illustration of this reasoning.
By lemma \ref{biperiodic:translation}, $\bipump{P}+\vect{P_iA}$ is in $\inuniterm$. Since $\vect{P_0P_{|P|-1}}$ is not collinear with $\vect{Q_0Q_{|Q|-1}}$ then by lemma \ref{lem:redun:nointer}, $\bipump{P}$ and $\bipump{P}+\vect{P_iA}$ do not intersect. Nevertheless, both paths are in $\inuniterm$ and thus there exists a path $Q$ in $\inuniterm$ such that $Q_0=P_i$, $\type{Q_{|Q|-1}}=\type{Q_0}=\type{A}$ and $Q_{|Q|-1}$ is a tile of $\bipump{P}+\vect{P_iA}$. Without loss of generality, we suppose that $Q$ is a path of minimal length which satisfies these three hypothesis. From lemma \ref{lem:redun:pump}, there exist $0\leq i' \leq j' \leq |Q|-1$ such that the path $R=Q_{i',\ldots,j'}$ is without redundancy, is in $\inuniterm$ and is bi-pumpable. If $\vect{R_0R_{|R|-1}}$ is not collinear with $\vect{P_0P_{|P|-1}}$ then the lemma is true. Suppose otherwise for the sake of contradiction. By corollary \ref{cor:periodic}, $\alpha=\alpha-\vect{R_0R_{|R|-1}}$ and since $Q_{i'}=R_0=R_{|R|-1}-\vect{R_0R_{|R|-1}}=Q_{j'}-\vect{R_0R_{|R|-1}}$, then the path $Q'=Q_{0,\ldots i'}(Q_{j'+1,\ldots, |Q|-1}-\vect{R_0R_{|R|-1}})$ is in $\inuniterm$. By definition of $Q$, we have $Q'_0=P_i$, $\type{Q'_{|Q'|-1}}=\type{Q_{|Q|-1}}=\type{A}$. Moreover by lemma \ref{biperiodic:sameperiod}, $\vect{R_0R_{|R|-1}}=s\vect{P_0P_{|P|-1}}$ with $s \in \{-1,1\}$ and then $Q'_{|Q'|-1}$ is in $\bipump{P}+\vect{P_iA}$. Since $|Q'|<|Q|$, this is a contradiction of the hypothesis that $Q$ is of minimal length. 

\begin{figure}
\centering

\begin{tikzpicture}[x=0.33cm,y=0.33cm]

%\drawyellow{30}{12}{100}

%figure centrale
%colle P
 \draw[very thick] (9,1.5) -| (9.5,3.5) -| (12.5,1.5) -| (19,1.5);
%chemin P
\drawpurple{9}{1}{100}
\drawpurple{9}{2}{95}
\drawpurple{9}{3}{90}
\drawpurple{10}{3}{85}
\drawpurple{11}{3}{80}
\drawpurple{12}{3}{75}
\drawpurple{12}{2}{70}
\drawpurple{12}{1}{65}
\drawpurple{13}{1}{60}
\drawpurple{14}{1}{55}
\drawpurple{15}{1}{50}
\drawpurple{16}{1}{45}
\drawpurple{17}{1}{40}
\drawpurple{18}{1}{35}

%colle P
 \draw[very thick] (0,3.5) -| (2.5,1.5) -| (9,1.5);
%chemin P
\drawpurple{0}{3}{85}
\drawpurple{1}{3}{80}
\drawpurple{2}{3}{75}
\drawpurple{2}{2}{70}
\drawpurple{2}{1}{65}
\drawpurple{3}{1}{60}
\drawpurple{4}{1}{55}
\drawpurple{5}{1}{50}
\drawpurple{6}{1}{45}
\drawpurple{7}{1}{40}
\drawpurple{8}{1}{35}

%colle P
 \draw[very thick] (19,1.5) -| (19.5,3.5) -| (22.5,1.5) -| (29,1.5);
%chemin P
\drawpurple{19}{1}{100}
\drawpurple{19}{2}{95}
\drawpurple{19}{3}{90}
\drawpurple{20}{3}{85}
\drawpurple{21}{3}{80}
\drawpurple{22}{3}{75}
\drawpurple{22}{2}{70}
\drawpurple{22}{1}{65}
\drawpurple{23}{1}{60}
\drawpurple{24}{1}{55}
\drawpurple{25}{1}{50}
\drawpurple{26}{1}{45}
\drawpurple{27}{1}{40}
\drawpurple{28}{1}{35}

%colle P
 \draw[very thick] (29,1.5) -| (29.5,3.5) -| (32.5,1.5) -| (34,1.5);
%chemin P
\drawpurple{29}{1}{100}
\drawpurple{29}{2}{95}
\drawpurple{29}{3}{90}
\drawpurple{30}{3}{85}
\drawpurple{31}{3}{80}
\drawpurple{32}{3}{75}
\drawpurple{32}{2}{70}
\drawpurple{32}{1}{65}
\drawpurple{33}{1}{60}

%tuile A
\drawpurple{27}{9}{55}

%rajouter P0 PP et A et le vecteur PA
% nom des tuiles P0=Q0 Q0
\draw (10.5, 0.5) node {\small $P_0$};
\fill[black] (9.5,1.5) circle (0.2);

\draw (21.2, 0.5) node {\small $P_{|P|-1}$};
\fill[black] (19.5,1.5) circle (0.2);

\draw (14.5, 2.6) node {\small $P_i$};
\fill[black] (14.5,1.5) circle (0.2);
%  vecteurs

\draw[very thick,->] (14.5,1.5)--(27.5,9.5);
\draw (23.5, 5.5) node {$\vect{P_iA}$};

\draw (27.5, 10.5) node {$A$};
\path [dotted, draw, thin] (0,0) grid[step=0.33cm] (34,13);
\end{tikzpicture}

a)The bi-pumpable path without redundancy $P$ and the tile $A$.
\vspace{+0.5em}

\begin{tikzpicture}[x=0.33cm,y=0.33cm]

%chemin lien
 \draw[very thick] (14.5,1.5) |- (18.5,7.5) |- (24.5,5.5) |- (27.5,7.5) |- (27.5,9.5);
\drawt{14}{2}
\drawt{14}{3}
\drawyellow{14}{4}{100}
\drawyellow{14}{5}{95}
\drawyellow{14}{6}{90}
\drawyellow{14}{7}{85}
\drawyellow{15}{7}{80}
\drawyellow{16}{7}{75}
\drawyellow{17}{7}{70}
\drawyellow{18}{7}{65}
\drawyellow{18}{6}{60}
\drawyellow{18}{5}{60}
\drawyellow{19}{5}{55}
\drawyellow{20}{5}{50}
\drawyellow{21}{5}{45}
\drawyellow{22}{5}{40}
\drawyellow{23}{5}{35}
\drawyellow{24}{5}{30}
\drawyellow{24}{6}{25}
\drawyellow{24}{7}{100}
\drawt{25}{7}
\drawt{26}{7}
\drawt{27}{7}
\drawt{27}{8}

%figure centrale
%colle P
 \draw[very thick] (9,1.5) -| (9.5,3.5) -| (12.5,1.5) -| (19,1.5);
%chemin P
\drawpurple{9}{1}{100}
\drawpurple{9}{2}{95}
\drawpurple{9}{3}{90}
\drawpurple{10}{3}{85}
\drawpurple{11}{3}{80}
\drawpurple{12}{3}{75}
\drawpurple{12}{2}{70}
\drawpurple{12}{1}{65}
\drawpurple{13}{1}{60}
\drawpurple{14}{1}{55}
\drawpurple{15}{1}{50}
\drawpurple{16}{1}{45}
\drawpurple{17}{1}{40}
\drawpurple{18}{1}{35}

%colle P
 \draw[very thick] (0,3.5) -| (2.5,1.5) -| (9,1.5);
%chemin P
\drawpurple{0}{3}{85}
\drawpurple{1}{3}{80}
\drawpurple{2}{3}{75}
\drawpurple{2}{2}{70}
\drawpurple{2}{1}{65}
\drawpurple{3}{1}{60}
\drawpurple{4}{1}{55}
\drawpurple{5}{1}{50}
\drawpurple{6}{1}{45}
\drawpurple{7}{1}{40}
\drawpurple{8}{1}{35}

%colle P
 \draw[very thick] (19,1.5) -| (19.5,3.5) -| (22.5,1.5) -| (29,1.5);
%chemin P
\drawpurple{19}{1}{100}
\drawpurple{19}{2}{95}
\drawpurple{19}{3}{90}
\drawpurple{20}{3}{85}
\drawpurple{21}{3}{80}
\drawpurple{22}{3}{75}
\drawpurple{22}{2}{70}
\drawpurple{22}{1}{65}
\drawpurple{23}{1}{60}
\drawpurple{24}{1}{55}
\drawpurple{25}{1}{50}
\drawpurple{26}{1}{45}
\drawpurple{27}{1}{40}
\drawpurple{28}{1}{35}

%colle P
 \draw[very thick] (29,1.5) -| (29.5,3.5) -| (32.5,1.5) -| (34,1.5);
%chemin P
\drawpurple{29}{1}{100}
\drawpurple{29}{2}{95}
\drawpurple{29}{3}{90}
\drawpurple{30}{3}{85}
\drawpurple{31}{3}{80}
\drawpurple{32}{3}{75}
\drawpurple{32}{2}{70}
\drawpurple{32}{1}{65}
\drawpurple{33}{1}{60}

 \draw[very thick] (0,9.5) -| (2,9.5);
\drawpurple{0}{9}{40}
\drawpurple{1}{9}{35}

 \draw[very thick] (2,9.5) -| (2.5,11.5) -| (5.5,9.5) -| (12,9.5);
\drawpurple{2}{9}{100}
\drawpurple{2}{10}{95}
\drawpurple{2}{11}{90}
\drawpurple{3}{11}{85}
\drawpurple{4}{11}{80}
\drawpurple{5}{11}{75}
\drawpurple{5}{10}{70}
\drawpurple{5}{9}{65}
\drawpurple{6}{9}{60}
\drawpurple{7}{9}{55}
\drawpurple{8}{9}{50}
\drawpurple{9}{9}{45}
\drawpurple{10}{9}{40}
\drawpurple{11}{9}{35}

 \draw[very thick] (12,9.5) -| (12.5,11.5) -| (15.5,9.5) -| (22,9.5);
\drawpurple{12}{9}{100}
\drawpurple{12}{10}{95}
\drawpurple{12}{11}{90}
\drawpurple{13}{11}{85}
\drawpurple{14}{11}{80}
\drawpurple{15}{11}{75}
\drawpurple{15}{10}{70}
\drawpurple{15}{9}{65}
\drawpurple{16}{9}{60}
\drawpurple{17}{9}{55}
\drawpurple{18}{9}{50}
\drawpurple{19}{9}{45}
\drawpurple{20}{9}{40}
\drawpurple{21}{9}{35}

 \draw[very thick] (22,9.5) -| (22.5,11.5) -| (25.5,9.5) -| (32,9.5);
\drawpurple{22}{9}{100}
\drawpurple{22}{10}{95}
\drawpurple{22}{11}{90}
\drawpurple{23}{11}{85}
\drawpurple{24}{11}{80}
\drawpurple{25}{11}{75}
\drawpurple{25}{10}{70}
\drawpurple{25}{9}{65}
\drawpurple{26}{9}{60}
%tuile A
\drawpurple{27}{9}{55}
\drawpurple{28}{9}{50}
\drawpurple{29}{9}{45}
\drawpurple{30}{9}{40}
\drawpurple{31}{9}{35}

 \draw[very thick] (32,9.5) -| (32.5,11.5) -| (34,11.5);
\drawpurple{32}{9}{100}
\drawpurple{32}{10}{95}
\drawpurple{32}{11}{90}
\drawpurple{33}{11}{85}

%rajouter P0 PP et A et le vecteur PA
% nom des tuiles P0=Q0 Q0
\draw (13, 4.5) node {\small $R_0$};
\fill[black] (14.5,4.5) circle (0.2);

\draw (22.1, 7.8) node {\small $R_{|R|-1}$};
\fill[black] (24.5,7.5) circle (0.2);

\draw (10.5, 0.5) node {\small $P_0$};
\fill[black] (9.5,1.5) circle (0.2);

\draw (21.2, 0.5) node {\small $P_{|P|-1}$};
\fill[black] (19.5,1.5) circle (0.2);

\draw (15.8, 2.6) node {\small $P_i$};
\fill[black] (14.5,1.5) circle (0.2);

%  vecteurs
\draw (27.5, 10.5) node {$A$};
\fill[black] (27.5,9.5) circle (0.2);
\path [dotted, draw, thin] (0,0) grid[step=0.33cm] (34,13);
\draw[very thick,->] (14.5,4.5)--(24.5,7.5);

\end{tikzpicture}

b) The path $\bipump{P}+\vect{P_iA}$ can be assembled and we can consider a path $Q$ from $P_i$ to a tile of the same tile type. This path contains a bi-pumpable path without redundancy  $R$ (which satisfies the hypothesis of the lemma in this case).
\vspace{+0.5em}

\begin{tikzpicture}[x=0.33cm,y=0.33cm]

%chemin lien
 \draw[very thick] (14.5,1.5) |- (18.5,7.5) |- (24.5,5.5) |- (27.5,7.5) |- (27.5,9.5);
\drawt{14}{2}
\drawt{14}{3}
\drawt{14}{4}
\drawt{14}{5}
\drawyellow{14}{6}{100}
\drawyellow{14}{7}{95}
\drawyellow{15}{7}{90}
\drawyellow{16}{7}{85}
\drawyellow{17}{7}{80}
\drawyellow{18}{7}{75}
\drawyellow{18}{6}{70}
\drawyellow{18}{5}{65}
\drawyellow{19}{5}{60}
\drawyellow{20}{5}{55}
\drawyellow{21}{5}{50}
\drawyellow{22}{5}{45}
\drawyellow{23}{5}{40}
\drawyellow{24}{5}{35}
\drawyellow{24}{6}{100}
\drawt{24}{7}
\drawt{25}{7}
\drawt{26}{7}
\drawt{27}{7}
\drawt{27}{8}

%nouveau lien
 \draw[very thick] (4.5,1.5) |- (7.5,7.5) |- (7.5,9.5);
\drawt{4}{2}
\drawt{4}{3}
\drawt{4}{4}
\drawt{4}{5}
\drawyellow{4}{6}{100}
\drawt{4}{7}
\drawt{5}{7}
\drawt{6}{7}
\drawt{7}{7}
\drawt{7}{8}

%figure centrale
%colle P
 \draw[very thick] (9,1.5) -| (9.5,3.5) -| (12.5,1.5) -| (19,1.5);
%chemin P
\drawpurple{9}{1}{100}
\drawpurple{9}{2}{95}
\drawpurple{9}{3}{90}
\drawpurple{10}{3}{85}
\drawpurple{11}{3}{80}
\drawpurple{12}{3}{75}
\drawpurple{12}{2}{70}
\drawpurple{12}{1}{65}
\drawpurple{13}{1}{60}
\drawpurple{14}{1}{55}
\drawpurple{15}{1}{50}
\drawpurple{16}{1}{45}
\drawpurple{17}{1}{40}
\drawpurple{18}{1}{35}

%colle P
 \draw[very thick] (0,3.5) -| (2.5,1.5) -| (9,1.5);
%chemin P
\drawpurple{0}{3}{85}
\drawpurple{1}{3}{80}
\drawpurple{2}{3}{75}
\drawpurple{2}{2}{70}
\drawpurple{2}{1}{65}
\drawpurple{3}{1}{60}
\drawpurple{4}{1}{55}
\drawpurple{5}{1}{50}
\drawpurple{6}{1}{45}
\drawpurple{7}{1}{40}
\drawpurple{8}{1}{35}

%colle P
 \draw[very thick] (19,1.5) -| (19.5,3.5) -| (22.5,1.5) -| (29,1.5);
%chemin P
\drawpurple{19}{1}{100}
\drawpurple{19}{2}{95}
\drawpurple{19}{3}{90}
\drawpurple{20}{3}{85}
\drawpurple{21}{3}{80}
\drawpurple{22}{3}{75}
\drawpurple{22}{2}{70}
\drawpurple{22}{1}{65}
\drawpurple{23}{1}{60}
\drawpurple{24}{1}{55}
\drawpurple{25}{1}{50}
\drawpurple{26}{1}{45}
\drawpurple{27}{1}{40}
\drawpurple{28}{1}{35}

%colle P
 \draw[very thick] (29,1.5) -| (29.5,3.5) -| (32.5,1.5) -| (34,1.5);
%chemin P
\drawpurple{29}{1}{100}
\drawpurple{29}{2}{95}
\drawpurple{29}{3}{90}
\drawpurple{30}{3}{85}
\drawpurple{31}{3}{80}
\drawpurple{32}{3}{75}
\drawpurple{32}{2}{70}
\drawpurple{32}{1}{65}
\drawpurple{33}{1}{60}

 \draw[very thick] (0,9.5) -| (2,9.5);
\drawpurple{0}{9}{40}
\drawpurple{1}{9}{35}

 \draw[very thick] (2,9.5) -| (2.5,11.5) -| (5.5,9.5) -| (12,9.5);
\drawpurple{2}{9}{100}
\drawpurple{2}{10}{95}
\drawpurple{2}{11}{90}
\drawpurple{3}{11}{85}
\drawpurple{4}{11}{80}
\drawpurple{5}{11}{75}
\drawpurple{5}{10}{70}
\drawpurple{5}{9}{65}
\drawpurple{6}{9}{60}
\drawpurple{7}{9}{55}
\drawpurple{8}{9}{50}
\drawpurple{9}{9}{45}
\drawpurple{10}{9}{40}
\drawpurple{11}{9}{35}

 \draw[very thick] (12,9.5) -| (12.5,11.5) -| (15.5,9.5) -| (22,9.5);
\drawpurple{12}{9}{100}
\drawpurple{12}{10}{95}
\drawpurple{12}{11}{90}
\drawpurple{13}{11}{85}
\drawpurple{14}{11}{80}
\drawpurple{15}{11}{75}
\drawpurple{15}{10}{70}
\drawpurple{15}{9}{65}
\drawpurple{16}{9}{60}
\drawpurple{17}{9}{55}
\drawpurple{18}{9}{50}
\drawpurple{19}{9}{45}
\drawpurple{20}{9}{40}
\drawpurple{21}{9}{35}

 \draw[very thick] (22,9.5) -| (22.5,11.5) -| (25.5,9.5) -| (32,9.5);
\drawpurple{22}{9}{100}
\drawpurple{22}{10}{95}
\drawpurple{22}{11}{90}
\drawpurple{23}{11}{85}
\drawpurple{24}{11}{80}
\drawpurple{25}{11}{75}
\drawpurple{25}{10}{70}
\drawpurple{25}{9}{65}
\drawpurple{26}{9}{60}
%tuile A
\drawpurple{27}{9}{55}
\drawpurple{28}{9}{50}
\drawpurple{29}{9}{45}
\drawpurple{30}{9}{40}
\drawpurple{31}{9}{35}

 \draw[very thick] (32,9.5) -| (32.5,11.5) -| (34,11.5);
\drawpurple{32}{9}{100}
\drawpurple{32}{10}{95}
\drawpurple{32}{11}{90}
\drawpurple{33}{11}{85}

%rajouter P0 PP et A et le vecteur PA
% nom des tuiles P0=Q0 Q0
\draw (13, 6.5) node {\small $R_0$};
\fill[black] (14.5,6.5) circle (0.2);

\draw (22.1, 6.8) node {\small $R_{|R|-1}$};
\fill[black] (24.5,6.5) circle (0.2);

\draw (10.5, 0.5) node {\small $P_0$};
\fill[black] (9.5,1.5) circle (0.2);

\draw (21.2, 0.5) node {\small $P_{|P|-1}$};
\fill[black] (19.5,1.5) circle (0.2);

\draw (15.8, 2.6) node {\small $P_i$};
\fill[black] (14.5,1.5) circle (0.2);

%\draw (5.5, 6.5) node {\small $Q'_0-\vect{Q'_0Q'_{|Q'|-1}}=Q'_0-2\vect{Q'_0Q'_{|Q'|-1}}$};
%\fill[black] (4.5,6.5) circle (0.2);

%  vecteurs

\draw (31, 8.5) node {\small $-2\vect{R_0R_{|R|-1}}$};
\draw[very thick,->] (27.5,8.5)--(7.5,8.5);
\draw (8.5, 5.5) node {\small $-\vect{R_0R_{|R|-1}}$};
\draw[very thick,->] (14.5,4.5)--(4.5,4.5);
\draw (27.5, 10.5) node {$A$};
\fill[black] (27.5,9.5) circle (0.2);
\path [dotted, draw, thin] (0,0) grid[step=0.33cm] (34,13);
\end{tikzpicture}

c)If $\vect{P_0P_{|P|-1}}$ is collinear with $\vect{R_0R_{|R|-1}}$ then we can build a new shorter path with the same properties as $Q$.
 
\caption{Illustration of the proof of lemma \ref{biperiodic:otherpath}.Consider a bi-pumpable path $P$ without redundancy such that a tile $A$ has the same tile type as a tile of $P$, then we can find another bi-pumpable path $R$ without redundancy such that $\vect{P_0P_{|P|-1}}$ is not collinear with $\vect{R_0R_{|R|-1}}$.}
\label{fig:trans}
\end{figure}

\end{proof}

\begin{lemma}
\label{alpha:type1}
Let $\mathcal T=(T,\sigma,1)$ be a directed tile assembly system, $\uniterm$ its unique terminal assembly and two bi-pumpable paths $P,Q\in \inuniterm$ without redundancy such that  $\vect{P_0P_{|P|-1}}$ is not collinear with $\vect{Q_0Q_{|Q|-1}}$ then there exists an assembly $\ass$ and two vectors $\vect{u}$ and $\vect{v}$ such that $|\ass| \leq |T|^2$ and $\alpha=\bigcup_{\ell,\ell' \in \Z} (\ass+\ell\vect{u}+\ell'\vect{v})$.
\end{lemma}

\begin{proof}
See Figure \ref{fig:MT1} for an illustration of this reasoning.
Since $\vect{P_0P_{|P|-1}}$ and $\vect{Q_0Q_{|Q|-1}}$ are not collinear then $\bipump{P}$ and $\bipump{Q}$ intersect at least one time and only a finite number of times. Without loss of generality, suppose that $P_0=Q_0$ is the only intersection between $\bipump{P}$ and $\pump{Q}$, \emph{i.e.} the last intersection between $P$ and $Q$ according to $Q$. Let $\vect{u}=\vect{P_0P_{|P|-1}}$. Now consider $i=\min\{j>0: \type{Q_j}=\type{P_k}\text{ for some } 0\leq k \leq |P|-1\}$ (note that $i$ is correctly defined since $\type{Q_{|Q|-1}}=\type{Q_0}=\type{P_0}$). Let $0\leq j <|P|-1$ such that the type of $P_j$ is the same as $Q_i$ and let $\vect{v}=\vect{P_jQ_i}$. By definition of $Q_0$ and since $i>0$, $Q_i$ is not in $\bipump{P}$ and by lemma \ref{biperiodic:notcol}, $\vect{u}$ and $\vect{v}$ are not collinear. Moreover, by definition of $i$ and $j$, we have $\type{Q_i}=\type{P_j}$ and by lemma \ref{biperiodic:translation}, $\bipump{P}+\vect{v}$ is in $\inuniterm$ and by lemma \ref{lem:redun:nointer}, it does not intersect with $\bipump{P}$. Consider the path $$R=P_j,P_{j-1},\ldots,P_0,Q_1,Q_2,\ldots, Q_i.$$ By definition of $i$ and $j$, the path $R$ is without redundancy and satisfies the hypothesis of lemma \ref{lem:redun:pump}, then $R\in \inuniterm$ is bi-pumpable. Remind that $\vect{v}=\vect{R_0R_{|R|-1}}$ and then by corollary \ref{cor:periodic}, $\alpha$ is $\vect{v}$-periodic. Let $$Q'=Q_{0,1,\ldots,i}.$$ Let $P'$ be the subpath of $\bipump{P}$ such that $P'_0=P_j=Q_i-\vect{v}$, $P'_{|P'|-1}=P_j+\vect{u}=Q_i-\vect{v}+\vect{u}$.
Since $\bipump{P}+\vect{v}$ does not intersect with $\bipump{P}$, then $P$ does not intersect with $P'+\vect{v}$. Similarly, since $\vect{u}$ and $\vect{Q_{0}Q_{|Q|-1}}$ are not collinear, $Q$ does not intersect with $Q+\vect{u}$ which implies that $Q'$ does not intersect with $Q'+\vect{u}$. 
By definition of $Q'$, the only intersection between $Q'$ and $P$ is $Q'_0=P_0$. Similarly the only intersection between $Q'+\vect{u}$ and $P$ is $Q'_0+\vect{u}=P_{|P|-1}$. The only intersection between $Q'$ and $P'+\vect{v}$ is $Q'_{|Q'|-1}=P'_{0}+\vect{v}$ by definition of $Q'_{|Q'|-1}=Q_i$ (and since $P'$ and $Q$ shares only one common tile type) and similarly, the only intersection between $Q'+\vect{u}$ and $P'+\vect{v}$ is $Q'_{|Q'|-1}+\vect{u}=P'_{|P'|-1}+\vect{v}$. Then the binding paths of the four paths $Q'$, $P$, $Q'+\vect{u}$ and $P'+\vect{v}$ form a simple cycle $C$ and let $\ass$ be the finite assembly which is the restriction of $\uniterm$ to the interior of this cycle. Since $Q'$ and $P$ are without redundancy and share only one common tile type then $|Q'|+|T'|$ is bounded by $|T|+2$ and thus the length of $C$ is bounded by $2|T|$ and its area is bounded by $|T|^2$. Then the size of $\ass$ is bounded by $|T|^2$. 
Consider the assembly $\bigcup_{\ell,\ell' \in \Z} (\ass+\ell\vect{u}+\ell'\vect{v})$, this assembly is correctly defined: there are no conflict for any $\ell,\ell' \in Z$ between $\ass$ and $\ass+\ell\vect{u}+\ell'\vect{v}$ since $\uniterm$ is $\vect{u}$-periodic and $\vect{v}$-periodic and this assembly is connected since $Q+\vect{u}$ is in both $\ass$ and $\ass+\vect{u}$ and $P'_0+\vect{v}$ is in both $\ass$ and $\ass+\vect{v}$. Then $\bigcup_{\ell,\ell' \in \Z} (\ass+\ell\vect{u}+\ell'\vect{v})$ is a subassembly of $\uniterm$. Now, for any $(x,y) \in \Z^2$ there exist $\ell$ and $\ell'$ such that $(x,y)+\ell\vect{u}+\ell'\vect{v}$ is in the interior of the simple cycle $C$ and thus for any tile $A$ of $\uniterm$, there are $\ell,\ell'$ and a tile $B$ of $\ass$ such that $B=A+\ell{u}+\ell'{v}$. Thus, $\bigcup_{\ell,\ell' \in \Z} (\ass+\ell\vect{u}+\ell'\vect{v})=\uniterm$.

\begin{figure}
\centering

\begin{minipage}{0.45\linewidth}
\begin{tikzpicture}[x=0.33cm,y=0.33cm]

%colle P
 \draw[very thick] (5.5,5.5) -| (6.5,3.5) -| (3.5,1.5) -| (9.5,5.5) -| (13.5,5.5);
 \draw[thick,dashed] (0,1.5) -| (1.5,5.5) -| (5.5,5.5);
 \draw[thick,dashed] (13.5,5.5) -| (14.5,3.5) -| (11.5,1.5) -| (16,1.5);

%colle Q
 \draw[very thick] (5.5,5.5) -| (5.5,11.5);
 \draw[thick,dashed] (5.5,11.5) -| (5.5,13);
 \draw[thick,dashed] (5.5,5.5) -| (5.5,0);

%chemin P
\drawpurple{5}{5}{100}
\drawpurple{6}{5}{96}
\drawpurple{6}{4}{92}
\drawpurple{6}{3}{88}
\drawpurple{5}{3}{84}
\drawpurple{4}{3}{80}
\drawpurple{3}{3}{76}
\drawpurple{3}{2}{72}
\drawpurple{3}{1}{68}
\drawpurple{4}{1}{64}
\drawpurple{5}{1}{60}
\drawpurple{6}{1}{56}
\drawpurple{7}{1}{52}
\drawpurple{8}{1}{48}
\drawpurple{9}{1}{44}
\drawpurple{9}{2}{40}
\drawpurple{9}{3}{36}
\drawpurple{9}{4}{32}
\drawpurple{9}{5}{28}
\drawpurple{10}{5}{24}
\drawpurple{11}{5}{20}
\drawpurple{12}{5}{16}
\drawpurple{13}{5}{100}

%chemin Q
\drawyellow{5}{6}{100}
\drawpurple{5}{7}{60}
\drawyellow{5}{8}{87}
\drawpurple{5}{9}{84}
\drawyellow{5}{10}{75}
\drawpurple{5}{11}{100}

% nom des tuiles P0=Q0 
\draw (3, 6) node {\small $P_0=Q_0$};
\fill[black] (5.5,5.5) circle (0.2);
\draw (3.5, 11.5) node {\small $Q_{|Q|-1}$};
\fill[black] (5.5,11.5) circle (0.2);
\draw (13.5, 6.5) node {\small $P_{|P|-1}$};
\fill[black] (13.5,5.5) circle (0.2);

% vecteurs
\draw[very thick,->] (0.5,0.5)--(8.5,0.5);
\draw (12, 0.5) node {$\vect{u}=\vect{P_0P_{|P|-1}}$};

\path [dotted, draw, thin] (0,0) grid[step=0.33cm] (16,13);
\end{tikzpicture}

a) The two bi-pumpable paths without redundancy $P$ and $Q$.
\end{minipage} \hfill
\begin{minipage}{0.48\linewidth}
\begin{tikzpicture}[x=0.33cm,y=0.33cm]

%colle P
 \draw[very thick] (5.5,5.5) -| (6.5,3.5) -| (3.5,1.5) -| (9.5,5.5) -| (13.5,5.5);

%colle Q
 \draw[very thick] (5.5,5.5) -| (5.5,7.5);

%chemin P
\drawpurple{5}{5}{100}
\drawpurple{6}{5}{96}
\drawpurple{6}{4}{92}
\drawpurple{6}{3}{88}
\drawpurple{5}{3}{84}
\drawpurple{4}{3}{80}
\drawpurple{3}{3}{76}
\drawpurple{3}{2}{72}
\drawpurple{3}{1}{68}
\drawpurple{4}{1}{64}
\drawpurple{5}{1}{60}
\drawpurple{6}{1}{56}
\drawpurple{7}{1}{52}
\drawpurple{8}{1}{48}
\drawpurple{9}{1}{44}
\drawpurple{9}{2}{40}
\drawpurple{9}{3}{36}
\drawpurple{9}{4}{32}
\drawpurple{9}{5}{28}
\drawpurple{10}{5}{24}
\drawpurple{11}{5}{20}
\drawpurple{12}{5}{16}
\drawpurple{13}{5}{100}

%chemin Q
\drawyellow{5}{6}{100}
\drawpurple{5}{7}{60}

% nom des tuiles P0=Q0 
\draw (3, 5.5) node {\small $P_0=Q'_0$};
\fill[black] (5.5,5.5) circle (0.2);
\draw (5.5, 8.8) node {\small $Q'_{|Q'|-1}$};
\fill[black] (5.5,7.5) circle (0.2);
\draw (13.5, 6.5) node {\small $P_{|P|-1}$};
\fill[black] (13.5,5.5) circle (0.2);

% vecteurs
\draw[very thick,->] (0.5,0.5)--(8.5,0.5);
\draw (12, 0.5) node {$\vect{u}=\vect{P_0P_{|P|-1}}$};

\path [dotted, draw, thin] (0,0) grid[step=0.33cm] (16,13);
\end{tikzpicture}

b) $Q'$ is a prefix of $Q$ which ends where a tile of $Q$ has the same tile type than a tile of $P$ for the first time.
\end{minipage}

\vspace{+0.5em}
\begin{minipage}{0.48\linewidth}
\begin{tikzpicture}[x=0.33cm,y=0.33cm]

%colle P
 \draw[very thick] (5.5,5.5) -| (6.5,3.5) -| (3.5,1.5) -| (9.5,5.5) -| (13.5,5.5);

%colle Q'
 \draw[very thick] (5.5,5.5) -| (5.5,7.5);

%colle P'
 \draw[very thick] (5.5,7.5) -| (9.5,11.5) -| (14.5,9.5) -| (11.5,7.5) -| (13.5,7.5);

%chemin P
\drawpurple{5}{5}{100}
\drawpurple{6}{5}{96}
\drawpurple{6}{4}{92}
\drawpurple{6}{3}{88}
\drawpurple{5}{3}{84}
\drawpurple{4}{3}{80}
\drawpurple{3}{3}{76}
\drawpurple{3}{2}{72}
\drawpurple{3}{1}{68}
\drawpurple{4}{1}{64}
\drawpurple{5}{1}{60}
\drawpurple{6}{1}{56}
\drawpurple{7}{1}{52}
\drawpurple{8}{1}{48}
\drawpurple{9}{1}{44}
\drawpurple{9}{2}{40}
\drawpurple{9}{3}{36}
\drawpurple{9}{4}{32}
\drawpurple{9}{5}{28}
\drawpurple{10}{5}{24}
\drawpurple{11}{5}{20}
\drawpurple{12}{5}{16}
\drawpurple{13}{5}{100}

%chemin P'
\drawpurple{5}{7}{60}
\drawpurple{6}{7}{56}
\drawpurple{7}{7}{52}
\drawpurple{8}{7}{48}
\drawpurple{9}{7}{44}
\drawpurple{9}{8}{40}
\drawpurple{9}{9}{36}
\drawpurple{9}{10}{32}
\drawpurple{9}{11}{28}
\drawpurple{10}{11}{24}
\drawpurple{11}{11}{20}
\drawpurple{12}{11}{16}
\drawpurple{13}{11}{100}
\drawpurple{14}{11}{96}
\drawpurple{14}{10}{92}
\drawpurple{14}{9}{88}
\drawpurple{13}{9}{84}
\drawpurple{12}{9}{80}
\drawpurple{11}{9}{76}
\drawpurple{11}{8}{72}
\drawpurple{11}{7}{68}
\drawpurple{12}{7}{64}
\drawpurple{13}{7}{60}

%chemin Q'
\drawyellow{5}{6}{100}

% nom des tuiles P0=Q0 Q0
\draw (3, 5.5) node {\small $P_0=Q'_0$};
\fill[black] (5.5,5.5) circle (0.2);
\draw (5.5, 8.8) node {\small $P'_0=Q'_{|Q'|-1}$};
\fill[black] (5.5,7.5) circle (0.2);
\draw (13.5, 4.2) node {\small $P_{|P|-1}$};
\fill[black] (13.5,5.5) circle (0.2);
\draw (13.5, 8.5) node {\small $P'_{|P'|-1}$};
\fill[black] (13.5,7.5) circle (0.2);

% 2 vecteurs
\draw[very thick,->] (0.5,0.5)--(8.5,0.5);
\draw (9.5, 0.5) node {$\vect{u}$};
\draw[very thick,->] (0.5,0.5)--(0.5,6.5);
\draw (0.5, 7.5) node {$\vect{v}$};

\path [dotted, draw, thin] (0,0) grid[step=0.33cm] (16,13);
\end{tikzpicture}

c) We grow$\bipump{P}+\vect{v}$ from $Q'_{|Q'|-1}$ to $Q'_{|Q'|-1}+\vect{u}$ and we obtain a path called $P'$.
\end{minipage} \hfill 
\begin{minipage}{0.48\linewidth}
\begin{tikzpicture}[x=0.33cm,y=0.33cm]

% intérieur à colorier
 \fill[red!10!white] (5.5,5.5) -| (6.5,3.5) -| (3.5,1.5) -| (9.5,5.5) -| (13.5,7.5) -| (11.5,9.5) -| (14.5,11.5) -| (9.5,7.5) -| (5.5,5.5);

%colle P
 \draw[very thick] (5.5,5.5) -| (6.5,3.5) -| (3.5,1.5) -| (9.5,5.5) -| (13.5,5.5);

%colle Q'
 \draw[very thick] (5.5,5.5) -| (5.5,7.5);

%colle P'
 \draw[very thick] (5.5,7.5) -| (9.5,11.5) -| (14.5,9.5) -| (11.5,7.5) -| (13.5,5.5);

%colle Q'+v
 \draw[very thick] (13.5,5.5) -| (13.5,7.5);

%chemin P
\drawpurple{5}{5}{100}
\drawpurple{6}{5}{96}
\drawpurple{6}{4}{92}
\drawpurple{6}{3}{88}
\drawpurple{5}{3}{84}
\drawpurple{4}{3}{80}
\drawpurple{3}{3}{76}
\drawpurple{3}{2}{72}
\drawpurple{3}{1}{68}
\drawpurple{4}{1}{64}
\drawpurple{5}{1}{60}
\drawpurple{6}{1}{56}
\drawpurple{7}{1}{52}
\drawpurple{8}{1}{48}
\drawpurple{9}{1}{44}
\drawpurple{9}{2}{40}
\drawpurple{9}{3}{36}
\drawpurple{9}{4}{32}
\drawpurple{9}{5}{28}
\drawpurple{10}{5}{24}
\drawpurple{11}{5}{20}
\drawpurple{12}{5}{16}
\drawpurple{13}{5}{100}

%chemin Q'
\drawyellow{5}{6}{100}

%chemin P'
\drawpurple{5}{7}{60}
\drawpurple{6}{7}{56}
\drawpurple{7}{7}{52}
\drawpurple{8}{7}{48}
\drawpurple{9}{7}{44}
\drawpurple{9}{8}{40}
\drawpurple{9}{9}{36}
\drawpurple{9}{10}{32}
\drawpurple{9}{11}{28}
\drawpurple{10}{11}{24}
\drawpurple{11}{11}{20}
\drawpurple{12}{11}{16}
\drawpurple{13}{11}{100}
\drawpurple{14}{11}{96}
\drawpurple{14}{10}{92}
\drawpurple{14}{9}{88}
\drawpurple{13}{9}{84}
\drawpurple{12}{9}{80}
\drawpurple{11}{9}{76}
\drawpurple{11}{8}{72}
\drawpurple{11}{7}{68}
\drawpurple{12}{7}{64}
\drawpurple{13}{7}{60}

%chemin Q'+v
\drawyellow{13}{6}{100}

% nom des tuiles P0=Q0 Q0
\draw (3, 5.5) node {\small $P_0=Q'_0$};
\fill[black] (5.5,5.5) circle (0.2);

% 2 vecteurs
\draw[very thick,->] (0.5,0.5)--(8.5,0.5);
\draw (9.5, 0.5) node {$\vect{u}$};
\draw[very thick,->] (0.5,0.5)--(0.5,6.5);
\draw (0.5, 7.5) node {$\vect{v}$};

\path [dotted, draw, thin] (0,0) grid[step=0.33cm] (16,13);
\end{tikzpicture}

d) $Q'+\vect{u}$ binds to $P_{|P|-1}$ and $P'_{|P'|-1}+\vect{v}$ and we obtain a simple cycle which delimits a finite area of the $2D$ plane.
\end{minipage}

\vspace{+0.5em}

\begin{tikzpicture}[x=0.33cm,y=0.33cm]

% intérieur à colorier
 \fill[red!10!white] (14.5,5.5) -| (15.5,3.5) -| (12.5,1.5) -| (18.5,5.5) -| (22.5,7.5) -| (20.5,9.5) -| (23.5,11.5) -| (18.5,7.5) -| (14.5,5.5);

\draw[very thick] (6.5,0) -| (6.5,1.5);
\drawyellow{6}{0}{100}

\draw[very thick] (14.5,0) -| (14.5,1.5);
\drawyellow{14}{0}{100}

\draw[very thick] (22.5,0) -| (22.5,1.5);
\drawyellow{22}{0}{100}

\draw[very thick] (30.5,0) -| (30.5,1.5);
\drawyellow{30}{0}{100}

\draw[very thick] (6.5,11.5) -| (6.5,13);
\drawyellow{6}{12}{100}

\draw[very thick] (14.5,11.5) -| (14.5,13);
\drawyellow{14}{12}{100}

\draw[very thick] (22.5,11.5) -| (22.5,13);
\drawyellow{22}{12}{100}

\draw[very thick] (30.5,11.5) -| (30.5,13);
\drawyellow{30}{12}{100}

%figure centrale
%colle P
 \draw[very thick] (14.5,5.5) -| (15.5,3.5) -| (12.5,1.5) -| (18.5,5.5) -| (22.5,5.5);
%colle Q'
 \draw[very thick] (14.5,5.5) -| (14.5,7.5);
%colle P'
 \draw[very thick] (14.5,7.5) -| (18.5,11.5) -| (23.5,9.5) -| (20.5,7.5) -| (22.5,5.5);
%colle Q'+v
 \draw[very thick] (22.5,5.5) -| (22.5,7.5);
%chemin P
\drawpurple{14}{5}{100}
\drawpurple{15}{5}{96}
\drawpurple{15}{4}{92}
\drawpurple{15}{3}{88}
\drawpurple{14}{3}{84}
\drawpurple{13}{3}{80}
\drawpurple{12}{3}{76}
\drawpurple{12}{2}{72}
\drawpurple{12}{1}{68}
\drawpurple{13}{1}{64}
\drawpurple{14}{1}{60}
\drawpurple{15}{1}{56}
\drawpurple{16}{1}{52}
\drawpurple{17}{1}{48}
\drawpurple{18}{1}{44}
\drawpurple{18}{2}{40}
\drawpurple{18}{3}{36}
\drawpurple{18}{4}{32}
\drawpurple{18}{5}{28}
\drawpurple{19}{5}{24}
\drawpurple{20}{5}{20}
\drawpurple{21}{5}{16}
\drawpurple{22}{5}{100}
%chemin Q'
\drawyellow{14}{6}{100}
%chemin P'
\drawpurple{14}{7}{60}
\drawpurple{15}{7}{56}
\drawpurple{16}{7}{52}
\drawpurple{17}{7}{48}
\drawpurple{18}{7}{44}
\drawpurple{18}{8}{40}
\drawpurple{18}{9}{36}
\drawpurple{18}{10}{32}
\drawpurple{18}{11}{28}
\drawpurple{19}{11}{24}
\drawpurple{20}{11}{20}
\drawpurple{21}{11}{16}
\drawpurple{22}{11}{100}
\drawpurple{23}{11}{96}
\drawpurple{23}{10}{92}
\drawpurple{23}{9}{88}
\drawpurple{22}{9}{84}
\drawpurple{21}{9}{80}
\drawpurple{20}{9}{76}
\drawpurple{20}{8}{72}
\drawpurple{20}{7}{68}
\drawpurple{21}{7}{64}
\drawpurple{22}{7}{60}
%chemin Q'+v
\drawyellow{22}{6}{100}

%colle P
 \draw[very thick] (6.5,5.5) -| (7.5,3.5) -| (4.5,1.5) -| (10.5,5.5) -| (14.5,5.5);
%colle Q'
 \draw[very thick] (6.5,5.5) -| (6.5,7.5);
%colle P'
 \draw[very thick] (6.5,7.5) -| (10.5,11.5) -| (15.5,9.5) -| (12.5,7.5) -| (14.5,7.5);
%chemin P
\drawpurple{6}{5}{100}
\drawpurple{7}{5}{96}
\drawpurple{7}{4}{92}
\drawpurple{7}{3}{88}
\drawpurple{6}{3}{84}
\drawpurple{5}{3}{80}
\drawpurple{4}{3}{76}
\drawpurple{4}{2}{72}
\drawpurple{4}{1}{68}
\drawpurple{5}{1}{64}
\drawpurple{6}{1}{60}
\drawpurple{7}{1}{56}
\drawpurple{8}{1}{52}
\drawpurple{9}{1}{48}
\drawpurple{10}{1}{44}
\drawpurple{10}{2}{40}
\drawpurple{10}{3}{36}
\drawpurple{10}{4}{32}
\drawpurple{10}{5}{28}
\drawpurple{11}{5}{24}
\drawpurple{12}{5}{20}
\drawpurple{13}{5}{16}
\drawpurple{14}{5}{100}
%chemin P'
\drawpurple{6}{7}{60}
\drawpurple{7}{7}{56}
\drawpurple{8}{7}{52}
\drawpurple{9}{7}{48}
\drawpurple{10}{7}{44}
\drawpurple{10}{8}{40}
\drawpurple{10}{9}{36}
\drawpurple{10}{10}{32}
\drawpurple{10}{11}{28}
\drawpurple{11}{11}{24}
\drawpurple{12}{11}{20}
\drawpurple{13}{11}{16}
\drawpurple{14}{11}{100}
\drawpurple{15}{11}{96}
\drawpurple{15}{10}{92}
\drawpurple{15}{9}{88}
\drawpurple{14}{9}{84}
\drawpurple{13}{9}{80}
\drawpurple{12}{9}{76}
\drawpurple{12}{8}{72}
\drawpurple{12}{7}{68}
\drawpurple{13}{7}{64}
\drawpurple{14}{7}{60}
%chemin Q
\drawyellow{6}{6}{100}

%colle P
 \draw[very thick] (0,1.5) -| (2.5,5.5) -| (6.5,5.5);
%colle P'
 \draw[very thick] (0,7.5) -| (2.5,11.5) -| (7.5,9.5) -| (4.5,7.5) -| (6.5,7.5);
%chemin P
\drawpurple{0}{1}{52}
\drawpurple{1}{1}{48}
\drawpurple{2}{1}{44}
\drawpurple{2}{2}{40}
\drawpurple{2}{3}{36}
\drawpurple{2}{4}{32}
\drawpurple{2}{5}{28}
\drawpurple{3}{5}{24}
\drawpurple{4}{5}{20}
\drawpurple{5}{5}{16}
\drawpurple{6}{5}{100}
%chemin P'
\drawpurple{0}{7}{52}
\drawpurple{1}{7}{48}
\drawpurple{2}{7}{44}
\drawpurple{2}{8}{40}
\drawpurple{2}{9}{36}
\drawpurple{2}{10}{32}
\drawpurple{2}{11}{28}
\drawpurple{3}{11}{24}
\drawpurple{4}{11}{20}
\drawpurple{5}{11}{16}
\drawpurple{6}{11}{100}
\drawpurple{7}{11}{96}
\drawpurple{7}{10}{92}
\drawpurple{7}{9}{88}
\drawpurple{6}{9}{84}
\drawpurple{5}{9}{80}
\drawpurple{4}{9}{76}
\drawpurple{4}{8}{72}
\drawpurple{4}{7}{68}
\drawpurple{5}{7}{64}
\drawpurple{6}{7}{60}

%figure centrale
%colle P
 \draw[very thick] (22.5,5.5) -| (23.5,3.5) -| (20.5,1.5) -| (26.5,5.5) -| (30.5,5.5);
%colle P'
 \draw[very thick] (22.5,7.5) -| (26.5,11.5) -| (31.5,9.5) -| (28.5,7.5) -| (30.5,7.5);
%colle Q'+v
 \draw[very thick] (30.5,5.5) -| (30.5,7.5);
%chemin P
\drawpurple{22}{5}{100}
\drawpurple{23}{5}{96}
\drawpurple{23}{4}{92}
\drawpurple{23}{3}{88}
\drawpurple{22}{3}{84}
\drawpurple{21}{3}{80}
\drawpurple{20}{3}{76}
\drawpurple{20}{2}{72}
\drawpurple{20}{1}{68}
\drawpurple{21}{1}{64}
\drawpurple{22}{1}{60}
\drawpurple{23}{1}{56}
\drawpurple{24}{1}{52}
\drawpurple{25}{1}{48}
\drawpurple{26}{1}{44}
\drawpurple{26}{2}{40}
\drawpurple{26}{3}{36}
\drawpurple{26}{4}{32}
\drawpurple{26}{5}{28}
\drawpurple{27}{5}{24}
\drawpurple{28}{5}{20}
\drawpurple{29}{5}{16}
\drawpurple{30}{5}{100}
%chemin P'
\drawpurple{22}{7}{60}
\drawpurple{23}{7}{56}
\drawpurple{24}{7}{52}
\drawpurple{25}{7}{48}
\drawpurple{26}{7}{44}
\drawpurple{26}{8}{40}
\drawpurple{26}{9}{36}
\drawpurple{26}{10}{32}
\drawpurple{26}{11}{28}
\drawpurple{27}{11}{24}
\drawpurple{28}{11}{20}
\drawpurple{29}{11}{16}
\drawpurple{30}{11}{100}
\drawpurple{31}{11}{96}
\drawpurple{31}{10}{92}
\drawpurple{31}{9}{88}
\drawpurple{30}{9}{84}
\drawpurple{29}{9}{80}
\drawpurple{28}{9}{76}
\drawpurple{28}{8}{72}
\drawpurple{28}{7}{68}
\drawpurple{29}{7}{64}
\drawpurple{30}{7}{60}
%chemin Q'+v
\drawyellow{30}{6}{100}

%figure centrale
%colle P
 \draw[very thick] (30.5,5.5) -| (31.5,3.5) -| (28.5,1.5) -| (34,1.5);
%colle P'
 \draw[very thick] (30.5,7.5) -| (34,7.5);
%chemin P
\drawpurple{30}{5}{100}
\drawpurple{31}{5}{96}
\drawpurple{31}{4}{92}
\drawpurple{31}{3}{88}
\drawpurple{30}{3}{84}
\drawpurple{29}{3}{80}
\drawpurple{28}{3}{76}
\drawpurple{28}{2}{72}
\drawpurple{28}{1}{68}
\drawpurple{29}{1}{64}
\drawpurple{30}{1}{60}
\drawpurple{31}{1}{56}
\drawpurple{32}{1}{52}
\drawpurple{33}{1}{48}
%chemin P'
\drawpurple{30}{7}{60}
\drawpurple{31}{7}{56}
\drawpurple{32}{7}{52}
\drawpurple{33}{7}{48}

\path [dotted, draw, thin] (0,0) grid[step=0.33cm] (34,13);
\end{tikzpicture}

e)This cycle tiles the $2D$ plane.

\caption{Illustration of the proof of lemma \ref{alpha:type1}. We consider two bi-pumpable paths without redundancy $P$ and $Q$ such that $\vect{P_0P_{|P|-1}}$ is not collinear with $\vect{Q_0Q_{|Q|-1}}$ and we show that the $2D$ plane can be filled by these paths.}
\label{fig:MT1}
\end{figure}

\end{proof}

Now we can prove the main theorem which describes the bi-periodic terminal assemblies.

\begin{theorem}[Description of the bi-periodic terminal assemblies]
\label{the:MT1bis}
Consider a tile assembly system $\mathcal{T}=(T,\sigma,1)$ whose terminal assembly $\uniterm$ is bi-periodic. Then there exist two non-collinear vectors $\vect{u}$ and $\vect{v}$ and an assembly $\ass$ of complexity $0$ whose size is bounded by $|T|^2$ such that $$\uniterm=\bigcup_{\ell,\ell' \in \Z} (\ass+\ell\vect{u}+\ell'\vect{v}).$$
\end{theorem}

\begin{proof}
Since $\uniterm$ is bi-periodic by lemma \ref{lem:exist:bipump}, there exists $P\in \inuniterm$ which is bi-pumpable. By corollary \ref{cor:exists:without}, there exists a path without redundancy $P'\in \inuniterm$ which is bi-pumpable. Since $\uniterm$ is bi-pumpable there exists a vector $\vect{w}$ such that $\uniterm$ is $\vect{w}$-periodic and $\vect{w}$ is not collinear with $\vect{P'_0P'_{|P'|-1}}$. Thus $P_0+\vect{w}$ is a tile of $\uniterm$ and is not a tile of $\bipump{P}$ and by lemma \ref{biperiodic:otherpath}, there exists a path without redundancy $Q\in \inuniterm$ which is bi-pumpable and such that $\vect{P'_0P'_{|P'|-1}}$ and $\vect{Q_0Q_{|Q|-1}}$ are not collinear. By lemma \ref{alpha:type1}, there exists a finite assembly $\ass$ and two vectors $\vect{u}$ and $\vect{v}$ such that $\alpha=\bigcup_{\ell,\ell' \in \Z} (\ass+\ell\vect{u}+\ell'\vect{v})$. Moreover the size of $\ass$ is bounded by $|T|^2$ and the complexity of a finite assembly is $0$ by definition.
\end{proof}

\subsubsection{Ordering the bi-pumpable paths}

%\begin{fact}
%if $Q>P$ then either $|Q|<|P|$ or there exists $0\leq i \leq |Q|-1$ such that $\pos{Q_i}$ is not in $\domain{\bipump{P}}$. 
%\end{fact}
%
\begin{lemma}
\label{lem:arcOnBiperiodic}
Consider a tile assembly system $\mathcal{T}=(T,\sigma,1)$ whose terminal assembly $\uniterm$ is simply periodic and a bi-pumpable path $P\in \inuniterm$ without redundancy. If an an arc $Q$ grows in the left side of $P$ then there exists a bi-pumpable path $R\in\inuniterm$ without redundancy such that $R>P$.
\end{lemma}

\begin{proof}
Firstly, if $Q$ satisfies the hypothesis of lemma \ref{lem:redun:pump}  (note that by definition of an arc both extremities of $Q$ are in $\bipump{P}$) then there exists $0\leq i <j \leq |Q|-1$ such that $R=Q_{i,i+1,\ldots,j }$ is a path without redundancy which is bi-pumpable. If $\vect{P_0P_{|P|-1}}$ and $\vect{R_0R_{|R|-1}}$ are not collinear then $\uniterm$ is bi-periodic (by corollary \ref{cor:periodic}) which contradicts the hypothesis of the lemma. By lemma \ref{biperiodic:sameperiod}, $\vect{P_0P_{|P|-1}}=s\vect{R_0R_{|R|-1}}$ with $s \in \{-1,1\}$ then since $Q$ grows in the left side of $\bipump{P}$ then $R$ and $\bipump{R}$ are in $\leftside{\bipump{P}}$ and then $R\geq P$ Moreover, $\bipump{P}= \bipump{R}$ would contradicts the definition of an arc and then $R>P$. 

Otherwise, $Q$ does not contains a path without redundancy. Then $Q$ and $Q+\vect{P_0P_{|P|-1}}$ cannot intersect otherwise then would exist $0\leq i,j \leq |Q|-1$ such that $Q_i=Q_j+\vect{P_0P_{|P|-1}}$ with $i\neq j $ (since $\vect{P_0P_{|P|-1}}$ is not null). Then, $\type{Q_i}=\type{Q_j}$ and $Q$ would satisfy the hypothesis of lemma \ref{lem:redun:pump} and we would be in the previous case. Since $Q$ does not intersect $Q+\vect{P_0P_{|P|-1}}$ then the width of $Q$ is less than $|P|-1$. Then, without loss of generality we can suppose that there exists $0\leq i <j \leq |P|-1$ such that $Q$ starts in $P_i$ and ends in $P_j$. Consider the path 
$$R=P_{0,1,\ldots,i}Q_{1,2, \ldots, |Q|-2}P_{j,j+1,\ldots |P|-1}.$$ 
By definition of an arc $R\neq P$. If $R$ is not without redundancy then there exist $0\leq k \leq |P|-1$ and $1\leq  k' \leq |Q|-2$ such that $\type{P_k}=\type{Q_{k'}}$. By definition of an arc $Q_{k'}$ is not a tile of $\bipump{P}$ and by lemma \ref{biperiodic:otherpath} and corollary \ref{cor:periodic}, $\uniterm$ is bi-periodic which contradicts the hypothesis of the lemma. By lemma \ref{lem:redun:pump}, $R$ is bi-pumpable and by definition of $R$ and $Q$, $R$ is in $\leftside{\bipump{P}}$ and $\vect{R_0R_{|R|-1}}=\vect{P_0P_{|P|-1}}$ then $\bipump{R}$ is in $\leftside{\bipump{P}}$ and thus $R\geq P$. Finally, since $R\neq P$ then $R>Q$.
\end{proof}

\begin{lemma}
\label{lem:Comparaison}
Consider a tile assembly system $\mathcal{T}=(T,\sigma,1)$ whose terminal assembly $\uniterm$ is simply periodic and two bi-pumpable paths $P$ and $Q$ of 
$\inuniterm$ such that $P$ is without redundancy then either:
\begin{itemize}
\item $P\geq Q$;
\item there exists a bi-pumpable path $R\in \inuniterm$ without redundancy such that $R>P$.
\end{itemize}
%Consider a tile assembly system $\mathcal{T}=(T,\sigma,1)$ whose terminal assembly $\uniterm$ is simply periodic and two bi-pumpable paths $P$ and $Q$ without redundancy of 
%$\inuniterm$ then either:
%\begin{itemize}
%\item $P<Q$ or $P>Q$ or $\bipump{P}=\bipump{Q}$;
%\item there exists a bi-pumpable path $R\in \inuniterm$ without redundancy such that $R>P$.
%\end{itemize}
\end{lemma}

\begin{proof} 
If $\bipump{P}$ does not intersect $\bipump{Q}$ then either $P>Q$ or $P<Q$. In the second case by lemma \ref{lem:redun:exists} and \ref{lem:redun:pump}, there exist $0\leq i <j \leq |Q|$ such that $R=Q_{i,i+1,\ldots,j}$ is bi-pumpable without redundancy. By corollary \ref{cor:periodic}, $\vect{P_0P_{|P|-1}}$ and $\vect{Q_iQ_{j}}$ are collinear otherwise $\alpha$ would not be simply periodic. Moreover, by lemma \ref{biperiodic:sameperiod}, $\vect{P_0P_{|P|-1}}=s\vect{Q_0Q_{|Q|-1}}$ for some $s\in \{-1,1\}$. Since $\bipump{P}$ does not intersect $\bipump{Q}$ then $\bipump{R}$ is in the left side of $\bipump{P}$ (since $P<Q$) and the lemma is true.

Otherwise, $\bipump{P}$ and $\bipump{Q}$ intersect. By corollary \ref{cor:periodic}, $\vect{P_0P_{|P|-1}}$ and $\vect{Q_0Q_{|Q|-1}}$ are collinear otherwise $\alpha$ would not be simply periodic. Moreover, since $P$ is without redundancy then by lemma~\ref{biperiodic:notcol}, $\vect{Q_0Q_{|Q|-1}}=\ell\vect{P_0P_{|P|-1}}$ for some $\ell \in \N^*$. Thus, there are an infinity of intersections between $\bipump{P}$ and $\bipump{Q}$. Then either $P \geq Q$ or then there exist $ i < j$ such that $\bipump{Q}_{i,i+1,\ldots,j}$ is an arc in the left side of $\bipump{P}$ and by lemma \ref{lem:arcOnBiperiodic} there exists a bi-pumpable path $R>P$ of $\inuniterm$.
\end{proof}

\begin{lemma}
\label{lem:max} 
Consider a tile assembly system $\mathcal{T}=(T,\sigma,1)$ whose terminal assembly $\uniterm$ is simply periodic. Then, there exist a path $P^+$ (resp. $P^-$) without redundancy which is bi-pumpable and maximum (resp. minimum). Moreover, $\vect{P^+_0P^+_{|P^+|-1}}=\vect{P^-_0P^-_{|P^-|-1}}$.
\end{lemma}

\begin{proof}
Since $\uniterm$ is simply periodic then by lemma \ref{lem:exist:bipump}, there exists a bi-pumpable path $P^{(0)}$ in $\inuniterm$ and by corollary \ref{cor:exists:without}, we can assume that $P^{(0)}$ is without redundancy. By lemma \ref{lem:Comparaison}, either $P^{(0)}$ is maximum or there exists another bi-pumpable path $P^{(1)}$  in $\inuniterm$ without redundancy such that $P^{(1)}>P^{(0)}$. Two cases may occur: either $|P^{(1)}|< |P^{(0)}|$ or $|P^{(1)}|\geq |P^{(0)}|$. In the second case, since $\bipump{P^{(1)}} \neq \bipump{P^{(0)}}$ there exists $0\leq i_1 \leq |P^{(1)}|-1$ such that $P^{(1)}_{i_1}$ is not a tile of  $P^{(0)}$. For the sake of contradiction, suppose that there more than $|T|^2$ bi-pumpable paths without redundancy of $\inuniterm$ such that $P^{(0)}<P^{(1)}<P^{(2)}<...<P^{(|T|^2)}$. 
Since the length of a path without redundancy is at least $2$ and less than $|T|+1$ then the case where for some $0\leq i \leq |T|^2$, $P^{(i+1)}<P^{(i)}$ occurs at most $|T|$ times consecutively. Since we have $|T|^2$ paths, there exist $i<j$ such that $P^{(i)}\leq P^{(j-1)} < P^{(j)}$ and there exists $0\leq k \leq |P^{(j)}|$ such that the tile $P^{(j)}_k$ is in the left side of $\bipump{P^{(j-1)}}$, is not a tile of $\bipump{P^{(j-1)}}$ and has the same tile type as a tile of $P^{(i)}$. Since $\bipump{P^{(i)}}$ is in the right side of $\bipump{P^{(j-1)}}$, then the tile $P^{(j)}_k$ is also on the left side of $\bipump{P^{(i)}}$ and is not a tile of $\bipump{P^{(i)}}$. Thus by lemmas \ref{biperiodic:notcol} and \ref{biperiodic:otherpath} and corollary \ref{cor:exists:without}, $\uniterm$ is bi-periodic which contradicts the hypothesis of the lemma. Thus there exists a maximal bi-pumpable path $P^+$ without redundancy and by lemma \ref{lem:Comparaison}, $P^+$ is also maximum. A similar reasoning shows that there is a minimum path $P^-$. By lemma \ref{cor:periodic}, $\vect{P^+_0P^+_{|P^+|-1}}$ and $\vect{P^-_0P^-_{|P^-|-1}}$ are collinear otherwise $\uniterm$ would not be simply periodic. By lemma \ref{biperiodic:sameperiod}, $\vect{P^+_0P^+_{|P^+|-1}}=s\vect{P^-_0P^-_{|P^-|-1}}$ for some $s\in \{-1,1\}$. If $s=-1$ then we can consider $\reverse{P^+}$ and the lemma is true.

\end{proof}

We introduce now the notion of $\vect{v}$-self-avoiding path, such a path cannot intersect with its translation by $\ell\vect{v}$ for some $\ell \in \N^*$. This notion will be useful in the next section in order to study the simply pumpable path (these paths are called combs in \cite{Doty-2011}).
.

\begin{definition}
Consider a directed tile assembly system $\mathcal{T}=(T,\sigma,1)$ whose terminal assembly is $\uniterm$ and a non null vector $\vect{v}$. A path $P$ is $\vect{v}$-self-avoiding if for any $\ell \in \N^*$, $P$ does not intersect with $P+\ell\vect{v}$ and there exists $L\in \N$ such that for all $\ell\geq L$, $P+\ell\vect{v} \in \inuniterm$.
%Consider a tile assembly system $\mathcal{T}=(T,\sigma,1)$ and a non null vector $\vect{v}$, a path $P$ is $\vect{v}$-self-avoiding if for any $\ell \in \N$, $P$ does not intersect with $Q+\ell\vect{v}$ where $Q$ is either $P$ or a path growing on $P$ at $P_i$ with $i\geq 1$.
%and any path growing on $P$ at $P_i$ with $i>0$ is also $\vect{v}$-self-avoiding. 
\end{definition}

\begin{definition}
Consider a directed tile assembly system $\mathcal{T}=(T,\sigma,1)$ and two non null, non collinear vectors $\vect{u}$ and $\vect{v}$. A path $P$ is $(\vect{u},\vect{v})$-self-avoiding if for any $\ell,\ell' \in \Z$ such that $(\ell,\ell')\neq (0,0)$, $P$ does not intersect with $P+\ell\vect{u}+\ell'\vect{v}$.

%Consider a tile assembly system $\mathcal{T}=(T,\sigma,1)$ and two non null, non collinear vectors $\vect{u}$, $\vect{v}$. A path $P$ is $(\vect{u},\vect{v})$-self-avoiding if for all $\ell,\ell'\in Z$, the path $P$ does not intersect with $P+\ell\vect{u}+\ell'\vect{v}$ and any path growing on $P$ at $P_i$ with $i>0$ is also $(\vect{u},\vect{v})$-self-avoiding.
\end{definition}

We conclude this section, with the first half of the second main theorem \ref{main:periodic}.

\begin{lemma}
\label{lem:firsthalf}
Consider a tile assembly system $\mathcal{T}=(T,\sigma,1)$ whose terminal assembly $\uniterm$ is simply periodic. There exist a finite assembly $\ass$ and two bi-pumpable paths $P^+$ and $P^-$ such that $\vect{v}=\vect{P^+_0P^+_{|P^+|-1}}=\vect{P^-_0P^-_{|P^-|-1}}$ and the restriction of $\uniterm$ to $\rightside{\bipump{P^+}}\cap \leftside{\bipump{P^-}}$ is $$\bigcup_{\ell\in Z}(\ass+\ell\vect{v}).$$ Moreover, $\asm{P^+}$ and $\asm{P^-}$ are subassembly of $\ass$, the size of $\ass$ is bounded by $|T|^2$ and no arc grows on the left (resp. right) side of $\bipump{P^+}$ (resp. $\bipump{P^-}$) and any path growing on the left (resp. right) side of $\bipump{P^+}$ (resp. $\bipump{P^-}$) is $\vect{v}$-self-avoiding.
%If there is a bi-pumpable path $P$ without redondancy in $\alpha$ and there exists a tile $A$ in $\alpha$ such that $type(A)=type(P_i)$ then either $\vect{P_iA}$ is not colinear or there exists an assembly $Q$ such that $Q$ fits in a box of size $\ell\times L$ with $\ell+L\leq |T|$ and $\alpha=\bigcup\bigcup k\vect{u}Q$.
\end{lemma}

\begin{proof}
Since $\uniterm$ is simply periodic, consider the paths $P^+$ and $P^-$ which satisfy the hypothesis of lemma \ref{lem:max} and let $\vect{v}=\vect{P^+_0P^+_{|P^+|-1}}=\vect{P^-_0P^-_{|P^-|-1}}$. %If an arc grows in the left side of $\bipump{P^+}$ (resp. the right side of $\bipump{P^-}$) then by lemma \ref{lem:arcOnBiperiodic}, it contradicts the maximality of $P^+$ (resp. minimality of $P^-$). 
Consider a path $P$ which grows on $\bipump{P^+}$ in its left side. By lemma \ref{lem:bipump:seed}, for all $\ell \in \Z$ the path $P+\ell\vect{v}$ is in $\inuniterm$ (we can consider that $P_0+\ell\vect{v}$ is the seed). If for some $\ell \in \N^*$ $P$ and $P+\ell\vect{v}$ intersect then they have to agree since they are both in $\inuniterm$ and they assemble an arc in the left side of $\bipump{P^+}$ which by lemma \ref{lem:arcOnBiperiodic} contradicts the maximality of $P^+$. Then, $P$ is $\vect{v}$-self-avoiding. Similarly any path growing in the right side of $\bipump{P^-}$ is $\vect{v}$-self-avoiding.

Now, we need analyze the restrictions of $\uniterm$ to $\rightside{\bipump{P^+}}\cap \leftside{\bipump{P^-}}$. First, if $\bipump{P^+}=\bipump{P^-}$, then the lemma is true with $\ass=P^+$. Secondly, if $\bipump{P^+}$ and $\bipump{P^-}$ intersect then we define $\beta$ as the restriction of $\uniterm$ to the interior of $\bipump{P^+}$ and $\bipump{P^-}$ (see appendix \ref{appB:curves} for this definition in this special case). Otherwise, since $\bipump{P^+}$ and $\bipump{P^-}$ are in $\inuniterm$ then there exists a path $Q$ such that the only intersection between $P^+$ (resp. $P^-$) and $Q$ is $Q_0$ (resp. $Q_{|Q|}-1$). Without loss of generality we suppose that $Q$ is the path of minimal length which satisfies these hypothesis and that $Q_0=P^+_0$ and $Q_0=P^-_0$. Note that $Q$ is in $\rightside{\bipump{P^+}}\cap \leftside{\bipump{P^-}}$. For the sake of contradiction, suppose that there exists $0\leq i<j\leq |Q|-1$ such that $\type{Q_i}=\type{Q_j}$, then by lemma \ref{lem:redun:pump}, the path $R=Q_{i,i+1,\ldots,j}$ is in $\inuniterm$ and is bi-pumpable.  If $\vect{v}$ and $\vect{R_0R_{|R|-1}}$ are not collinear then by corollary \ref{cor:periodic}, $\uniterm$ is bi-periodic which is a contradiction of the hypothesis. By lemma \ref{biperiodic:sameperiod}, $\vect{R_0R_{|R|-1}}=s\vect{v}$ for some $s\in\{-1,1\}$, then the path $Q_{0,1,\ldots,i}(Q_{j+1,j+2,\ldots, |Q|-1}-s\vect{v})$ satisfies the same hypothesis has $Q$ contradicting its minimal length. If $Q$ and $Q+\vect{v}$ intersect then there exists $i\neq j$ such that $\type{Q_i}=\type{Q_j}$ which as previously leads to a contradiction. Then by definition of $Q$ and since $\uniterm$ is $\vect{v}$-periodic then the binding paths of $P^+$, $Q+\vect{v}$, $P^-$ and $Q$ form a simple cycle $C$. We define $\ass$ as the restriction of $\uniterm$ to the interior of this cycle. Moreover since $\uniterm$ is not bi-periodic then by lemma \ref{biperiodic:otherpath}, $P^+$ and $P^-$ cannot share a common tile type and the only common tile type between $P^+$ (resp. $P^-$) and $Q$ is $P^+_0=Q_0$ (resp. $P^-_0=Q_{|Q|-1}$). Then the length of the cycle $C$ is bounded by $2|T|$ and thus the size of $\ass$ is bounded by $|T|^2$. Consider the assembly $\bigcup_{\ell \in \Z} (\ass+\ell\vect{v})$, this assembly is correctly defined: there are no conflict for any $\ell\in Z$ between $\ass$ and $\ass+\ell\vect{v}$ since $\uniterm$ is $\vect{v}$-periodic and $\ass$ is connected since $Q+\vect{v}$ is in both $\ass$ and $\ass+\vect{v}$. Then $\bigcup_{\ell \in \Z} (\ass+\ell\vect{v})$ is a subassembly of $\alpha$. 
Finally, for any $(x,y) \in Z^2$ which is in $\rightside{\bipump{P^+}}\cap \leftside{\bipump{P^-}}$, there exists $\ell$ such that $(x,y)+\ell\vect{v}$ is in the interior of $C$. Thus for any tile $A$ of the restriction of $\uniterm$ to $\rightside{\bipump{P^+}}\cap \leftside{\bipump{P^-}}$, there are $\ell$ and a tile $B$ of $\ass$ such that $B=A+\ell\vect{v}$ (see appendix \ref{appB:curves}). Thus, $\bigcup_{\ell \in \Z} (\ass+\ell\vect{v})$ is equal to the restriction of $\uniterm$ to $\rightside{\bipump{P^+}}\cap \leftside{\bipump{P^-}}$.

\end{proof}

\subsection{Simply pumpable path}
\label{subsec:pumpablepath}

This section follows the key ideas of the arXiv version of \cite{Doty-2011}. We only give more details in lemma \ref{lem:magicindice} and theorem \ref{last:theorem} on how to deal with the combs (called here $\vect{v}$-self avoiding paths).

\subsubsection{Growing on a strictly pumpable path}

The three following lemmas shows that an arc of width at least $|P|$ cannot grow on a simply pumpable path $P$ which implies that any path growing on a simply pumpable path $P$ is $\vect{P_0P_{|P|-1}}$-self-avoiding. Figure~\ref{fig:simpleself} illustrates the reasonings of the three following lemmas.

\begin{figure}
\centering

\begin{tikzpicture}[x=0.33cm,y=0.33cm]

%colle
 \draw[very thick] (14.5,8.5) |- (12.5,5.5) |- (1.5,9.5) |- (5.5,1.5) |- (35,3.5);
 \draw[very thick] (10.5,3.5) |- (8.5,7.5) |- (3.5,5.5) |- (3.5,3.5);
 \draw[very thick] (22.5,3.5) |- (20.5,7.5) |- (18.5,5.5);
 \draw[very thick] (26.5,3.5) |- (24.5,7.5) |- (22.5,5.5); 
\draw[very thick] (30.5,3.5) |- (28.5,7.5) |- (26.5,5.5);

\draws{14}{8}
\draws{14}{7}
\draws{14}{6}
\draws{14}{5}
\draws{13}{5}
\draws{12}{5}
\draws{12}{6}
\draws{12}{7}
\draws{12}{8}
\draws{12}{9}
\draws{11}{9}
\draws{10}{9}
\draws{9}{9}
\draws{8}{9}
\drawblue{7}{9}
\drawblue{6}{9}
\drawblue{5}{9}
\drawblue{4}{9}
\drawblue{3}{9}
\drawblue{2}{9}
\drawblue{1}{9}
\drawblue{1}{8}
\drawblue{1}{7}
\drawblue{1}{6}
\drawblue{1}{5}
\drawblue{1}{4}
\drawblue{1}{3}
\drawblue{1}{2}
\drawblue{1}{1}
\drawblue{2}{1}
\drawblue{3}{1}
\drawblue{4}{1}
\drawblue{5}{1}
\drawblue{5}{2}
\drawgray{5}{3}
\drawgray{6}{3}
\drawgray{7}{3}
\drawgray{8}{3}
\drawt{9}{3}
\drawt{10}{3}
\drawt{11}{3}
\drawt{12}{3}
\drawt{13}{3}
\drawt{14}{3}
\drawt{15}{3}
\drawt{15}{3}
\drawt{16}{3}
\drawt{17}{3}
\drawt{18}{3}
\drawt{19}{3}
\drawt{20}{3}
\drawt{21}{3}
\drawt{22}{3}
\drawt{23}{3}
\drawt{24}{3}
\drawt{25}{3}
\drawt{26}{3}
\drawt{27}{3}
\drawt{28}{3}
\drawt{29}{3}
\drawt{30}{3}
\drawt{31}{3}
\drawt{32}{3}
\drawt{33}{3}
\drawt{34}{3}

\drawred{10}{3}
\drawred{10}{4}
\drawred{10}{5}
\drawred{10}{6}
\drawred{10}{7}
\drawred{9}{7}
\drawred{8}{7}
\drawred{8}{6}
\drawred{8}{5}
\drawred{7}{5}
\drawred{6}{5}
\drawred{5}{5}
\drawred{4}{5}
\drawred{3}{5}
\drawred{3}{4}
\drawred{3}{3}

\drawr{22}{3}
\drawr{22}{4}
\drawr{22}{5}
\drawr{22}{6}
\drawr{22}{7}
\drawr{21}{7}
\drawr{20}{7}
\drawr{20}{6}
\drawr{20}{5}
\drawr{19}{5}
\drawr{18}{5}

\drawr{26}{3}
\drawr{26}{4}
\drawr{26}{5}
\drawr{26}{6}
\drawr{26}{7}
\drawr{25}{7}
\drawr{24}{7}
\drawr{24}{6}
\drawr{24}{5}
\drawr{23}{5}
\drawr{22}{5}

\drawr{30}{3}
\drawr{30}{4}
\drawr{30}{5}
\drawr{30}{6}
\drawr{30}{7}
\drawr{29}{7}
\drawr{28}{7}
\drawr{28}{6}
\drawr{28}{5}
\drawr{27}{5}
\drawr{26}{5}

\draw (6.5, 6.8) node {$Q$};
\draw (5.5, 4.5) node {$P_0$};
\draw (9, 2.2) node {$P_{|P|-1}$};

\draw[->,thick] (10.5, 2.5) -- (22.5,2.5);
\draw (16.5, 1.2) node {3$\vect{P_0P_{|P|-1}}$};

\draw[->,thick] (22.5, 2.5) -- (26.5,2.5);
\draw (24.2, 1.2) node {$\vect{P_0P_{|P|-1}}$};

\draw[->,thick] (26.5, 2.5) -- (30.5,2.5);
\draw (28.8, 1.2) node {$\vect{P_0P_{|P|-1}}$};

\path [dotted, draw, thin] (0,0) grid[step=0.33cm] (35,13);
\end{tikzpicture}

a) The seed is in black, the simply pumpable path $P$ is in gray and its pumping is in white, a blue path binds with the seed and $P_0$ and a path $Q$ (in red) growing on $\pump{P}$ intersects with $\bipump{P}$. Then, it is possible to grow some translations of a prefix of $Q$ (in light red) which intersect.
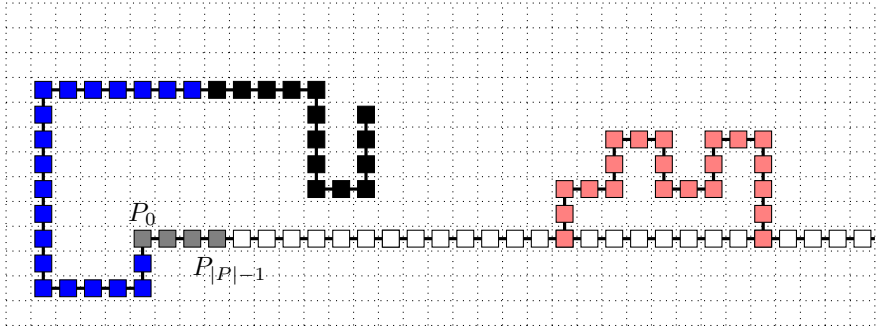
\vspace{+0.5em}

\begin{tikzpicture}[x=0.33cm,y=0.33cm]

%colle
 \draw[very thick] (14.5,8.5) |- (12.5,5.5) |- (1.5,9.5) |- (5.5,1.5) |- (35,3.5);
 \draw[very thick] (22.5,3.5) |- (22.5,5.5);
 \draw[very thick] (26.5,5.5) |- (24.5,7.5) |- (22.5,5.5); 
\draw[very thick] (30.5,3.5) |- (28.5,7.5) |- (26.5,5.5);

\draws{14}{8}
\draws{14}{7}
\draws{14}{6}
\draws{14}{5}
\draws{13}{5}
\draws{12}{5}
\draws{12}{6}
\draws{12}{7}
\draws{12}{8}
\draws{12}{9}
\draws{11}{9}
\draws{10}{9}
\draws{9}{9}
\draws{8}{9}
\drawblue{7}{9}
\drawblue{6}{9}
\drawblue{5}{9}
\drawblue{4}{9}
\drawblue{3}{9}
\drawblue{2}{9}
\drawblue{1}{9}
\drawblue{1}{8}
\drawblue{1}{7}
\drawblue{1}{6}
\drawblue{1}{5}
\drawblue{1}{4}
\drawblue{1}{3}
\drawblue{1}{2}
\drawblue{1}{1}
\drawblue{2}{1}
\drawblue{3}{1}
\drawblue{4}{1}
\drawblue{5}{1}
\drawblue{5}{2}
\drawgray{5}{3}
\drawgray{6}{3}
\drawgray{7}{3}
\drawgray{8}{3}
\drawt{9}{3}
\drawt{10}{3}
\drawt{11}{3}
\drawt{12}{3}
\drawt{13}{3}
\drawt{14}{3}
\drawt{15}{3}
\drawt{15}{3}
\drawt{16}{3}
\drawt{17}{3}
\drawt{18}{3}
\drawt{19}{3}
\drawt{20}{3}
\drawt{21}{3}
\drawt{22}{3}
\drawt{23}{3}
\drawt{24}{3}
\drawt{25}{3}
\drawt{26}{3}
\drawt{27}{3}
\drawt{28}{3}
\drawt{29}{3}
\drawt{30}{3}
\drawt{31}{3}
\drawt{32}{3}
\drawt{33}{3}
\drawt{34}{3}

\drawr{22}{3}
\drawr{22}{4}

\drawr{26}{5}
\drawr{26}{6}
\drawr{26}{7}
\drawr{25}{7}
\drawr{24}{7}
\drawr{24}{6}
\drawr{24}{5}
\drawr{23}{5}
\drawr{22}{5}

\drawr{30}{3}
\drawr{30}{4}
\drawr{30}{5}
\drawr{30}{6}
\drawr{30}{7}
\drawr{29}{7}
\drawr{28}{7}
\drawr{28}{6}
\drawr{28}{5}
\drawr{27}{5}

\draw (5.5, 4.5) node {$P_0$};
\draw (9, 2.2) node {$P_{|P|-1}$};

\path [dotted, draw, thin] (0,0) grid[step=0.33cm] (35,13);
\end{tikzpicture}

b) Using the tiles in light red, it is possible to assemble a path of width $2(|P|-1)$.
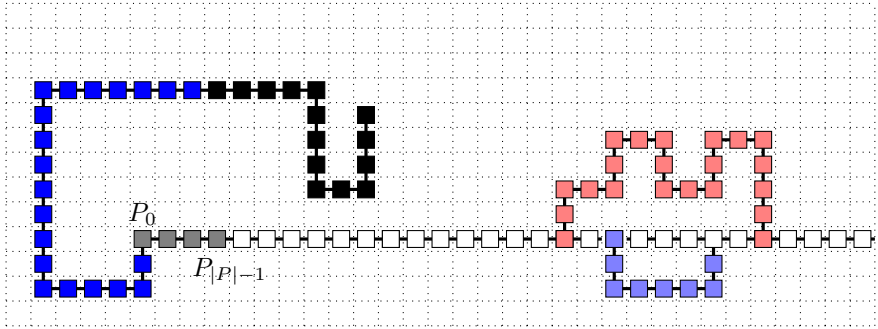
\vspace{+0.5em}

\begin{tikzpicture}[x=0.33cm,y=0.33cm]

%colle
 \draw[very thick] (14.5,8.5) |- (12.5,5.5) |- (1.5,9.5) |- (5.5,1.5) |- (24,3.5);
 \draw[very thick] (25,3.5) |- (35,3.5);
 \draw[very thick] (22.5,3.5) |- (22.5,5.5);
 \draw[very thick] (26.5,5.5) |- (24.5,7.5) |- (22.5,5.5); 
\draw[very thick] (30.5,3.5) |- (28.5,7.5) |- (26.5,5.5);
\draw[very thick] (28.5,3.5) |- (24.5,1.5) |- (24.5,3.5);

\draws{14}{8}
\draws{14}{7}
\draws{14}{6}
\draws{14}{5}
\draws{13}{5}
\draws{12}{5}
\draws{12}{6}
\draws{12}{7}
\draws{12}{8}
\draws{12}{9}
\draws{11}{9}
\draws{10}{9}
\draws{9}{9}
\draws{8}{9}
\drawblue{7}{9}
\drawblue{6}{9}
\drawblue{5}{9}
\drawblue{4}{9}
\drawblue{3}{9}
\drawblue{2}{9}
\drawblue{1}{9}
\drawblue{1}{8}
\drawblue{1}{7}
\drawblue{1}{6}
\drawblue{1}{5}
\drawblue{1}{4}
\drawblue{1}{3}
\drawblue{1}{2}
\drawblue{1}{1}
\drawblue{2}{1}
\drawblue{3}{1}
\drawblue{4}{1}
\drawblue{5}{1}
\drawblue{5}{2}
\drawgray{5}{3}
\drawgray{6}{3}
\drawgray{7}{3}
\drawgray{8}{3}
\drawt{9}{3}
\drawt{10}{3}
\drawt{11}{3}
\drawt{12}{3}
\drawt{13}{3}
\drawt{14}{3}
\drawt{15}{3}
\drawt{15}{3}
\drawt{16}{3}
\drawt{17}{3}
\drawt{18}{3}
\drawt{19}{3}
\drawt{20}{3}
\drawt{21}{3}
\drawt{22}{3}
\drawt{23}{3}
\drawt{24}{3}
\drawt{25}{3}
\drawt{26}{3}
\drawt{27}{3}
\drawt{28}{3}
\drawt{29}{3}
\drawt{30}{3}
\drawt{31}{3}
\drawt{32}{3}
\drawt{33}{3}
\drawt{34}{3}

\drawb{24}{3}
\drawb{24}{2}
\drawb{24}{1}
\drawb{25}{1}
\drawb{26}{1}
\drawb{27}{1}
\drawb{28}{1}
\drawb{28}{2}

\drawr{22}{3}
\drawr{22}{4}

\drawr{26}{5}
\drawr{26}{6}
\drawr{26}{7}
\drawr{25}{7}
\drawr{24}{7}
\drawr{24}{6}
\drawr{24}{5}
\drawr{23}{5}
\drawr{22}{5}

\drawr{30}{3}
\drawr{30}{4}
\drawr{30}{5}
\drawr{30}{6}
\drawr{30}{7}
\drawr{29}{7}
\drawr{28}{7}
\drawr{28}{6}
\drawr{28}{5}
\drawr{27}{5}

\draw (5.5, 4.5) node {$P_0$};
\draw (9, 2.2) node {$P_{|P|-1}$};

\path [dotted, draw, thin] (0,0) grid[step=0.33cm] (35,13);
\end{tikzpicture}

c) We try to assemble a translation of the blue path and the seed (in light blue), the previous arc allows us to remove a tile of $\pump{P}$ which creates a conflict and a contradiction.

\caption{Consider a simple pumpable path $P$, if a path grows on $\pump{P}$ then it is $\vect{P_0P_{|P|-1}}$-self-avoiding (see lemmas \ref{lem:transpath},\ref{lem:transarc} and \ref{lem:noarc}).}
\label{fig:simpleself}
\end{figure}

\begin{lemma} 
\label{lem:transpath}
Consider a tile assembly system $\mathcal{T}=(T,\sigma,1)$ whose terminal assembly is $\uniterm$ and a simply pumpable path $P$ in $\inuniterm$. If there exists a path $Q$ growing on $\pump{P}$ at index $i\geq |P|-1$ such that $Q$ and $Q+j\vect{P_0P_{|P|-1}}$ intersect for some $j \in \N^*$, then there exists a path $R$ growing on $\pump{P}$ such that $R$ and $R+k\vect{P_0P_{|P|-1}}$ intersect for some $k \in \N^*$ and for all $\ell\in \N$, $R+\ell\vect{P_0P_{|P|-1}}$ is in $\inuniterm$.
\end{lemma}

\begin{proof}
Since $P$ is in $\inuniterm$, there exists a finite producible subassembly $\ass$  ($\sigma$ is a subassembly of $\ass$) such that $P_0$ is a tile of $\ass$. Since $\ass$ is finite and $\vect{P_0P_{|P|-1}}$ is not null, there exists $L \in \N$ such that for all $\ell\geq L$, $Q+\ell\vect{P_0P_{|P|-1}}$ does not intersect with $\ass$. If for all $\ell\geq L$, $Q+\ell\vect{P_0P_{|P|-1}}$ does not intersect with $\pump{P}$ then $Q+\ell\vect{P_0P_{|P|-1}}$ is a path of $\inuniterm$ which grows on $P$. In this case, we set $R=Q+L\vect{P_0P_{|P|-1}}$ and $k=j$ and the lemma is true. Otherwise, we can define $$m=\min\{n>0:\textrm{there exists } \ell>L \textrm{ such that } \pos{Q}_{n}+\ell\vect{P_0P_{|P|-1}}\textrm{ is in } \pos{\pump{P}}\}.$$ 
By definition of $m$, there exists $L' \geq L$ such that $Q_m+L'\vect{P_0P_{|P|-1}}$ intersects with a tile $P_n$ of $\pump{P}$. In this case, we set $k=1$ and $R=Q_{0,1,\ldots,m-1}+L'\vect{P_0P_{|P|-1}}$. By definition of $L'$, $L$ and $m$, for all $\ell\geq 0$, $R+\ell\vect{P_0P_{|P|-1}}$ grows on $\pump{P}$ on tile $P_{i+(L'+\ell)(|P|-1)}$ and does not intersect with $\ass$ and thus is in $\inuniterm$.  Now, since by hypothesis $Q$ grows on $\pump{P}$ then $Q_m$ is not a tile of $\pump{P}$ and since  $i \geq |P|-1$, we have $(i+L'(|P|-1))-n>|P|-1$. The binding graph of $R(Q_m+L'\vect{P_0P_{|P|-1}})$ is the binding graph of an arc of width at least $|P|$ thus $R$ intersects with $R+\vect{P_0P_{|P|-1}}$. 
\end{proof}

\begin{lemma} 
\label{lem:transarc}
Consider a tile assembly system $\mathcal{T}=(T,\sigma,1)$ whose terminal assembly is $\uniterm$ and a simply pumpable path $P$ in $\inuniterm$. If there exists a path $Q$ growing on $\pump{P}$ such that $Q$ and $Q+i\vect{P_0P_{|P|-1}}$ intersect for some $i \in \N^*$ and for all $\ell\in \N$, $Q+\ell\vect{P_0P_{|P|-1}}$ is in $\inuniterm$, then there exists an arc $A$ of width at least $|P|$ growing on $\pump{P}$ such that for all $\ell\in\N$, $A+\ell\vect{P_0P_{|P|-1}}$ is in $\inuniterm$.
\end{lemma}

\begin{proof}
Since $Q+i\vect{P_0P_{|P|-1}}$ intersects with $Q$ and $Q+2i\vect{P_0P_{|P|-1}}$ and since all intersections are agreement (these three paths all belongs to $\inuniterm$), then we can define the following assembly:
$$\ass=Q \cup (Q_{1,2,\ldots, |Q|-1}+i\vect{P_0P_{|P|-1}}) \cup (Q+2i\vect{P_0P_{|P|-1}}).$$ Since for all $\ell\in \N$, $Q+\ell\vect{P_0P_{|P|-1}}$ is in $\inuniterm$, then for all $\ell\in \N$, $\ass+\ell\vect{P_0P_{|P|-1}}$ is a subassembly of $\uniterm$. By definition of growing, the only intersection between $\pump{P}$ and $\ass$ is $Q_0$ and $Q_0+2i\vect{P_0P_{|P|-1}}$. Thus there exists an arc $A$ growing on $\pump{P}$ of width $2i(|P|-1)>|P|$ such that $\asm{A}$ is a subassembly of $\ass$. Moreover, for all $\ell \in \N$, $A+\ell\vect{P_0P_{|P|-1}}$ is in $\inuniterm$.
\end{proof}

\begin{lemma} 
\label{lem:noarc}
Consider a tile assembly system $\mathcal{T}=(T,\sigma,1)$ whose terminal assembly is $\uniterm$ and a simply pumpable path $P$ in $\inuniterm$ then there exists no arc $A$ of width at least $|P|$ growing on $\pump{P}$.  
\end{lemma}

\begin{proof}
For the sake of contradiction suppose that exists an arc $A$ of width at least $|P|$ growing on $\pump{P}$. Since $P$ is in $\inuniterm$, there exists a finite subassembly $\ass$ of $\uniterm$ such that $\ass$ is producible ($\sigma$ is a subassembly of $\ass$) and $P_0$ is a tile of $\ass$. Since $\ass$ is finite and $\vect{P_0P_{|P|-1}}$ is not null, there exists $L \in \N$ such that for all $\ell\geq L$, $A+\ell\vect{P_0P_{|P|-1}}$ does not intersect with $\ass$. If for all $\ell\geq L$, $A+\ell\vect{P_0P_{|P|-1}}$ does not intersect with $\pump{P}$ then $A+\ell\vect{P_0P_{|P|-1}}$ is a path of $\inuniterm$. In this case, we set $B=A+L\vect{P_0P_{|P|-1}}$. Otherwise, since the width of $A$ is at least $|P|$ we can suppose that $A_0=P_i$ with $i>|P|-1$ and there is an intersection between $A$ and $A+\vect{P_0P_{|P|-1}}$ which does not occur at $\pos{A_0}$ or $\pos{A_{|A|-1}}$. Let $Q=A_{0,1,\ldots,|A|-2}$, this path grows on $\pump{P}$ at index $i$ and intersects with $Q+\vect{P_0P_{|P|-1}}$. Then by lemmas \ref{lem:transpath} and \ref{lem:transarc}, there exists an arc $B$ of width at least $|P|$ growing on $\pump{P}$ such that for all $\ell\in\N$, $B+\ell\vect{P_0P_{|P|-1}}$ is in $\inuniterm$. 

Now, let $0\leq m <n $ such that the arc $B$ grows between $P_m$ and $P_n$ (then $n-m\geq |P|$). If there is no conflict between $\ass$ and $\bipump{P}$ then $P$ would be bi-pumpable which contradicts the hypothesis of the lemma. Then there exists a conflict between $\bipump{P}$ and $\ass$. Moreover, there exists $j\in \N$ such that $\ass+j\vect{P_0P_{|P|-1}}$ does not intersect with $\ass$ and there is at least one conflict between $\ass+j\vect{P_0P_{|P|-1}}$ and $\pump{P}$ at $\pos{P_k}$ for some $k>m$. Let $t$ be the tile such that $\pos{t}=\pos{P_k}$ and $\type{t}=(\ass+j\vect{P_0P_{|P|-1}})(\pos{P_k})$ then $\type{t}\neq \type{P_k}$. By definition of $\ass$, there is at least one agreement between $\ass+j\vect{P_0P_{|P|-1}}$ and $\pump{P}$ at $P_0+j\vect{P_0P_{|P|-1}}$. Thus, there exists a path $R$ which grows on $P$ such that $\asm{R}$ is a sub-assembly of $\ass+j\vect{P_0P_{|P|-1}}$ and the tile $R_{|R|-1}$ interacts with $t$. Moreover, since $\ass+j\vect{P_0P_{|P|-1}}$ does not intersect with $\ass$ then $R$ is in $\inuniterm$. Now since $k>m$ and $n-m\geq |P|$, there exists $\ell$ such that $m+\ell(|P|-1)<k<n+\ell(|P|-1)$. We remind that $\ass$ intersects with $\pump{P}$ by definition, that $B+\ell\vect{P_0P_{|P|-1}}$ and $R$ grow on $\pump{P}$ and are in $\inuniterm$, then the following assembly is correctly defined: $$\gamma=\ass \cup \pump{P} \cup (B+\ell\vect{P_0P_{|P|-1}}) \cup R.$$ Moreover since $\ass$ is producible then $\gamma$ is also producible. Removing the tile $P_k$ in $\pump{P}$ disconnects the path in two parts but adding $B+\ell\vect{P_0P_{|P|-1}}$ reconnects these two parts, then it is possible to remove the tile $P_k$ in $\gamma$ and since $\pos{t}=\pos{P_k}$ and $t$ interacts with $R_{|R|-1}$, it is possible to switch $P_k$ by $t$ in $\gamma$. Since $\type{P_k}\neq \type{t}$, $(T,\sigma,1)$ is not directed which contradicts the hypothesis.
\end{proof}

\begin{corollary} 
\label{cor:selfavoid}
Consider a tile assembly system $\mathcal{T}=(T,\sigma,1)$ whose terminal assembly is $\uniterm$, a simply pumpable path $P$ in $\inuniterm$ and a path $Q$ growing on $\pump{P}$ at index $i \geq |P|-1$ then $Q$ is $\vect{P_0P_{|P|-1}}$-self-avoiding and for all $\ell\in\N$ $Q+\ell\vect{P_0P_{|P|-1}}$ grows on $\pump{P}$.
\end{corollary}

\begin{proof}
For the sake of contradiction, suppose that there exists a path $Q$ growing on $\pump{P}$ on a tile $P_i$ with $i \geq |P|-1$ which is not $\vect{P_0P_{|P|-1}}$-self-avoiding. If for some $j\in \N^*$, $Q$ intersects with $Q+j\vect{P_0P_{|P|-1}}$ then by lemmas \ref{lem:transpath} and \ref{lem:transarc}, there exists an arc $A$ growing on $P$ of width at least $|P|$ which is a contradiction of lemma \ref{lem:noarc}. By observation \ref{obs:grow}, for all $\ell\in\N$, $Q+\ell\vect{P_0P_{|P|-1}}$ grows on $\pump{P}$. Finally, consider a finite assembly $\ass$ producible by $(T,\sigma,1)$ such that $P_0$ is a tile of $\ass$, then $\ass \cup \pump{P}$ is producible by $(T,\sigma,1)$. Since $\ass$ is finite there exists $L\in \N$ such that $Q+L\vect{P_0P_{|P|-1}}$ does not intersect with $\ass$ and thus for all $\ell\geq L$, $\ass \cup \pump{P}\cup (Q+\ell\vect{P_0P_{|P|-1}})$ is producible by $(T,\sigma,1)$ and then $Q+\ell\vect{P_0P_{|P|-1}}$ is in $\inuniterm$.
\end{proof}

The following lemmas shows that any path growing on a $\vect{v}$-self-avoiding path $P$ is $(\vect{v},\vect{P_0P_{|P|-1}})$-self-avoiding.

\begin{lemma} 
\label{lem:doubleavoid}
Consider a tile assembly system $\mathcal{T}=(T,\sigma,1)$ whose terminal assembly is $\uniterm$, a simply pumpable (resp. bi-pumpable) path $P$ in $\inuniterm$ such that $\pump{P}$ (resp. $\bipump{P}$) is $\vect{v}$-self-avoiding (with $\vect{v}$ not collinear with $\vect{P_0P_{|P|-1}}$) then any path growing on $\pump{P}$ (resp. $\bipump{P}$) at index $i \geq |P|-1$ (resp. $i\in\Z$) is $(\vect{v},\vect{P_0P_{|P|-1}})$-self-avoiding. 
\end{lemma} 

\begin{proof}
Consider a path $Q$ growing on $\pump{P}$ (resp. $\bipump{P}$) at index $i \geq |P|-1$ (resp. $i\in\Z$). %By definition of $\vect{v}$-self-avoiding, $Q$ is also $\vect{v}$-self-avoiding and by corollary \ref{cor:selfavoid}, $Q$ is $\vect{P_0P_{|P|-1}}$-self-avoiding. 
For the sake of contradiction suppose that $Q$ is not $(\vect{v},\vect{P_0P_{|P|-1}})$-self-avoiding, then there exists $\ell,\ell' \in \Z$ such that $Q$ intersect with $R=Q+\ell\vect{v}+\ell'\vect{P_0P_{|P|-1}}$. See Figure \ref{fig:doubleself} for an illlustation of the following reasoning.
If $\ell=0$ then $Q$ is not $\vect{P_0P_{|P|-1}}$-self-avoiding which contradicts corollary \ref{cor:selfavoid}. 
Remark that $Q$ cannot intersect with $\bipump{P}$ except at $Q_0$ because $Q$ grows on $\pump{P}$ and if it intersects with $\bipump{P}$ then there exists $j \in \N$ such that $Q+j\vect{v}$ intersects with $\pump{P}$ and thus $Q+j\vect{v}$ would not grow on $\pump{P}$ which contradicts corollary~\ref{cor:selfavoid} (resp. $Q$ grows on $\bipump{P}$). 
Without loss of generality we suppose that $Q$ turns left of $\pump{P}$ (resp. $\bipump{P}$) and then $Q$ is in the left side of $\bipump{P}$ which implies that $R$ is in the left side of $\bipump{P}+\ell\vect{v}$. 
Paths $\bipump{P}+\ell\vect{v}$ and $\bipump{P}$ cannot intersect since $\pump{P}$ (resp. $\bipump{P}$) is $\vect{v}$-self-avoiding and if $\bipump{P}+\ell\vect{v}<\bipump{P}$, remark that there is also an intersection between $Q-\ell\vect{v}-\ell'\vect{P_0P_{|P|-1}}$ and $R-\ell\vect{v}-\ell'\vect{P_0P_{|P|-1}}=Q$ and that $\bipump{P}-\ell\vect{v}>\bipump{P}$ in this case. Then without loss of generality, we can suppose that $\bipump{P}+\ell\vect{v}>\bipump{P}$ and then $Q_0$ is in the right side of $\bipump{P}+\ell\vect{v}$ thus the path $Q$ must intersect with $\bipump{P}+\ell\vect{P_iP_j}$ to intersect with  $R$.
%lemma à patcher 
In this case, since $Q$ is $\vect{P_0P_{|P|-1}}$-self-avoiding (by corollary \ref{cor:selfavoid}) there exists $L \in \Z$ such that $Q+L\vect{P_iP_j}$ and $Q+(L+2)\vect{P_iP_j}$ are in $\inuniterm$, grow on $\pump{P}$ and intersect with $\pump{P}+\ell\vect{v}$. 
%which contradicts the fact that $P$ is $\vect{v}$-self-avoiding. 
Moreover, we can suppose that $\pump{P}+\ell\vect{v}$ is in $\inuniterm$ (up to some translation since $\pump{P}$ is $\vect{v}$-self-avoiding) and then there exists an arc of $P$ of width $2(|P|-1)>|P|$ which is a subassembly of $Q+L\vect{P_iP_j}$, $Q+(L+2)\vect{P_iP_j}$ and $\pump{P}+\ell\vect{v}$. This arc contradicts lemma \ref{lem:noarc}.

\begin{figure}
\centering

\begin{tikzpicture}[x=0.33cm,y=0.33cm]

%colle

\draw[dashed,very thick] (0.5,6.5) |- (2.5,6.5)  |- (5.5,3.5);
\draw[very thick] (5.5,3.5) |- (8.5,6.5)  |- (11.5,3.5);
\drawgray{5}{3}
\drawgray{5}{4}
\drawgray{5}{5}
\drawgray{5}{6}
\drawgray{5}{6}
\drawgray{6}{6}
\drawgray{7}{6}
\drawgray{8}{6}
\drawgray{8}{5}
\drawgray{8}{4}
\drawgray{8}{3}
\drawgray{9}{3}
\drawgray{10}{3}

 \draw[very thick] (11.5,3.5) |- (14.5,6.5)  |- (17.5,3.5);
\drawgray{11}{3}
\drawt{11}{4}
\drawt{11}{5}
\drawt{11}{6}
\drawt{12}{6}
\drawt{13}{6}
\drawt{14}{6}
\drawt{14}{5}
\drawt{14}{4}
\drawt{14}{3}
\drawt{15}{3}
\drawt{16}{3}

 \draw[very thick] (17.5,3.5) |- (20.5,6.5)  |- (23.5,3.5);
\drawt{17}{3}
\drawt{17}{4}
\drawt{17}{5}
\drawt{17}{6}
\drawt{18}{6}
\drawt{19}{6}
\drawt{20}{6}
\drawt{20}{5}
\drawt{20}{4}
\drawt{20}{3}
\drawt{21}{3}
\drawt{22}{3}

 \draw[very thick] (23.5,3.5) |- (26.5,6.5)  |- (29.5,3.5);
\drawt{23}{3}
\drawt{23}{4}
\drawt{23}{5}
\drawt{23}{6}
\drawt{24}{6}
\drawt{25}{6}
\drawt{26}{6}
\drawt{26}{5}
\drawt{26}{4}
\drawt{26}{3}
\drawt{27}{3}
\drawt{28}{3}

 \draw[very thick] (29.5,3.5) |- (32.5,6.5)  |- (35,3.5);
\drawt{29}{3}
\drawt{29}{4}
\drawt{29}{5}
\drawt{29}{6}
\drawt{30}{6}
\drawt{31}{6}
\drawt{32}{6}
\drawt{32}{5}
\drawt{32}{4}
\drawt{32}{3}
\drawt{33}{3}
\drawt{34}{3}

\draw[dashed,very thick] (0.5,16.5) |- (2.5,16.5)  |- (5.5,13.5);
\draw[very thick] (5.5,13.5) |- (8.5,16.5)  |- (11.5,13.5);
 \draw[very thick] (11.5,13.5) |- (14.5,16.5)  |- (17.5,13.5);
 \draw[very thick] (17.5,13.5) |- (20.5,16.5)  |- (23.5,13.5);
 \draw[red,very thick] (9.5,13.5) -| (11.5,16.5) -| (14.5,13.5) -| (17.5,16.5) -| (20.5,13.5)  -| (21.5,13.5);
\drawt{5}{13}
\drawt{5}{14}
\drawt{5}{15}
\drawt{5}{16}
\drawt{5}{16}
\drawt{6}{16}
\drawt{7}{16}
\drawt{8}{16}
\drawt{8}{15}
\drawt{8}{14}
\drawt{8}{13}
\drawt{9}{13}
\drawt{10}{13}

\drawt{11}{13}
\drawt{11}{14}
\drawt{11}{15}
\drawt{11}{16}
\drawt{12}{16}
\drawt{13}{16}
\drawt{14}{16}
\drawt{14}{15}
\drawt{14}{14}
\drawt{14}{13}
\drawt{15}{13}
\drawt{16}{13}

\drawt{17}{13}
\drawt{17}{14}
\drawt{17}{15}
\drawt{17}{16}
\drawt{18}{16}
\drawt{19}{16}
\drawt{20}{16}
\drawt{20}{15}
\drawt{20}{14}
\drawt{20}{13}
\drawt{21}{13}
\drawt{22}{13}

 \draw[very thick] (23.5,13.5) |- (26.5,16.5)  |- (29.5,13.5);
\drawt{23}{13}
\drawt{23}{14}
\drawt{23}{15}
\drawt{23}{16}
\drawt{24}{16}
\drawt{25}{16}
\drawt{26}{16}
\drawt{26}{15}
\drawt{26}{14}
\drawt{26}{13}
\drawt{27}{13}
\drawt{28}{13}

 \draw[very thick] (29.5,13.5) |- (32.5,16.5)  |- (35,13.5);
\drawt{29}{13}
\drawt{29}{14}
\drawt{29}{15}
\drawt{29}{16}
\drawt{30}{16}
\drawt{31}{16}
\drawt{32}{16}
\drawt{32}{15}
\drawt{32}{14}
\drawt{32}{13}
\drawt{33}{13}
\drawt{34}{13}

\draw[dashed,very thick] (0.5,26.5) |- (2.5,26.5)  |- (5.5,23.5);
\draw[dashed,very thick] (5.5,23.5) |- (8.5,26.5)  |- (11.5,23.5);
\draw[dashed,very thick] (11.5,23.5) |- (14.5,26.5)  |- (17.5,23.5);
\draw[dashed,very thick] (17.5,23.5) |- (20.5,26.5)  |- (23.5,23.5);
\draw[dashed,very thick] (23.5,23.5) |- (26.5,26.5)  |- (29.5,23.5);
\draw[dashed,very thick] (29.5,23.5) |- (32.5,26.5)  |- (35,23.5);

 \draw[very thick] (11.5,6.5) |- (7.5,8.5)  |- (3.5,10.5) |- (3.5,20.5);
\drawred{11}{6}
\drawred{11}{7}
\drawred{11}{8}
\drawred{10}{8}
\drawred{9}{8}
\drawred{8}{8}
\drawred{7}{8}
\drawred{7}{9}
\drawred{7}{10}
\drawred{6}{10}
\drawred{5}{10}
\drawred{4}{10}
\drawred{3}{10}
\drawred{3}{11}
\drawred{3}{12}
\drawred{3}{13}
\drawred{3}{14}
\drawred{3}{15}
\drawred{3}{16}
\drawred{3}{17}
\drawred{3}{18}
\drawred{3}{19}
\drawred{3}{20}

 \draw[very thick] (11.5,16.5) |- (7.5,18.5)  |- (3.5,20.5) |- (3.5,30.5);
\drawr{11}{16}
\drawr{11}{17}
\drawr{11}{18}
\drawr{10}{18}
\drawr{9}{18}
\drawr{8}{18}
\drawr{7}{18}
\drawr{7}{19}
\drawr{7}{20}
\drawr{6}{20}
\drawr{5}{20}
\drawr{4}{20}
\drawr{3}{20}
\drawr{3}{21}
\drawr{3}{22}
\drawr{3}{23}
\drawr{3}{24}
\drawr{3}{25}
\drawr{3}{26}
\drawr{3}{27}
\drawr{3}{28}
\drawr{3}{29}
\drawr{3}{30}

\draw[red,very thick] (17.5,6.5) |- (13.5,8.5)  |- (9.5,10.5) |- (9.5,13.5);
\drawlr{17}{6}
\drawlr{17}{7}
\drawlr{17}{8}
\drawlr{16}{8}
\drawlr{15}{8}
\drawlr{14}{8}
\drawlr{13}{8}
\drawlr{13}{9}
\drawlr{13}{10}
\drawlr{12}{10}
\drawlr{11}{10}
\drawlr{10}{10}
\drawlr{9}{10}
\drawlr{9}{11}
\drawlr{9}{12}
\drawlr{9}{13}

\draw[red,very thick] (29.5,6.5) |- (25.5,8.5)  |- (21.5,10.5) |- (21.5,13.5);
\drawlr{29}{6}
\drawlr{29}{7}
\drawlr{29}{8}
\drawlr{28}{8}
\drawlr{27}{8}
\drawlr{26}{8}
\drawlr{25}{8}
\drawlr{25}{9}
\drawlr{25}{10}
\drawlr{24}{10}
\drawlr{23}{10}
\drawlr{22}{10}
\drawlr{21}{10}
\drawlr{21}{11}
\drawlr{21}{12}
\drawlr{21}{13}

\draw[thick,->] (31.5, 6.5) -- (31.5, 16.5);
\draw (32.5, 11.5) node {$\vect{v}$};

\draw[thick,->] (3.5, 11.5) -- (9.5, 11.5);
\draw (6.5, 12.5) node {$\vect{P_0P_{|P|-1}}$};

\draw[thick,->] (7.5, 9.5) -- (25.5, 9.5);
\draw (17.5, 10.5) node {$3\vect{P_0P_{|P|-1}}$};

\draw[very thick,->] (4.5, 21.5) -- (3.5, 20.5);
\draw (7.5, 22) node {intersection};

\draw (5.5, 2.3) node {$P_0$};
\draw (12, 2.2) node {$P_{|P|-1}$};
\draw (10.2, 6.5) node {$Q_0$};
\path [dotted, draw, thin] (0,0) grid[step=0.33cm] (35,32);
\end{tikzpicture}
\caption{Illustration of proof \ref{lem:doubleavoid}. Consider a simply pumpable path $P$ in gray which is $\vect{v}$-self-avoiding. Its pumping and its translation by $\vect{v}$ are in white. A path $Q$ (in dark red) grows on $\pump{P}$ and collides with $Q+\vect{v}$ (in red), then using two translations of $Q$ (in light red) by $\vect{P_0P_{|P|-1}}$ and $3\vect{P_0P_{|P|-1}}$, we assemble an arc of width $2(|P|-1)$ on $\pump{P}$ with the help of $\pump{P}+\vect{v}$.}
\label{fig:doubleself}
\end{figure}
\end{proof}

The following lemma shows that for any simple pumpable path $P$ there is an index $i$ such that any path growing on $\pump{P}$ after this index belongs to $\inuniterm$. This lemma is useful to study the terminal assembly since it implies that ultimately we do not care about what happens near the seed when analyzing a periodic path.

\begin{lemma} 
\label{lem:magicindice}
Consider a tile assembly system $\mathcal{T}=(T,\sigma,1)$ whose terminal assembly is $\uniterm$, a simply pumpable path $P$ in $\inuniterm$ and a producible finite assembly $\ass$ such that $P_0$ is a tile of $\ass$. There exists an index $i$ such that any path growing on $\pump{P}$ at index $j\geq i$ is in $\inuniterm$. Moreover, this index $i$ depends only of $|\ass|$, $P$ and $|T|$.
\end{lemma}

\begin{proof}
Without loss of generality we suppose that $P_0$ is the only intersection between $\ass$ and $\pump{P}$. Let $j=(4|\ass|+2)(|P|-1)+1$ and consider the assembly $\gamma=\ass\cup \pump{P}_{0,1,\ldots,j}$, by definition of $\ass$, this assembly intersects with $P$ and all intersections are agreement, then $\gamma$ is producible by $=(T,\sigma,1)$. Any assembly producible by $(T,\gamma,1)$ is also producible by $(T,\sigma,1)$ then $(T,\gamma,1)$ is directed and its terminal assembly is also $\uniterm$ (since $\gamma$ is a sub-assembly of $\uniterm$). Since $\ass$ is finite and $\vect{P_0P_{|P|-1}}$ is not null there exists $i>j+(|P|-1)$ such that the distance between any tile of $\pump{P}_{i,i+1,..,+\infty}$ and any tile of $\ass$ is at least $\funpump{|T|}{|\gamma|}+1$ (see lemma \ref{theorem:pumping}). Remark that $i$ depends only of $|\ass|$, $P$ and $|T|$. See Figure \ref{fig:magicIndex} for an illustration of the following reasoning.

\begin{figure}
\centering

\begin{tikzpicture}[x=0.33cm,y=0.33cm]

%colle
 \draw[very thick] (10.5,12.5) -| (19.5,9.5) -| (24.5,12.5);
 \draw[very thick] (8.5,10.5) -| (5.5,12.5) -| (8.5,12.5);
 \draw[very thick] (2.5,9.5) -| (17.5,10.5) |- (10.5,10.5);
 \draw[very thick] (2.5,9.5) |- (1.5,5.5) |- (4.5,1.5) |- (29.5,3.5);
 \draw[very thick] (31.5,3.5) |- (35,3.5);
 \draw[very thick] (33.5,3.5) |- (31.5,11.5);
 \draw[very thick] (29.5,11.5) |- (25,11.5);

 \draw[very thick] (6.5,3.5) |- (6.5,4);
 \draw[very thick] (10.5,3.5) |- (10.5,4);
 \draw[very thick] (14.5,3.5) |- (14.5,4);

 \draw[very thick] (18.5,3.5) |- (18.5,4);

 \draw[very thick] (22.5,3.5) |- (22.5,4);

 \draw[very thick] (26.5,3.5) |- (26.5,4);
\draws{24}{12}
\draws{24}{11}
\draws{24}{10}
\draws{24}{9}
\draws{23}{9}
\draws{22}{9}
\draws{21}{9}
\draws{20}{9}
\draws{19}{9}
\draws{19}{10}
\draws{19}{11}
\draws{19}{12}
\draws{18}{12}
\draws{17}{12}
\draws{16}{12}
\draws{15}{12}
\draws{14}{12}
\draws{13}{12}
\draws{12}{12}
\draws{11}{12}

\draws{7}{12}
\draws{6}{12}
\draws{5}{12}
\draws{5}{11}
\draws{5}{10}
\draws{6}{10}
\draws{7}{10}

\draws{11}{10}
\draws{12}{10}
\draws{13}{10}
\draws{14}{10}
\draws{15}{10}
\draws{16}{10}
\draws{17}{10}
\draws{17}{9}
\draws{16}{9}
\draws{15}{9}
\draws{14}{9}
\draws{13}{9}
\draws{12}{9}
\draws{11}{9}
\draws{10}{9}
\draws{9}{9}
\draws{8}{9}
\drawblue{7}{9}
\drawblue{6}{9}
\drawblue{5}{9}
\drawblue{4}{9}
\drawblue{3}{9}
\drawblue{2}{9}
\drawblue{2}{8}
\drawblue{2}{7}
\drawblue{2}{6}
\drawblue{2}{5}
\drawblue{1}{5}
\drawblue{1}{4}
\drawblue{1}{3}
\drawblue{1}{2}
\drawblue{1}{1}
\drawblue{2}{1}
\drawblue{3}{1}
\drawblue{4}{1}
\drawblue{4}{2}
\drawdgray{4}{3}
\drawdgray{5}{3}
\drawdgray{6}{3}
\drawdgray{7}{3}
\drawgray{8}{3}
\drawgray{9}{3}
\drawgray{10}{3}
\drawgray{11}{3}
\drawgray{12}{3}
\drawgray{13}{3}
\drawgray{14}{3}
\drawgray{15}{3}
\drawgray{15}{3}
\drawgray{16}{3}
\drawgray{17}{3}
\drawgray{18}{3}
\drawgray{19}{3}
\drawgray{20}{3}
\drawgray{21}{3}
\drawgray{22}{3}
\drawgray{23}{3}
\drawgray{24}{3}
\drawgray{25}{3}
\drawgray{26}{3}
\drawgray{27}{3}
\drawt{28}{3}
\drawt{32}{3}
\drawt{33}{3}
\drawt{34}{3}

\drawred{33}{4}
\drawred{33}{5}
\drawred{33}{6}
\drawred{33}{7}
\drawred{33}{8}
\drawred{33}{9}
\drawred{33}{10}
\drawred{33}{11}
\drawred{32}{11}

\drawred{28}{11}
\drawred{27}{11}
\drawred{26}{11}
\drawred{25}{11}

\draw (4.5, 4.5) node {$P_0$};
\draw (8.5, 2.2) node {$P_{|P|-1}$};
\draw (27.5, 2.2) node {$P_{j}$};
\draw (32.5, 2.2) node {$P_{i}$};
\draw (30.5, 3.5) node {\Large $\ldots$};
\draw (9.5, 10.5) node {\Large $\ldots$};
\draw (9.5, 12.5) node {\Large $\ldots$};
\draw (30.5, 11.5) node {\Large $\ldots$};
\path [dotted, draw, thin] (0,0) grid[step=0.33cm] (35,13);
\end{tikzpicture}

a) The seed is in black, the simply pumpable path $P$ is in dark gray and its pumping is in light gray until index $j$ and in white afterwards, a blue path binds with the seed and $P_0$ and $Q$ (in red)  intersects with the seed.
\vspace{+0.5em}

\begin{tikzpicture}[x=0.33cm,y=0.33cm]

%colle
 \draw[very thick] (10.5,12.5) -| (19.5,9.5) -| (24.5,12.5);
 \draw[very thick] (8.5,10.5) -| (5.5,12.5) -| (8.5,12.5);
 \draw[very thick] (2.5,9.5) -| (17.5,10.5) |- (10.5,10.5);
 \draw[very thick] (2.5,9.5) |- (1.5,5.5) |- (4.5,1.5) |- (29.5,3.5);
 \draw[very thick] (31.5,3.5) |- (35,3.5);
 \draw[very thick] (33.5,3.5) |- (31.5,11.5);
 \draw[very thick] (29.5,11.5) |- (25,11.5);

 \draw[very thick] (6.5,3.5) |- (6.5,4);
 \draw[very thick] (10.5,3.5) |- (10.5,4);
 \draw[very thick] (14.5,3.5) |- (14.5,4);

 \draw[very thick] (18.5,3.5) |- (18.5,4);

 \draw[very thick] (22.5,3.5) |- (22.5,4);

 \draw[very thick] (26.5,3.5) |- (26.5,4);

\draw[orange!50!black,very thick] (18.5,3.5) |- (33.5,7.5) |- (33.5,3.5);
\draw[orange!50!black,very thick] (22.5,3.5) |- (30.5,5.5) |- (30.5,3.5);

\draws{24}{12}
\draws{24}{11}
\draws{24}{10}
\draws{24}{9}
\draws{23}{9}
\draws{22}{9}
\draws{21}{9}
\draws{20}{9}
\draws{19}{9}
\draws{19}{10}
\draws{19}{11}
\draws{19}{12}
\draws{18}{12}
\draws{17}{12}
\draws{16}{12}
\draws{15}{12}
\draws{14}{12}
\draws{13}{12}
\draws{12}{12}
\draws{11}{12}

\draws{7}{12}
\draws{6}{12}
\draws{5}{12}
\draws{5}{11}
\draws{5}{10}
\draws{6}{10}
\draws{7}{10}

\draws{11}{10}
\draws{12}{10}
\draws{13}{10}
\draws{14}{10}
\draws{15}{10}
\draws{16}{10}
\draws{17}{10}
\draws{17}{9}
\draws{16}{9}
\draws{15}{9}
\draws{14}{9}
\draws{13}{9}
\draws{12}{9}
\draws{11}{9}
\draws{10}{9}
\draws{9}{9}
\draws{8}{9}
\drawblue{7}{9}
\drawblue{6}{9}
\drawblue{5}{9}
\drawblue{4}{9}
\drawblue{3}{9}
\drawblue{2}{9}
\drawblue{2}{8}
\drawblue{2}{7}
\drawblue{2}{6}
\drawblue{2}{5}
\drawblue{1}{5}
\drawblue{1}{4}
\drawblue{1}{3}
\drawblue{1}{2}
\drawblue{1}{1}
\drawblue{2}{1}
\drawblue{3}{1}
\drawblue{4}{1}
\drawblue{4}{2}
\drawdgray{4}{3}
\drawdgray{5}{3}
\drawdgray{6}{3}
\drawdgray{7}{3}
\drawgray{8}{3}
\drawgray{9}{3}
\drawgray{10}{3}
\drawgray{11}{3}
\drawgray{12}{3}
\drawgray{13}{3}
\drawgray{14}{3}
\drawgray{15}{3}
\drawgray{15}{3}
\drawgray{16}{3}
\drawgray{17}{3}
\drawgray{18}{3}
\drawgray{19}{3}
\drawgray{20}{3}
\drawgray{21}{3}
\drawgray{22}{3}
\drawgray{23}{3}
\drawgray{24}{3}
\drawgray{25}{3}
\drawgray{26}{3}
\drawgray{27}{3}
\drawt{28}{3}
\drawt{32}{3}
\drawt{33}{3}
\drawt{34}{3}

\drawo{18}{4}
\drawo{18}{5}
\drawo{18}{6}
\drawo{18}{7}
\drawo{19}{7}
\drawo{20}{7}
\drawo{21}{7}
\drawo{22}{7}
\drawo{23}{7}
\drawo{24}{7}
\drawo{25}{7}
\drawo{26}{7}
\drawo{27}{7}
\drawo{28}{7}
\drawo{29}{7}
\drawo{30}{7}
\drawo{31}{7}
\drawo{32}{7}

\drawo{22}{4}
\drawo{22}{5}
\drawo{23}{5}
\drawo{24}{5}
\drawo{25}{5}
\drawo{26}{5}
\drawo{27}{5}
\drawo{28}{5}
\drawo{29}{5}
\drawo{30}{5}
\drawo{30}{4}

\drawred{33}{4}
\drawred{33}{5}
\drawred{33}{6}
\drawred{33}{7}
\drawred{33}{8}
\drawred{33}{9}
\drawred{33}{10}
\drawred{33}{11}
\drawred{32}{11}

\drawred{28}{11}
\drawred{27}{11}
\drawred{26}{11}
\drawred{25}{11}

\draw (4.5, 4.5) node {$P_0$};
\draw (8.5, 2.2) node {$P_{|P|-1}$};
\draw (27.5, 2.2) node {$P_{j}$};
\draw (32.5, 2.2) node {$P_{i}$};
\draw (30.5, 3.5) node {\Large $\ldots$};
\draw (9.5, 10.5) node {\Large $\ldots$};
\draw (9.5, 12.5) node {\Large $\ldots$};
\draw (30.5, 11.5) node {\Large $\ldots$};
\path [dotted, draw, thin] (0,0) grid[step=0.33cm] (35,13);
\end{tikzpicture}

b) If a path growing on the pumping between indices $|P|-1$ and $j-|P|-1$ intersects with $Q$ or the white part of the pumping then we have an arc of width at least $|P|$.
\vspace{+0.5em}

\begin{tikzpicture}[x=0.33cm,y=0.33cm]

%colle
 \fill[red!15!white] (1.5,1.5) -| (4.5,3.5) -| (33.5,11.5) -| (24.5,9.5) -| (19.5,12.5) -| (5.5,10.5) -| (17.5,9.5) -| (2.5,5.5) -| (1.5,1.5);

 \draw[very thick] (10.5,12.5) -| (19.5,9.5) -| (24.5,12.5);
 \draw[very thick] (8.5,10.5) -| (5.5,12.5) -| (8.5,12.5);
 \draw[very thick] (2.5,9.5) -| (17.5,10.5) |- (10.5,10.5);
 \draw[very thick] (2.5,9.5) |- (1.5,5.5) |- (4.5,1.5) |- (29.5,3.5);
 \draw[very thick] (31.5,3.5) |- (35,3.5);
 \draw[very thick] (33.5,3.5) |- (31.5,11.5);
 \draw[very thick] (29.5,11.5) |- (25,11.5);

 \draw[very thick] (6.5,3.5) |- (6.5,9);
 \draw[very thick] (10.5,3.5) |- (10.5,9);
 \draw[very thick] (14.5,3.5) |- (14.5,9);

 \draw[very thick] (18.5,3.5) |- (10.5,11.5);
 \draw[very thick] (6.5,11.5) |- (8.5,11.5);

 \draw[very thick] (22.5,3.5) |- (22.5,9);

\draws{24}{12}
\draws{24}{11}
\draws{24}{10}
\draws{24}{9}
\draws{23}{9}
\draws{22}{9}
\draws{21}{9}
\draws{20}{9}
\draws{19}{9}
\draws{19}{10}
\draws{19}{11}
\draws{19}{12}
\draws{18}{12}
\draws{17}{12}
\draws{16}{12}
\draws{15}{12}
\draws{14}{12}
\draws{13}{12}
\draws{12}{12}
\draws{11}{12}

\draws{7}{12}
\draws{6}{12}
\draws{5}{12}
\draws{5}{11}
\draws{5}{10}
\draws{6}{10}
\draws{7}{10}

\draws{11}{10}
\draws{12}{10}
\draws{13}{10}
\draws{14}{10}
\draws{15}{10}
\draws{16}{10}
\draws{17}{10}
\draws{17}{9}
\draws{16}{9}
\draws{15}{9}
\draws{14}{9}
\draws{13}{9}
\draws{12}{9}
\draws{11}{9}
\draws{10}{9}
\draws{9}{9}
\draws{8}{9}
\drawblue{7}{9}
\drawblue{6}{9}
\drawblue{5}{9}
\drawblue{4}{9}
\drawblue{3}{9}
\drawblue{2}{9}
\drawblue{2}{8}
\drawblue{2}{7}
\drawblue{2}{6}
\drawblue{2}{5}
\drawblue{1}{5}
\drawblue{1}{4}
\drawblue{1}{3}
\drawblue{1}{2}
\drawblue{1}{1}
\drawblue{2}{1}
\drawblue{3}{1}
\drawblue{4}{1}
\drawblue{4}{2}
\drawdgray{4}{3}
\drawdgray{5}{3}
\drawdgray{6}{3}
\drawdgray{7}{3}
\drawgray{8}{3}
\drawgray{9}{3}
\drawgray{10}{3}
\drawgray{11}{3}
\drawgray{12}{3}
\drawgray{13}{3}
\drawgray{14}{3}
\drawgray{15}{3}
\drawgray{15}{3}
\drawgray{16}{3}
\drawgray{17}{3}
\drawgray{18}{3}
\drawgray{19}{3}
\drawgray{20}{3}
\drawgray{21}{3}
\drawgray{22}{3}
\drawgray{23}{3}
\drawgray{24}{3}
\drawgray{25}{3}
\drawgray{26}{3}
\drawgray{27}{3}
\drawt{28}{3}
\drawt{32}{3}
\drawt{33}{3}
\drawt{34}{3}

\drawlr{6}{4}
\drawlr{6}{5}
\drawlr{6}{6}
\drawlr{6}{7}
\drawlr{6}{8}

\drawlr{10}{4}
\drawlr{10}{5}
\drawlr{10}{6}
\drawlr{10}{7}
\drawlr{10}{8}

\drawlr{14}{4}
\drawlr{14}{5}
\drawlr{14}{6}
\drawlr{14}{7}
\drawlr{14}{8}

\drawr{18}{4}
\drawr{18}{5}
\drawr{18}{6}
\drawr{18}{7}
\drawr{18}{8}
\drawr{18}{9}
\drawr{18}{10}
\drawr{18}{11}
\drawr{17}{11}
\drawr{16}{11}
\drawr{15}{11}
\drawr{14}{11}
\drawr{13}{11}
\drawr{12}{11}
\drawr{11}{11}
\drawr{7}{11}
\drawr{6}{11}

\drawlr{22}{4}
\drawlr{22}{5}
\drawlr{22}{6}
\drawlr{22}{7}
\drawlr{22}{8}

\drawred{33}{4}
\drawred{33}{5}
\drawred{33}{6}
\drawred{33}{7}
\drawred{33}{8}
\drawred{33}{9}
\drawred{33}{10}
\drawred{33}{11}
\drawred{32}{11}

\drawred{28}{11}
\drawred{27}{11}
\drawred{26}{11}
\drawred{25}{11}

\draw (4.5, 4.5) node {$P_0$};
\draw (8.5, 2.2) node {$P_{|P|-1}$};
\draw (27.5, 2.2) node {$P_{j}$};
\draw (32.5, 2.2) node {$P_{i}$};
\draw (30.5, 3.5) node {\Large $\ldots$};
\draw (9.5, 10.5) node {\Large $\ldots$};
\draw (9.5, 11.5) node {\Large $\ldots$};
\draw (9.5, 12.5) node {\Large $\ldots$};
\draw (30.5, 11.5) node {\Large $\ldots$};
\path [dotted, draw, thin] (0,0) grid[step=0.33cm] (35,13);
\end{tikzpicture}

c) All the translations of $Q$ growing on the light gray part of the pumping start in the red area of the grid. The seed and the blue path cannot block them all without creating intersections. One of the translations must fully grow and becomes pumpable. Its pumping must leave the finite red area.

\caption{Illustration of proof \ref{lem:magicindice}.}
\label{fig:magicIndex}
\end{figure}

Consider a path $Q$ which grows on $\pump{P}$ at position $P_k$ with $k\geq i$ and for the sake of contradiction, suppose that this path is not in $\inuniterm$ which implies that $Q$ conflicts with $\ass$ and by the definition of $i$ the vertical height or the horizontal width of $Q$ is at least $\funpump{|T|}{|\gamma|}+1$. Let $m=\min\{n:\pos{Q_n} \in \domain{\ass}\}$ and by definition of $m$, $Q_{0,1,\ldots,m-1}$ is in $\inuniterm$. Consider a finite simple cycle $C$ which is made of the binding path of $\pump{P}_{0,1,\ldots,k}$, the binding path of $Q_{0,1,\ldots, m}$ and a path in the binding graph of $\ass$ which links $\pos{P_0}$ to $\pos{Q_m}$. Let $R$ be the translation of $Q$ by $\ell\vect{P_0P_{|P|-1}}$ for some $\ell\in\N$ such that $R_0$ is a tile of $P+\vect{P_0P_{|P|-1}}$. Remark that either for all $\ell \in \N$, $R+\ell$ turns left of $\pump{P}$ or for all $\ell \in \N$, $R+\ell$ turns right of $\pump{P}$. In both cases, for all $\ell\in\N$, the tile $R_1+\ell\vect{P_0P_{|P|-1}}$ is in the interior of the cycle $C$. Remark that, there are at most $4|\ass|$ positions which are neighbors of a tile of $\ass$ which implies that if for all $0 \leq \ell \leq 4|\ass|$, the path $R+\ell\vect{P_0P_{|P|-1}}$ conflicts with $\ass$ then there exists $\ell \neq \ell' \in \N$  such that $R+\ell\vect{P_0P_{|P|-1}}$ and $R+\ell'\vect{P_0P_{|P|-1}}$ intersect before intersecting with $\ass$ which contradicts corollary \ref{cor:selfavoid}. 
%By corollary \ref{cor:selfavoid}, both of these paths grow of $\pump{P}$ and then by lemma \ref{lem:transpath} and lemma \ref{lem:transarc}, there exists an arc of width at least $|P|$ growing on $\pump{P}$ which contradicts lemma \ref{lem:noarc}. 
Thus, there exists $0 \leq \ell \leq 4|\ass|$ such that $R+\ell\vect{P_0P_{|P|-1}}$ does not intersect with $\ass$. Moreover, by definition of $R$, $R_0+\ell\vect{P_0P_{|P|-1}}$ is a tile of $\pump{P}_{|P|-1,|P|,\ldots,(4|\ass|+1)(|P|-1)}$ and then the path $R_{1,2,\ldots,|R|-1}+\ell\vect{P_0P_{|P|-1}}$ is producible by $(T,\gamma,1)$ and is in $\inuniterm$. Since the vertical height or the horizontal width of $R$ is at least $\funpump{|T|}{|\gamma|}$, by the pumping lemma \ref{theorem:pumping}, $R_{1,2,\ldots,|R|-1}+\ell\vect{P_0P_{|P|-1}}$ is infinitely pumpable and let $S$ be its pumping (see definition \ref{def:pumpingPbetweeniandj}). By definition of a pumping $S_0=R_1+\ell\vect{P_0P_{|P|-1}}$ and $S$ is producible by $(T,\gamma,1)$ and thus cannot intersect with $\gamma$. Moreover, since $S$ is infinite, it cannot stay in the interior of $C$ and thus it must intersect with either $Q_{0,1,\ldots, m-1}$ or $\pump{P}_{j+1,j+1,\ldots,k}$. In the first case, tile $R_0+\ell\vect{P_0P_{|P|-1}}$, path $Q_{0,1,\ldots, m-1}$ and path $S$ assemble an arc of $\pump{P}$ of width at least $i-j\geq|P|$ (since $R_0+\ell\vect{P_0P_{|P|-1}}$ is a tile of $\pump{P}_{0,1,\ldots,j}$ and $k\geq i>j+(|P|-1)$) and in the second case, tile $R_0+\ell\vect{P_0P_{|P|-1}}$ and a prefix of $S$ assemble an arc of $\pump{P}$ of width at least $|P|$ (since $R_0+\ell\vect{P_0P_{|P|-1}}$ is a tile of $\pump{P}_{0,1,\ldots,(4|\ass|+1)(|P|-1)}$ and $j-(4|\ass|+1)(|P|-1)>|P|$). Both cases contradict lemma~\ref{lem:noarc}.
\end{proof}

\subsubsection{Proving the last main theorems}

We focus on the terminal assembly $\uniterm$ of an aperiodic deterministic tile assembly system $\mathcal{T}=(T,\sigma,1)$ (Appendix \ref{app:C} shows the analysis of the tile assembly system described in Figure \ref{fig:appC:step1} and its terminal assembly is shown in Figure \ref{fig:appC:step2}). Firstly, from the seed, we grow all the finite paths which are not infinitely pumpable and all the prefixes of the infinitely pumpable paths until the end of the first simply pumpable path which appears in these paths (excluding the last tile of the simply pumpable path) as explained in Theorem \ref{last:theorem} and shown in Figure \ref{fig:appC:step3}. The pumping lemma provides a bound on the length of these paths and thus we obtain a finite assembly (the gray one in Figure \ref{fig:appC:step3}). The complexity of the assemblies growing at the end of the simply pumpable paths is less than $2$ (the red, green, orange and blue ones in Figure \ref{fig:appC:step3}). Secondly, we study these assemblies of complexity~$2$ as explained in lemma \ref{lem:compdeux}. Lemma \ref{lem:magicindice} allow us to find an index where there is no more intersection with the previous assembly (the gray one in Figure \ref{fig:appC:step4}) after this index. The analysis of the assemblies of complexity $2$ is made in two parts according to this index: the assembly growing after this index is obtained by "pumping" an assembly of complexity less than $1$; before this index, the previous assembly (the gray one in Figure \ref{fig:appC:step4}) may block the grow of some paths, then we use a finite assembly to memorize what is happening in the finite area where these interactions may occur and only some assemblies of complexity $1$ may grow on this finite assembly (see the red and green assemblies in Figure \ref{fig:appC:step4}). We proceed similarly to analyze the assemblies of complexity $1$ as explained in lemma \ref{lem:compun} and shown in Figure \ref{fig:appC:step5}. No more pumpable paths can grow after this step as explained in lemma \ref{lem:compzero}. 

Note that the algorithm explained here is described in reverse order in the following lemmas where we start by analyzing the assemblies of complexity $0$, then $1$, then $2$ and finally we analyze the terminal assembly in the two last remaining theorem.

\begin{lemma} 
\label{lem:compzero}
Consider a tile assembly system $\mathcal{T}=(T,\sigma,1)$. If a path $P$ is $(\vect{u},\vect{v})$-self-avoiding, then $P$ is finite and its length is bounded by a function depending only of $||\vect{u}||$ and $||\vect{v}||$.
\end{lemma} 

\begin{proof}
We consider $\R^2$, the continuous $2D$ plane. Consider the line $L$ (resp. $L'$) passing by $(0,0)$ of direction $\vect{u}$ (resp. $\vect{v}$). Consider the finite cycle $C$ which is the polygon defined by the four points $(0,0)$, $(0,0)+\vect{u}$, $(0,0)+\vect{u}+\vect{v}$, $(0,0)+\vect{v}$. Let $t\in \N$ be the number of positions of $\Z^2$ which are inside the interior of $C$, note that $t$ is correctly defined since the area of the interior of $C$ is finite. Moreover, $t$ is bounded by a function depending only of $||\vect{u}||$ and $||\vect{v}||$. Without loss of generality, we can suppose that $L(\R)$ (resp. $L'(\R)$) is included into the left hand side of $L+\vect{v}$  (resp. $L'+\vect{u}$). Then, the interior of this cycle is the intersection between the right hand side of $L$, the left hand side of $L+\vect{v}$, the right hand side of $L'$ and the left hand side of $L'+\vect{u}$. Thus any for any $(x,y) \in \Z^2$ there exists $\ell \in Z$ and $\ell'\in Z$ such that $(x,y)+\ell\vect{u}+\ell'\vect{v}$ is in the interior of $C$. Finally, if $|P|$ is greater than $t$ then there exists $i,j$ and $\ell,\ell' \in \Z$ such that $\pos{P_i}=\pos{P_{j}}+\ell\vect{u}+\ell'\vect{v}$ and $P$ is not $(\vect{u},\vect{v})$-self-avoiding.
\end{proof}

\begin{lemma} 
\label{lem:compun}
Consider a tile assembly system $\mathcal{T}=(T,\sigma,1)$ whose terminal assembly is $\uniterm$, a simply pumpable path $P$ in $\inuniterm$ which is $\vect{v}$-self-avoiding (with $\vect{v}$ not collinear with $\vect{P_0P_{|P|-1}}$) and a producible finite assembly $\ass$ such that $P_0$ is a tile of $\ass$. Then the union of $\pump{P}$ and all paths of $\inuniterm$ growing on $\pump{P}$ on a tile $\pump{P}_i$ with $i\geq|P|-1$ is an assembly of complexity $1$ whose domain is a semilinear set whose size depends of $|\beta|$, $\vect{v}$, $P$ and $|T|$. 
\end{lemma} 

\begin{proof}
All paths growing on $\pump{P}$ on a tile $\pump{P}_i$ with $i\geq|P|-1$ are $(\vect{v},\vect{P_0P_{|P|-1}})$-self-avoiding (by lemma \ref{lem:doubleavoid}) and are thus finite and their length bounded by a function which depends only of $\vect{v}$ and $P$ (by lemma \ref{lem:compzero}). We remind that a finite path has a complexity of $0$. By lemma \ref{lem:magicindice}, there exists an index $j$ such that all paths growing on $\pump{P}$ at index $k>j$ are in $\inuniterm$, this index depends on $|\beta|$, $P$ and $|T|$. Remark, that there is only a finite number of different paths of $\inuniterm$ which can grow on $\pump{P}$ at index $|P|-1\leq k\leq j-1$ since their length is bounded by a function which depends only of $\vect{v}$ and $P$. Thus, we consider the assembly $\gamma_0$ which is the union of $\pump{P}_{0,\ldots,j}$ and all paths of $\inuniterm$ growing on $\pump{P}$ at index $|P|-1\leq k \leq  j-1$, this assembly is a finite union of finite paths and is thus finite (and its complexity is $0$), its size depends of $|\beta|$, $\vect{v}$, $P$ and $|T|$ and by its definition it is a subassembly of $\alpha$. Similarly, consider the assembly $\gamma_1$ which is of the union of $\pump{P}_{j,j+1,\ldots,j+|P|-1}$ and all paths of $\inuniterm$ growing on $\pump{P}$ at index $j\leq k \leq j+|P|-2$, this assembly is a finite union of finite paths and is thus finite (and its complexity is $0$), its size depends of $|\beta|$, $\vect{v}$, $P$ and $|T|$ and by its definition it is a subassembly of $\alpha$. 
Remark that for all $\ell \in \N$, $\gamma_1+\ell\vect{P_0P_{|P|+1}}$ intersects with $\gamma_1+(\ell+1)\vect{P_0P_{|P|+1}}$ and that $\gamma_1+\ell\vect{P_0P_{|P|+1}}$ is sub-assembly of $\uniterm$ by definition of $j$ and by corollary \ref{cor:selfavoid} and then the assembly $\bigcup_{\ell \in \N}(\gamma_1+\ell\vect{P_0P_{|P|+1}})$ is correctly defined. 
Moreover, this assembly is of complexity $1$. Consider the assembly: $$\gamma=\gamma_0 \cup \left(\bigcup_{\ell \in \N}(\gamma_1+\ell\vect{P_0P_{|P|+1}})\right).$$ This assembly is a subassembly of $\uniterm$ and for any path of $\inuniterm$ growing on $\pump{P}$ on a tile $\pump{P}_i$ with $i\geq|P|-1$ either this path is a subassembly of $\gamma_0$ or there is $\ell \in \N$ such that the translation of this path by $-\ell\vect{P_0P_{|P|+1}}$ is a subassembly of $\gamma_1$. Thus the union of all paths of $\inuniterm$ growing on $\pump{P}$ at index $i\geq|P|-1$ is $\gamma$ which is the union of an assembly of complexity $0$ with an assembly of complexity $1$ and thus is an assembly of complexity $1$ whose domain is a semilinear set whose size depends of $|\beta|$, $\vect{v}$, $P$ and $|T|$.
\end{proof}

\begin{lemma} 
\label{lem:compdeux}
Consider a tile assembly system $\mathcal{T}=(T,\sigma,1)$ whose terminal assembly is $\uniterm$, a simply pumpable path $P$ in $\inuniterm$ and a producible finite assembly $\ass$ such that $P_0$ is a tile of $\ass$. Then the union of $\pump{P}$ and all paths of $\inuniterm$ growing on $\pump{P}$ on a tile $\pump{P}_i$ with $i\geq|P|-1$ is an assembly of complexity $2$ whose domain is a semilinear set whose size depends of $|\beta|$, $P$ and $|T|$. 
\end{lemma} 

\begin{proof}
By lemma \ref{lem:magicindice}, there exists an index $j$ such that all paths growing on $\pump{P}$ at index $k>j$ are in $\inuniterm$, this index depends on $|\beta|$, $P$ and $|T|$. We define $\beta'$ as $\beta\cup \pump{P}_{0,1,\ldots, j+(|P|-1)}$. 
%%%%
Consider a path $Q$ which grows on $\pump{P}$ at index $|P|-1\leq i \leq j+|P|-1$. By lemma \ref{cor:selfavoid}, $Q$ is $\vect{P_0P_{|P|-1}}$-self-avoiding. By the pumping lemma \ref{theorem:pumping}, either $Q$ is finite and its vertical height or horizontal width is less then $\funpump{|T|}{|\beta'|}$ or it is pumpable between $m$ and $n$ for some $1 \leq m <n \leq |Q|-1$. In the second case, without loss of generality, we can assume that the vertical height or horizontal width of $Q_{0,1,\ldots,n}$ is less then $\funpump{|T|}{|\beta'|}$ (otherwise there exist some smaller indices with the same properties as $m$ and $n$) and then the union of $Q_{0,1,\ldots,n}$ and all paths of $\inuniterm$ which grow on $Q$ on index $n$ is an assembly of complexity $1$ (by lemma \ref{lem:compun}). Remark that, there are only a finite number of different paths of $\inuniterm$ of vertical height or horizontal width bounded by $\funpump{|T|}{|\beta'|}$ which can grow on $\pump{P}_{|P|-1,\ldots,j+|P|-2}$. 
%%%%
Thus, we consider the assembly $\gamma_0$ which is the union of $\pump{P}_{0,\ldots,j}$ and all paths of $\inuniterm$ growing on $\pump{P}$ at index $|P|-1\leq i \leq  j-1$, this assembly is a finite union of assemblies of complexity $1$ and is thus of complexity $1$ and its domain is a semilinear set whose size is bounded by a function which depends only of $|\beta'|$, $P$ and $|T|$. Moreover, by its definition it is a subassembly of $\alpha$. Similarly, consider the assembly $\gamma_1$ which is of the union of $\pump{P}_{j,j+1,\ldots,j+|P|-1}$ and all paths of $\inuniterm$ growing on $\pump{P}$ at index $j\leq i \leq j+|P|-2$, the complexity of this assembly is $1$ and its domain is a semilinear set whose size is bounded by a function which depends only of $|\beta'|$, $P$ and $|T|$. Moreover, by its definition it is a subassembly of $\alpha$. 
Remark that for all $\ell \in \N$, $\gamma_1+\ell\vect{P_0P_{|P|+1}}$ intersects with $\gamma_1+(\ell+1)\vect{P_0P_{|P|+1}}$ and that $\gamma_1+\ell\vect{P_0P_{|P|+1}}$ is sub-assembly of $\uniterm$ by definition of $j$ and by corollary \ref{cor:selfavoid} and then the assembly $\bigcup_{\ell \in \N}(\gamma_1+\ell\vect{P_0P_{|P|+1}})$ is correctly defined. 
Moreover, this assembly is of complexity $2$. Consider the assembly: $$\gamma=\gamma_0 \cup \left(\bigcup_{\ell \in \N}(\gamma_1+\ell\vect{P_0P_{|P|+1}})\right).$$ This assembly is a subassembly of $\uniterm$ and for any path of $\inuniterm$ growing on $\pump{P}$ at index $i\geq|P|-1$ either this path is a subassembly of $\gamma_0$ or there is $\ell \in \N$ such that the translation of this path by $-\ell\vect{P_0P_{|P|+1}}$ is a subassembly of $\gamma_1$. Thus the union of all paths of $\inuniterm$ growing on $\pump{P}$ at index $i\geq|P|-1$, is $\gamma$ which is the union of an assembly of complexity $1$ with an assembly of complexity $2$ and thus is an assembly of complexity $2$ whose domain is a semilinear set whose size depends of $|\beta'|$, $P$ and $|T|$. Remark that $|\beta'|$ depends on $|\beta|$ and $j$ and that $j$ depends on $|\beta|$, $P$ and $T$, hence the result.
\end{proof}

We can now prove the second half of the main theorem

\begin{theorem} 
Consider a directed tile assembly system $\mathcal{T}=(T,\sigma,1)$ whose terminal assembly $\uniterm$ is simply periodic. Then there exists a vector $\vect{v}$ and an assembly $\ass$ of complexity $1$ such that $$\uniterm=\bigcup_{\ell\in \Z} (\ass+\ell\vect{v}).$$
\end{theorem} 

\begin{proof}
Consider the paths $P^+$ and $P^-$, the assembly $\ass$ and the vector $\vect{v}=\vect{P^+_0P^+_{|P^+|-1}}=\vect{P^-_0P^-_{|P^-|-1}}$ of lemma \ref{lem:firsthalf}. Consider a path $Q$ which grows on the left side of $\bipump{P^+}$ and by lemma \ref{lem:bipump:seed}, $\uniterm$ is the unique terminal assembly of $(T,Q_0,1)$. By lemma \ref{lem:firsthalf}, $Q$ is $\vect{P^+_0P^+_{|P^+|-1}}$-self-avoiding. By the pumping lemma \ref{theorem:pumping}, either $Q$ is finite and its vertical height or horizontal width is less then $\funpump{T}{1}$ or it is pumpable between $i$ and $j$ for some $1 \leq i <j \leq |Q|-1$. In the second case, all paths of $\inuniterm$ which grow on tile $Q_j$ of $Q_{0,1\ldots,j}$ belongs to an assembly of complexity $1$ whose domain is a semilinear set whose size depends of $P^+$, $Q$ and $|T|$. Moreover, without loss of generality we can suppose that the vertical height or horizontal width of $Q_{0,1\ldots,j}$ is less then $\funpump{T}{1}$. Since there exist only a finite number of paths of vertical height or horizontal width bounded by $\funpump{T}{1}$, we can consider the assembly $\ass^+$ which is the union of $P^+$ and all the paths growing on $\bipump{P^+}$ at index $0\leq k < |P|-1$. The complexity of $\ass^+$ is $1$ since it is the finite union of assembly of complexity less than $1$. Moreover, $\ass^-$ can be defined similarly for the right side of $\bipump{P^-}$ and let $$\gamma=\ass^+\cup\ass\cup\ass^-.$$ Since $\alpha$ is $\vect{v}$-periodic, we have $\uniterm=\cup_{\ell\in \Z} (\gamma+\ell\vect{v})$.
\end{proof}

We can now prove the last main theorem

\begin{theorem} 
\label{last:theorem}
Consider a directed tile assembly system $\mathcal{T}=(T,\sigma,1)$ whose terminal assembly $\uniterm$ is aperiodic, then the complexity of $\uniterm$ is $2$.
\end{theorem} 

\begin{proof}
Consider the assembly $\ass$ which the union of the seed and all producible paths of $(T,\sigma,1)$ whose vertical height and horizontal width are less than $\funpump{T}{|\sigma|}$. If this assembly is terminal then $\uniterm=\ass$ is finite and its complexity is $0$. %Now, consider the set of tiles $\textrm{Removed}$ defined as the tiles of $\alpha$ which are not tiles of $\ass$. The set $\textrm{Removed}$ may not be an assembly since it may not be connected but $\textrm{Removed}\cup\ass=\alpha$.
%Otherwise, only the tiles of $\ass$ which are at distance $\funpump{T}{|\sigma|}+1$ of $\sigma$ can have a free glue (otherwise, this glue could be used to bind a tile which would contradicts the definition of $\ass$). 
%Let $\textrm{Pumpable}$ be initially defined as the empty set $\emptyset$. 
Let $\gamma=\ass$ and for any producible path $P$ of $(T,\sigma,1)$ such that $P$ is infinitely pumpable and $\asm{P}$ is a subassembly of $\ass$ do the following: 
\begin{itemize}
\item find the smallest index $j$ such that there exists $0\leq i<j$ such that $P$ is pumpable between $i$ and $j$;
\item $\gamma=\gamma \cup \pump{P_{i,i+1,\ldots,j}} \cup \{Q:Q \in \inuniterm \text{ and Q grows on } \pump{P_{i,i+1,\ldots,j}} \text{ at index $k\geq j$} \}$;
\end{itemize}
Since $\ass$ is finite then the number of producible path which are subassembly of $\ass$ is finite. By lemma \ref{lem:compdeux}, the complexity of the assembly $\pump{P_{i,i+1,\ldots,j}} \cup \{Q:Q \in \inuniterm \text{ and Q grows on } \pump{P_{i,i+1,\ldots,j}} \text{ at index $k\geq j$}\}$ is less than $2$ and thus at the end of this algorithm the complexity of the assembly $\gamma$ is less than $2$. Moreover, we claim that $\alpha=\gamma$ which would conclude this result. Since $\gamma$ is the union of paths of $\inuniterm$ then $\gamma$ is a subassembly of $\alpha$. For any tile $A$ of $\alpha$ either $A$ is a tile of $\ass$ (and thus of $\gamma$) or there exists a producible path $P$ of $(T,\sigma,1)$ such that $P_{|P|-1}=A$. If $A$ is not a tile of $\ass$, then the vertical height or horizontal width of $P$ is more than $\funpump{T}{|\sigma|}$ by definition of $\ass$ and thus there exists a prefix of $P$ whose vertical height or horizontal width is $\funpump{T}{|\sigma|}$. Again by definition of $\ass$,  this prefix is a subassembly of $\ass$. This prefix is infinitely pumpable by the pumping lemma \ref{theorem:pumping}. Then there exist $0\leq i <j \leq |P|-1$, such that $P$ is pumpable between $i$ and $j$ and $\asm{P_{0,1,\ldots, j}}$ is a subassembly of  $\ass$. Without loss of generality we can suppose that $j$ is minimal and then either $A$ is a tile of $\pump{P_{i,i+1,\ldots, j}}$ or a suffix of $P_{j,j+1,\ldots, |P|-1}$ grows on $\pump{P_{i,i+1,\ldots, j}}$ at index $k>j$. In both cases, $A$ is a tile of $\gamma$.

\end{proof}

\appendix

\section{Discussion about the seed}
\label{app:seed}

In \cite{Doty-2011}, it was assumed that the seed of a directed tile assembly system at temperature $1$ could be reduced to a single tile without loss of generality. In lemma~\ref{lem:bipump:seed}, we show that this statement is true when the terminal assembly is periodic but this statement is not true in the general case. Figure \ref{fig;appA} exhibits a counter-example. 

Our counter example is made of four kind of tile types, see Figure \ref{fig;appA}a. The first kind of tile types are the ten black tiles which constitute the seed. The second and third ones are the three blue tiles and the three green tiles which are used to assemble two simply pumpable paths which can bind on the seed (using glues $g_2$ and $g_1$). The fourth set of tile types are the eleven red tiles which are used to assemble a path whose extremities bind with the first iteration of the two simply pumpable paths (using glues $g_3$ and $g_4$).

Figure \ref{fig;appA}b shows the unique terminal assembly of this tile assembly system. Remark that only the first iteration of the red link can fully grow because all of its translations on the simply pumpable paths quickly collide with a previous period. Nevertheless, this link creates a cycle and if the seed is reduced to a single tile, as shown in Figure \ref{fig;appA}c where this single tile is pointed by a red arrow, then it is possible to grow one half of the seed, then the green pumpable path, the red path and finally the blue path in both direction which leads to another terminal assembly where there is a blue tile where a black one should be (pointed by a black arrow in the Figure). If another initial tile is chosen, either the same technique leads to a conflict or we can grow an assembly in order to pump the green pumpable path first which will lead to a conflict.

Figure \ref{fig;appA}d shows how the result of section \ref{subsec:pumpablepath} describes this terminal assembly using only a finite amount of information: 
\begin{itemize}
\item there is a finite assembly in gray made of the seed, the red path and the first iteration of the green and blue pumpable paths. 
\item a vector $\vect{v}$ and a purple finite assembly made of one period of the green pumpable path and some red tiles, for all $i \in \N$ the translation of this assembly by $i\vect{v}$ is in the terminal assembly.
\item a vector $\vect{w}$ and an orange finite assembly made of one period of the blue pumpable path and some red tiles, for all $i \in \N$ the translation of this assembly by $i\vect{w}$ is in the terminal assembly.
\end{itemize}

%Note that this remark does not invalidate the results of \cite{Doty-2011}. 
The interest of this remark is to show that the first period of a pumpable path may still assemble some artifacts which do not appear on the rest of the pumpable path.

\section{Wee lemmas and left/right turns for curves}

\begin{lemma}[Lemma 6.3 of~\cite{OneTile}]
  \label{lem:precious}
  Consider a two-dimensional, bounded, connected, regular closed set $S$, i.e. $S$ is equal to the topological closure of its interior points. Suppose $S$ is translated by a vector $v$ to obtain shape $S_v$, such that $S$ and $S_v$ do not overlap. Then the shape $S_{c*v}$ obtained by translating $S$ by $c*v$ for any integer $c\neq 0$ also does not overlap $S$.
\end{lemma}

The following lemma~\cite{STOC2017} formalizes the intuition behind Definition~\ref{def:pumpingPbetweeniandj}:
\begin{lemma}[Lemma 2.5 of \cite{STOC2017}]
  \label{lem:torture}
  Let $P$ be a path with tiles from some tileset $T$, $i<j$ be two integers, and $\olq$ be the pumping of $P$ between $i$ and $j$.
  Then for all integers $k\geq i$, $\olq_{k+(j-i)} = \olq_k + \vect{P_iP_j}$.
\end{lemma}
\begin{proof}
  By the definition of $\olq$ (and using the fact that $(j-i) \mod (j-i)=0$):
  \begin{eqnarray*}
    \olq_{k + (j-i)} &=& \torture P {k+(j-i)} i j\\
        &=& \torture P k i j + \vect{P_iP_j}\\
        &=& \olq_{k}+ \vect{P_iP_j}
      \end{eqnarray*}
\end{proof}

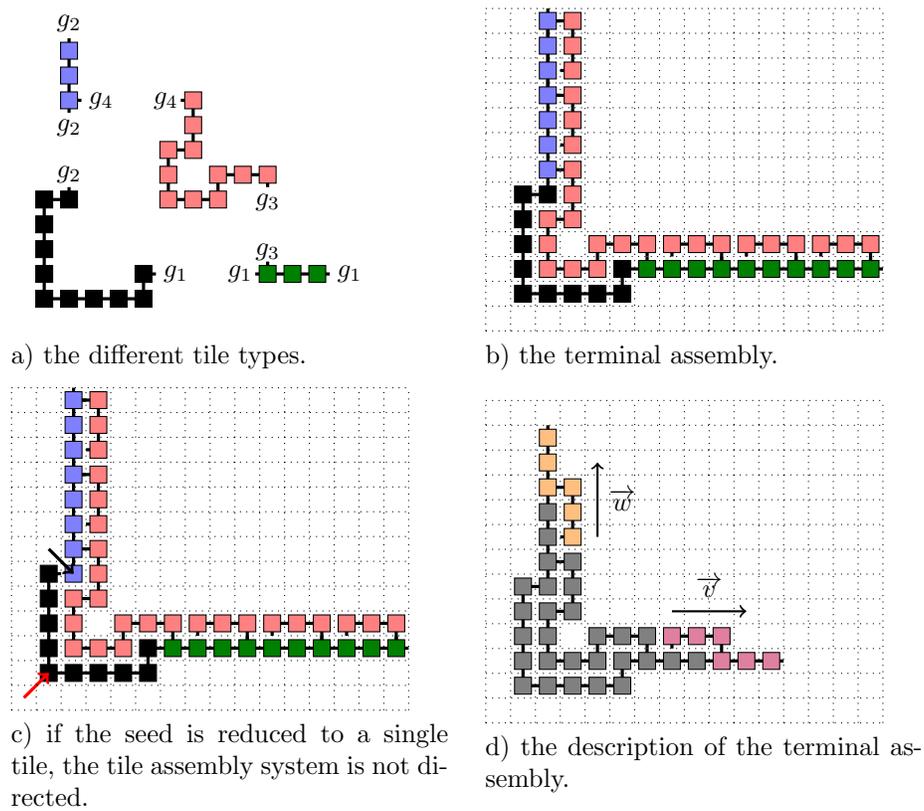
\begin{figure}
\centering

\begin{minipage}{0.45\linewidth}
\begin{tikzpicture}[x=0.33cm,y=0.33cm]

%colle
 \draw[very thick] (2.5,6) |- (1.5,5.5) |- (5.5,1.5) |- (6,2.5);

\drawwhite{0}{0}
\draws{1}{1}
\draws{2}{1}
\draws{3}{1}
\draws{4}{1}
\draws{5}{1}
\draws{5}{2}
\draws{1}{2}
\draws{1}{3}
\draws{1}{4}
\draws{1}{5}
\draws{2}{5}

\draw (6.8, 2.5) node {$g_1$};
\draw (2.5, 6.5) node {$g_2$};

\draw (9.4, 2.5) node {$g_1$};
\draw (10.5, 3.4) node {$g_3$};
\draw[very thick] (10,2.5) |- (13,2.5);
\draw[very thick] (10.5,2.5) |- (10.5,3);
\draw (13.8, 2.5) node {$g_1$};

\drawg{10}{2}
\drawg{11}{2}
\drawg{12}{2}

\draw (2.5, 8.4) node {$g_2$};
\draw (3.8, 9.5) node {$g_4$};
\draw[very thick] (2.5,9) |- (2.5,12);
\draw[very thick] (2.5,9.5) |- (3,9.5);
\draw (2.5, 12.6) node {$g_2$};

\drawb{2}{9}
\drawb{2}{10}
\drawb{2}{11}

\draw[very thick] (7,9.5) -| (7.5,7.5)  -| (6.5,5.5) -| (8.5,6.5)  -| (10.5,6);
\draw (6.4, 9.5) node {$g_4$};
\draw (10.5, 5.4) node {$g_3$};

\drawr{7}{9}
\drawr{7}{8}
\drawr{7}{7}
\drawr{6}{7}
\drawr{6}{6}
\drawr{6}{5}
\drawr{7}{5}
\drawr{8}{5}
\drawr{8}{6}
\drawr{9}{6}
\drawr{10}{6}

%\path [dotted, draw, thin] (0,0) grid[step=0.29cm] (5,5);
\end{tikzpicture}

a) the different tile types.
\end{minipage} \hfill
\begin{minipage}{0.48\linewidth}
\begin{tikzpicture}[x=0.33cm,y=0.33cm]

%colle
 \draw[very thick] (2.5,13) |- (1.5,5.5) |- (5.5,1.5) |- (16,2.5);
\draw[very thick] (9.5,2.5) |- (7.5,3.5)  -- (7.5,3);
\draw[very thick] (12.5,2.5) |- (10.5,3.5)  -- (10.5,3);
\draw[very thick] (15.5,2.5) |- (13.5,3.5)  -- (13.5,3);
\draw[very thick] (2.5,9.5) -| (3.5,7.5)  -- (3,7.5);
\draw[very thick] (2.5,12.5) -| (3.5,10.5)  -- (3,10.5);

\draws{1}{1}
\draws{2}{1}
\draws{3}{1}
\draws{4}{1}
\draws{5}{1}
\draws{5}{2}
\draws{1}{2}
\draws{1}{3}
\draws{1}{4}
\draws{1}{5}
\draws{2}{5}

\draw[very thick] (6.5,2.5) |- (6.5,3);
\draw[very thick] (9.5,2.5) |- (9.5,3);
\draw[very thick] (12.5,2.5) |- (12.5,3);
\draw[very thick] (15.5,2.5) |- (15.5,3);

\drawg{6}{2}
\drawg{7}{2}
\drawg{8}{2}
\drawg{9}{2}
\drawg{10}{2}
\drawg{11}{2}
\drawg{12}{2}
\drawg{13}{2}
\drawg{14}{2}
\drawg{15}{2}

\draw[very thick] (2.5,9.5) |- (3,12.5);
\draw[very thick] (2.5,9.5) |- (3,9.5);
\draw[very thick] (2.5,6.5) |- (3,6.5);

\drawb{2}{6}
\drawb{2}{7}
\drawb{2}{8}
\drawb{2}{9}
\drawb{2}{10}
\drawb{2}{11}
\drawb{2}{12}

\draw[very thick] (3,6.5) -| (3.5,4.5)  -| (2.5,2.5) -| (4.5,3.5)  -| (6.5,3);

\drawr{3}{6}
\drawr{3}{5}
\drawr{3}{4}
\drawr{2}{4}
\drawr{2}{3}
\drawr{2}{2}
\drawr{3}{2}
\drawr{4}{2}
\drawr{4}{3}
\drawr{5}{3}
\drawr{6}{3}

\drawr{7}{3}
\drawr{8}{3}
\drawr{9}{3}
\drawr{10}{3}
\drawr{11}{3}
\drawr{12}{3}
\drawr{13}{3}
\drawr{14}{3}
\drawr{15}{3}

\drawr{3}{7}
\drawr{3}{8}
\drawr{3}{9}
\drawr{3}{10}
\drawr{3}{11}
\drawr{3}{12}

\path [dotted, draw, thin] (0,0) grid[step=0.33cm] (16,13);
\end{tikzpicture}

b) the terminal assembly.
\end{minipage}

\vspace{+0.5em}
\begin{minipage}{0.48\linewidth}
\begin{tikzpicture}[x=0.33cm,y=0.33cm]

%colle
 \draw[very thick] (2.5,13) -- (2.5,5.5);
 \draw[very thick] (2,5.5) -- (1.5,5.5) |- (5.5,1.5) |- (16,2.5);
\draw[very thick] (9.5,2.5) |- (7.5,3.5)  -- (7.5,3);
\draw[very thick] (12.5,2.5) |- (10.5,3.5)  -- (10.5,3);
\draw[very thick] (15.5,2.5) |- (13.5,3.5)  -- (13.5,3);
\draw[very thick] (2.5,9.5) -| (3.5,7.5)  -- (3,7.5);
\draw[very thick] (2.5,12.5) -| (3.5,10.5)  -- (3,10.5);

\draws{1}{1}
\draws{2}{1}
\draws{3}{1}
\draws{4}{1}
\draws{5}{1}
\draws{5}{2}
\draws{1}{2}
\draws{1}{3}
\draws{1}{4}
\draws{1}{5}
\drawb{2}{5}

\draw[very thick] (6.5,2.5) |- (6.5,3);
\draw[very thick] (9.5,2.5) |- (9.5,3);
\draw[very thick] (12.5,2.5) |- (12.5,3);
\draw[very thick] (15.5,2.5) |- (15.5,3);

\drawg{6}{2}
\drawg{7}{2}
\drawg{8}{2}
\drawg{9}{2}
\drawg{10}{2}
\drawg{11}{2}
\drawg{12}{2}
\drawg{13}{2}
\drawg{14}{2}
\drawg{15}{2}

\draw[very thick] (2.5,9.5) |- (3,12.5);
\draw[very thick] (2.5,9.5) |- (3,9.5);
\draw[very thick] (2.5,6.5) |- (3,6.5);

\drawb{2}{6}
\drawb{2}{7}
\drawb{2}{8}
\drawb{2}{9}
\drawb{2}{10}
\drawb{2}{11}
\drawb{2}{12}

\draw[very thick] (3,6.5) -| (3.5,4.5)  -| (2.5,2.5) -| (4.5,3.5)  -| (6.5,3);

\drawr{3}{6}
\drawr{3}{5}
\drawr{3}{4}
\drawr{2}{4}
\drawr{2}{3}
\drawr{2}{2}
\drawr{3}{2}
\drawr{4}{2}
\drawr{4}{3}
\drawr{5}{3}
\drawr{6}{3}

\drawr{7}{3}
\drawr{8}{3}
\drawr{9}{3}
\drawr{10}{3}
\drawr{11}{3}
\drawr{12}{3}
\drawr{13}{3}
\drawr{14}{3}
\drawr{15}{3}

\drawr{3}{7}
\drawr{3}{8}
\drawr{3}{9}
\drawr{3}{10}
\drawr{3}{11}
\drawr{3}{12}

\draw[very thick,red,->] (0.5,0.5)--(1.5,1.5);
\draw[very thick,->] (1.5,6.5)--(2.5,5.5);

\path [dotted, draw, thin] (0,0) grid[step=0.33cm] (16,13);
\end{tikzpicture}

c) if the seed is reduced to a single tile, the tile assembly system is not directed.
\end{minipage} \hfill 
\begin{minipage}{0.48\linewidth}
\begin{tikzpicture}[x=0.33cm,y=0.33cm]

%colle
 \draw[very thick] (2.5,12) |- (1.5,5.5) |- (5.5,1.5) |- (12,2.5);
\draw[very thick] (9.5,2.5) |- (7.5,3.5)  -- (7.5,3);
\draw[very thick] (2.5,9.5) -| (3.5,7.5)  -- (3,7.5);

\drawgray{1}{1}
\drawgray{2}{1}
\drawgray{3}{1}
\drawgray{4}{1}
\drawgray{5}{1}
\drawgray{5}{2}
\drawgray{1}{2}
\drawgray{1}{3}
\drawgray{1}{4}
\drawgray{1}{5}
\drawgray{2}{5}

\draw[very thick] (6.5,2.5) |- (6.5,3);

\drawgray{6}{2}
\drawgray{7}{2}
\drawgray{8}{2}
\drawp{9}{2}
\drawp{10}{2}
\drawp{11}{2}

\draw[very thick] (2.5,6.5) |- (3,6.5);

\drawgray{2}{6}
\drawgray{2}{7}
\drawgray{2}{8}
\drawo{2}{9}
\drawo{2}{10}
\drawo{2}{11}

\draw[very thick] (3,6.5) -| (3.5,4.5)  -| (2.5,2.5) -| (4.5,3.5)  -| (6.5,3);

\drawgray{3}{6}
\drawgray{3}{5}
\drawgray{3}{4}
\drawgray{2}{4}
\drawgray{2}{3}
\drawgray{2}{2}
\drawgray{3}{2}
\drawgray{4}{2}
\drawgray{4}{3}
\drawgray{5}{3}
\drawgray{6}{3}

\drawp{7}{3}
\drawp{8}{3}
\drawp{9}{3}

\drawo{3}{7}
\drawo{3}{8}
\drawo{3}{9}

\draw[thick,->] (4.5,7.5)--(4.5,10.5);
\draw (5.5, 9) node {$\vect{w}$};

\draw[thick,->] (7.5,4.5)--(10.5,4.5);
\draw (9, 5.5) node {$\vect{v}$};

\path [dotted, draw, thin] (0,0) grid[step=0.33cm] (16,13);
\end{tikzpicture}

d) the description of the terminal assembly.
\end{minipage}

\caption{The seed cannot be reduced to a single tile.}
\label{fig;appA}
\end{figure}

\subsection{Jordan curve theorem for infinite polygonal curves and left/right turns}\label{appB:curves}

Consider a vector $\vect{w} = (u,v)$, the height of $(x,y) \in \R^2$ is defined as $\height((x,y))=-vx+uy$. The line $\ell:(-\infty,+\infty)\rightarrow \mathbb{R}^2$ of vector $\vect{w}$ passing through a point $(a,b) \in \mathbb{R}^2$ is defined as $\{(x,y) \mid \height((x,y))=\height((a,b))\}$. This line cuts the 2D plane $\mathbb{R}^2$ into two connected components: the right-hand side $R=\{(x,y) \mid  \height((x,y)) \leq \height((a,b))\}$ and the left-hand side $L=\{(x,y) \mid \height((x,y)) \geq \height((a,b))\}$. We say that a curve $c$ turns right (\resp left) of $\ell$ if there exist $\epsilon>0$ and $t \in \mathbb{R}$ such that $c(t)$ is a point of $\ell$ and for all $t< z \leq t+\epsilon$, $c(z)$ belongs to $R$ (\resp $L$) and is not on $\ell$. Moreover, we say that {\em $c$ crosses $\ell$ from left to right} (\resp {\em from right to left}) if there exist $t$ and $\epsilon>0$ such that $c$ turns right (\resp left) of $\ell$ at $c(t)$, for all $t-\epsilon \leq z\leq t$, $c(t)\in L$ (\resp $c(t)\in R$) and $c(t-\epsilon)$ is not on~$\ell$.\footnote{\label{fn:cross}The definition of ``$c$ crosses $\ell$ from left to right'' ensures that certain kinds of intersections between $c$ and $\ell$ are not counted as crossing, those include: coincident intersection between two straight segments of length $>0$, and ``glancing'' of a line ``tangentially'' to a corner.} 
We say that {\em  $c$ crosses $\ell$} if $c$ crosses $\ell$ either from left to right or from right to left.  
Now, we generalise these notions to a specific class of polygonal curves:

\begin{definition}\label{def:simple bi-infinite periodic polygonal curve}
A {\em simple bi-infinite periodic polygonal curve} is a simple curve that is a concatenation of all the translations of a finite curve $c:[a,b]\rightarrow R^2$ which is made of horizontal and vertical segments of length 1, \emph{i.e}: $$\mathrm{concat}_{i \in Z}\!\left(c+i\vect{c(a)c(b)}\right).$$
\end{definition}

Let $w_0,w_1\in\mathbb{R}^2$ and let $c$ be a curve. We say that {\em $w_0$ is connected to $w_1$ while avoiding} $c$ if there is a curve $d: [0,1] \rightarrow \mathbb{R}^2$ with $d(0)=w_0$ and $d(1)=w_1$ and $d$ does not intersect $c$.

Consider a finite curve $c:[a,b]\rightarrow R^2$ which is made of horizontal and vertical segments of lengths $1$, let $\vect{w}=\vect{c(a)c(b)}=(u,v)$ and such that $c$ and $c+\vect{w}$ intersect only at $c(b)=c(a)+\vect{w}$. Then the curve $d=\mathrm{concat}_{i \in Z}\!\left(c+i\vect{w}\right)$ is a simple bi-infinite periodic polygonal curve (see Figure \ref{fig:appB}). We call $\vect{w}$ the \emph{direction} of $c$ and $d$. The curve $c$ considered here  are such that $\height^+=\max\{\height((x,y)) \mid (x,y)\in c([a,b])\}$ and $\height^-=\min\{\height((x,y)) \mid (x,y)\in c([a,b])\}$ are bounded. Remark that for any $(x,y)\in \R^2$, the height of $(x,y)$ is the height of $(x,y)+\vect{w}$, \emph{i.e.} $-v(x+u)+u(y+v)=-vx+uy=\height((x,y))$ and thus for all $(x,y) \in d(\R)$, $\height^-\leq \height((x,y)) \leq \height^+$. Intuitively, the curve $c$ is between two lines of direction $\vect{w}$ (the lines $\ell^+$ and $\ell^-$ of Figure \ref{fig:appB}). We say that $(x, y) \in \R^2$ is below (resp. over)  $c$ if $\height((x,y)) \leq \height^-$ (resp.  $\height((x,y)) \geq \height^+$). Consider a vector $\vect{w'}$ which is not collinear to $\vect{w}$ then for any $(x,y)\in \R^2$, the height of $(x,y)+\vect{w'}$ is not the same as the height of $(x,y)$. 
In particular, if $\vect{w'}=(-v,u)$ (resp. $\vect{w'}=(v,-u)$), \emph{i.e.} $\vect{v}$ rotated by $\pi/2$ (resp. $-\pi/2$), then $\height((x,y))< \height((x,y)+\vect{w'})$ (resp. $\height(x,y)>\height((x,y)+\vect{w'})$).

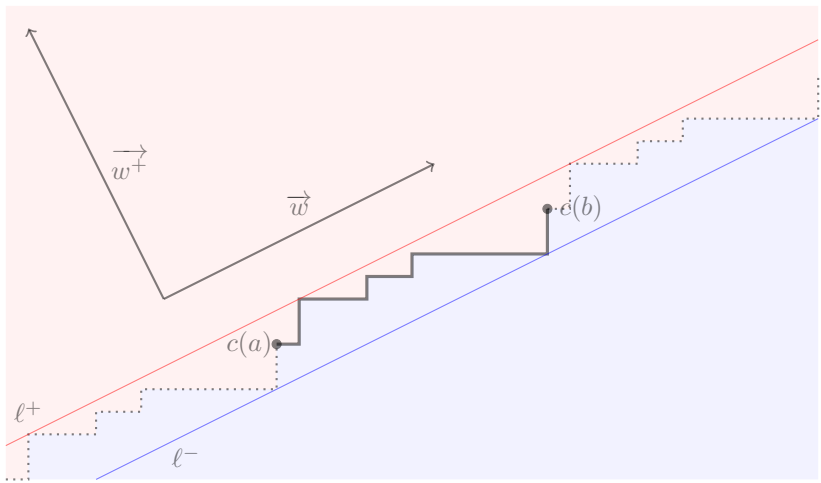
\begin{figure}
\centering
\begin{tikzpicture}[x=0.3cm,y=0.3cm]

%2 zone
\transparent{0.5}
\fill[red!10!white] (0,0) |- (36,21) |- (30,16) |- (28,15) |- (25,14) |- (24,12) |- (18,10) |-(16,9) |- (13,8) |- (12,6) |- (6,4) |- (4,3) |- (1,2) |- (0,0);
\fill[blue!10!white] (0,0) -| (36,18) |- (30,16) |- (28,15) |- (25,14) |- (24,12) |- (18,10) |-(16,9) |- (13,8) |- (12,6) |- (6,4) |- (4,3) |- (1,2) |- (0,0);

%courbe
\draw[very thick] (12,6) -| (13,8)  -| (16,9) -| (18,10)  -| (24,12);
\draw[thick,dotted] (0,0) -| (1,2)  -| (4,3) -| (6,4)  -| (12,6);
\draw[thick,dotted] (24,12) -| (25,14)  -| (28,15) -| (30,16)  -| (36,18);

%l+ l-
\draw[red] (0,1.5) -- (36,19.5);
\draw[blue] (4,0) -- (36,16);
\draw (1,3) node {$\ell^+$};
\draw (8, 1) node {$\ell^-$};

%nom des lignes

\draw[thick,->] (7,8)--(19,14);
\draw[thick,->] (7,8)--(1,20);
\draw (13, 12.2) node {$\vect{w}$};
\draw (5.5, 14) node {$\vect{w^+}$};

%c(a) c(b)
\draw[fill=black] (12,6) circle (0.2);
\draw[fill=black] (24,12) circle (0.2);
\draw (10.8,6) node {$c(a)$};
\draw (25.5,12) node {$c(b)$};

%\draw (12, 6) node {\Huge $\cdot$};
%\draw (24, 12) node {$\cdot$};

\end{tikzpicture}

\caption{A curve $c:[a,b]\rightarrow \R^2$ of direction $\vect{w}$ which generates a simple bi-infinite periodic polygonal curve and cuts the plane in two parts. All points which are over the curve are in the left hand side (in red) and all points which are under the curve are in the right hand side of the curve (in blue).}
\label{fig:appB}
\end{figure}

\begin{theorem}
  \label{thm:infinite-jordan}
  Let $c : \mathbb{R} \rightarrow \mathbb{R}^2$ be a simple bi-infinite periodic polygonal curve, 
  and let $(x_\textrm{min}, y_\textrm{min})$ and $(x_\textrm{max}, y_\textrm{max})$ be respectively below and over $c$. 
    Then $c$ cuts $\R^2$ into two connected components:
  \begin{enumerate}
  \item\label{thm:infinite-jordan:conc:1} the left-hand side of $c$:  
  $$c(\mathbb{R}) \cup  \{ w \mid w \in\mathbb{R}^2 \textrm{ is connected to } (x_\textrm{max},y_\textrm{max}) \textrm{ while avoiding } c \},$$ 
  \item\label{thm:infinite-jordan:conc:2} the right-hand side of $c$:  
  $$c(\mathbb{R}) \cup  \{ w \mid w \in\mathbb{R}^2 \textrm{ is connected to } (x_\textrm{min},y_\textrm{min}) \textrm{ while avoiding } c \}$$
  \end{enumerate}
\end{theorem}
\begin{proof}
  The proof is similar to the proof of the Jordan Curve Theorem for polygonal curves. 

Let $\vect{v}$ be the direction of $c$, let $\vect{v^+}$ (resp. $\vect{v^-}$) be the vector $\vect{v}$ rotated by $\pi/2$ (resp. $-\pi/2$). For each point $(x, y)\in\R^2$ we define $\ell(x,y)$ to be the line of vector $\vect{v^+}$ through $(x, y)$, $\ell^+(x,y)$ to be the ray of vector $\vect{v^+}$ from $(x, y)$, and $\ell^-(x,y)$ to be the ray of vector $\vect{v^-}$ from $(x,y)$.

  Since $\vect{v}^+$ and $\vect{v}$ are not collinear,
  for all $(x, y)\in\R^2$,
   $\ell^-(x,y)$ and $\ell^+(x, y)$ intersects $c$ at a finite number of points.
  Moreover, because $\ell(x,y)$ cuts the plane into two connected components, $\ell(x,y)$ crosses $c$
  an odd number of times,

  Let $L$ be the subset of $\R^2$ such that for all $(x, y)\in L$, $c$ crosses $\ell^-(x, y)$ an odd number of times,\footnote{The term ``crosses'' was defined with respect to a line $\ell(x,y)$. This definition is easily generalised to the (coincident) ray $\ell^+(x,y)$, by considering only those crossings at locations that are simultaneously on both $\ell$ and the ray. Likewise for  the ray $\ell^-(x,y)$.} and let $R$ be the subset of $\R^2$ such that for all $(x, y)\in R$, $c$ crosses $\ell^+(x, y)$ an odd number of times. 
  Note that $L\cap R = c(\R)$  (where $c(\R)$ is the range of~$c$), or in other words the intersection of $L$ and $R$ is the set of all points of $c$.

  We claim that $L$ (\resp $R$) is a connected component: indeed, since $c$ is connected, if $(x_0,y_0)$ and $(x_1,y_1)$ are both in $L$ (\resp both in $R$), let $t_0,t_1\in\R$ be the smallest real numbers such that $\ell^-(x_0,y_0)(t_0)$ (\resp $\ell^+(x_0,y_0)(t_0)$) is on $c$, and $\ell^-(x_1,y_1)(t_1)$ (\resp $\ell^+(x_1,y_1)(t_1)$) is on $c$.
  We know there is at least one such intersection because $\ell^-(x_0,y_0)$ and $\ell^-(x_1,y_1)$ (\resp $\ell^+(x_0,y_0)$ and  $\ell^+(x_1,y_1)$) cross $c$ an odd number of times, hence at least once. Without loss of generality we suppose that $\ell^-(x_0,y_0)(t_0)$ (\resp $\ell^+(x_0,y_0)(t_0)$) is before $\ell^-(x_1,y_1)(t_1)$  (\resp $\ell^+(x_1,y_1)(t_1)$) according to the order of positions along $c$.

  Let $d$ be the curve defined as the concatenation of:
  \begin{itemize}
  \item $\ell^-(x_0,y_0)$ (\resp, $\ell^+(x_0,y_0)$) from $(x_0,y_0)$ up to $\ell^-(x_0,y_0)(t_0)$ (\resp, $\ell^+(x_0,y_0)(t_0)$),
  \item $c$ from $\ell^-(x_0,y_0)(t_0)$ to  $\ell^-(x_1,y_1)(t_1)$ (\resp, from $\ell^+(x_0,y_0)(t_0)$ to  $\ell^+(x_1,y_1)(t_1)$),
  \item $\reverse{\ell^-(x_1,y_1)}$  (\resp, $\reverse{\ell^+(x_1,y_1)}$) from $\ell^-(x_1,y_1)(t_1)$ (\resp, $\ell^+(x_1,y_1)(t_1)$) to $(x_1,y_1)$.
  \end{itemize}
  The curve $d$ is entirely in $L$ (\resp in $R$), starts at $(x_0,y_0)$ and ends at $(x_1,y_1)$, which proves that $L$ (\resp $R$) is indeed a connected component.

  We claim that $L\cup R=\R^2$: indeed, for all $(x,y)\in\R^2$, since $\ell(x,y)$ crosses $c$ an odd number of times, $\ell^+(x,y)$ crosses $c$ an even number of times if and only if $\ell^-(x,y)$ crosses $c$ an odd number of times (and vice-versa),
  hence $(x,y)$ is in at least one of $L$ and $R$, and only the points of $c$ are in both.

  As a conclusion, $c$ cuts the plane into two disjoint connected components:  
   $L\setminus c(\R)$ which we call the \emph{strict left-hand side} of $c$ and
   $R\setminus c(\R) $ which we call the \emph{strict right-hand side} of $c$.
  
  Let $(x,y)$ be a point over $c$ then $\ell^+(x,y)$ does not intersect $c$, and thus $\ell^-(x,y)$ intersects $c$ an odd number of times. 
  $(x,y)$ is in $L$, and since $L$ is connected, we get Conclusion~\ref{thm:infinite-jordan:conc:1}. 
  Let $(x,y)$ be a point under $c$ then $\ell^-(x,y)$ does not intersect $c$, and thus $\ell^+(x,y)$ intersects $c$ an odd number of times. 
  $(x,y)$ is in $R$, and since $R$ is connected, we get Conclusion~\ref{thm:infinite-jordan:conc:2}. 
\end{proof}

Using the technique in the proof of Theorem~\ref{thm:infinite-jordan}, the definitions of turning right and left can be extended from a line to an infinite polygonal curve. 

\begin{definition}[left-hand and right-hand side of a curve]\label{def:curve LHS RHS}
The conclusion of Theorem~\ref{thm:infinite-jordan} defines the {\em left-hand side} $\mathcal{L}\subsetneq \mathbb{R}^2$, and 
{\em right-hand side} $\mathcal{R}\subsetneq \mathbb{R}^2$ of a simple bi-infinite periodic polygonal curve $c$. 
The {\em strict left-hand side} of a simple bi-infinite periodic polygonal curve $c:\mathbb{R}\rightarrow \mathbb{R}^2$ is the set $\mathcal{L} \setminus c(\mathbb{R})$. 
Likewise the {\em strict right-hand side} of a simple bi-infinite periodic polygonal curve $c$ is the set $\mathcal{R} \setminus c(\mathbb{R})$, where $\mathcal{R}$ is the right-hand side of $c$.
\end{definition}

%The right side of $c$ will $\rightside(c)=\mathcal{R}\cap \Z^2$
We have already defined what this means for a curve to cross a line. 
Theorem~\ref{thm:infinite-jordan} enables us to generalise that definition to one curve turning from another
simple bi-infinite periodic polygonal curve.   
\begin{definition}[One curve turning left or right from another] \label{def:curve turn}
Let $d$ be a curve and let $c : \mathbb{R}\rightarrow \mathbb{R}^2$ be a simple bi-infinite periodic polygonal curve.
We say that $d$ turns left (\resp, right) from curve $c$ at the point $d(z) = c(w)$, for some $z,w \in \mathbb R$, if  there is an $\epsilon > 0$ such that $d(z+\epsilon)$ is in the strict left-hand (\resp, right-hand) side of $d$ and 
  $d(z')$ is not on $c$, for all $z'$ where $z < z' < z+\epsilon$.
\end{definition}

Definition~\ref{def:curve turn} is consistent with the definition of one path turning left/right from another (Section~\ref{sec:defs-paths}) in the following sense.  
Consider two paths $P$ and $Q$ such that $Q$ turns right (\resp, left) of $P$ and consider a curve $c$ which contains $\embed{P}$ (with the same orientation) then $\embed{Q}$ turns right (\resp, left) of $c$. 
However, not all curves are the embedding of some path, hence the reverse implication does not hold. 
Also, the curve turning definition has no requirement analogous to the orientation requirement for path turns (implied by $i>0$ and the definition of $\tau$ in the path turning definition).

Throughout the article, we will use the following fact several times. Consider a simple bi-infinite periodic polygonal curve $c$ of direction $\vect{v}$ and a vector $\vect{w}$ which is not collinear to $\vect{v}$ and such that $c$ and $c+\vect{w}$ do not intersect. Without loss of generality, we suppose that $c(\R)$ is included into $\mathcal{L}(c+\vect{w})$. In this case, $\mathcal{R}(c) \cap \mathcal{L}(c+\vect{w}) \subset \R^2$ is connected and we claim that for any $(x,y)\in \R^2$, there exists $i\in\Z$ such that $(x,y)+i\vect{w}$ is in $\mathcal{R}(c) \cap \mathcal{L}(c+\vect{w}) \subset \R^2$. Indeed, for all $t\in \Z$, let $$\height^+_t=\max\{\height((x,y)+t\vect{w}) \mid (x,y)\in c([a,b])\} \text{ and }$$  $$\height^-_t=\min\{\height((x,y)+t\vect{w}) \mid (x,y)\in c([a,b])\}. $$ Since, $c(\R)$ is included into $\mathcal{L}(c+\vect{w})$ then for all $t\in Z$, $h^+_t>h^+_{t+1}$ and $h^-_t>h^-_{t+1}$. Thus, for any $(x,y)\in \R^2$ there exist $j<k$ such that $(x,y)$ is under $c+j\vect{w}$ (resp. over $c+k\vect{w}$) and thus in $\mathcal{R}(c+t\vect{w})$ (in $\mathcal{L}(c+k\vect{w})$). Then, $i=\max\{t: (x,y)\in \mathcal{R}(c+t\vect{w})\}$ is correctly defined and $(x,y)$ is in $\mathcal{R}(c+i\vect{w}) \cap \mathcal{L}(c+(i+1)\vect{w})$. Thus, $(x,y)-i\vect{w}$ is in $\mathcal{R}(c) \cap \mathcal{L}(c+\vect{w})$.

Also consider two simple bi-infinite periodic polygonal curves $c$ and $c'$ both of direction $\vect{v}$ with $c(0)=c'(0)$ and such that $c'(\R)$ is a subset of $\mathcal{R}(c)$. Then $\mathcal{R}(c) \cap \mathcal{L}(c')$ is connected and there exists $A\subset \Z^2$ such that $A$ is finite and $\Z \cap (\mathcal{R}(c) \cap \mathcal{L}(c'))=\bigcup_{i\in\Z}(A+i\vect{v})$. Indeed, consider $\ell$ (resp. $\ell'$) which is the restriction of $c$ (resp. $c'$) which starts in $c(0)$ (resp. $c'(0)$) and ends in $c(0)+\vect{v}$ then $\concat{\ell,\ell'}$ is a cycle but it may not be simple. Let $\vect{v^+}$ be the rotation of $\vect{v}$ by $\pi/2$. Consider the translation of $\ell$ by $\epsilon\vect{v^+}$ with $\epsilon>0$, then $$\concat{\left [c(0),c(0)+\epsilon\vect{v^+} \right],\ell+\epsilon\vect{v^+},\left[c(0)+\vect{v}+\epsilon\vect{v^+},c(0)+\vect{v}\right] , \ell'}$$ is now a simple cycle and let $I$ be its interior and let $A=I\cap Z$. Moreover, $\epsilon$ could chosen small enough such that for any $(x,y) \in \Z^2$ such that $(x,y) \in A$, we have  $(x,y) \in \mathcal{R}(c) \cap \mathcal{L}(c')$. Then $\Z\cap(\mathcal{R}(c) \cap \mathcal{L}(c'))=\bigcup_{i\in\Z}(A+i\vect{v})$. We call $A$, the interior of $c$ and $c'$. 

%First, let $c$ be a simple infinite almost-vertical polygonal curve (Definition~\ref{def:simple infinite almost-vertical polygonal curve}), 
%and let $\ell^1$ and $\ell^2$ be two vertical rays from the south such that $\ell^1$ is the vertical ray from the south used to define $c$, and $\ell^2$ is strictly to the east (\resp, west) of $\ell^1$ and does not intersect $c$, then $\ell^2$ belongs to the right-hand side (\resp, left-hand side) of $c$.
%The same conclusion holds if $\ell^1$ and $\ell^2$ go to the north, with $\ell^1$ being the vertical ray from the north used to define $c$. 

\section{Analysis of an aperiodic tile assembly system.}
\label{app:C}

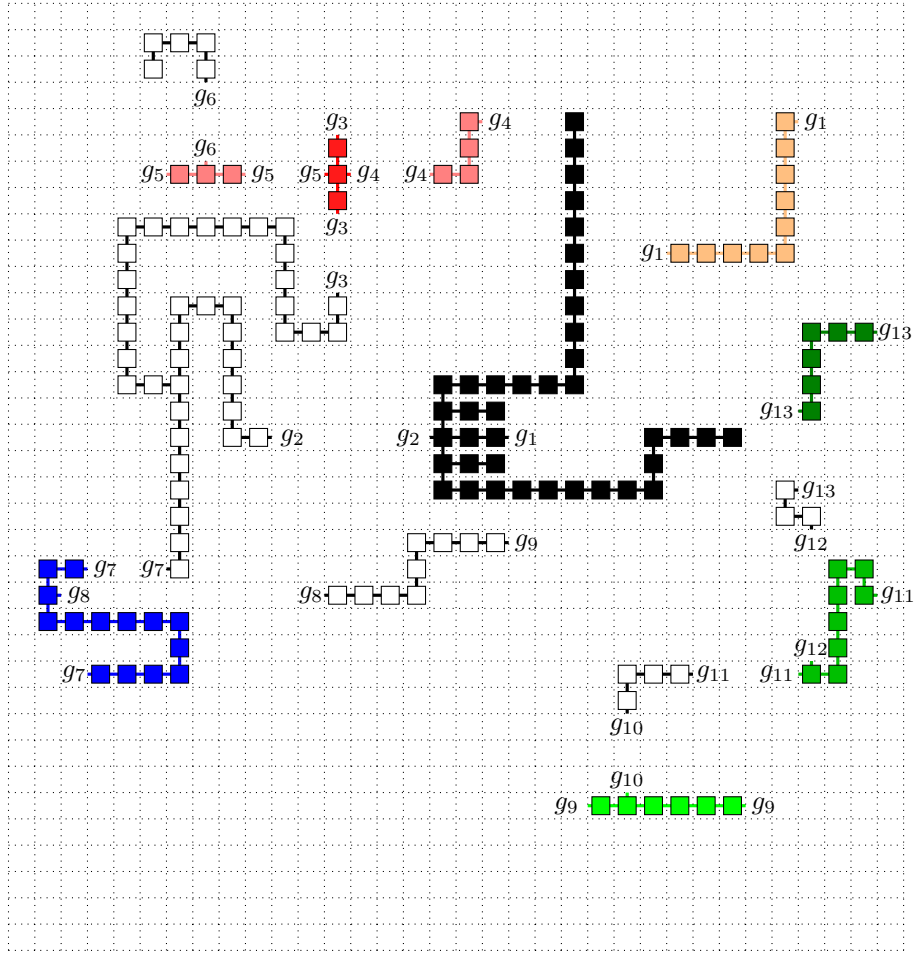
\begin{figure}
\centering
\begin{tikzpicture}[x=0.35cm,y=0.35cm]

%graine
\draw[very thick] (16.5,21.5) -| (21.5,31.5);
\draw[very thick] (16.5,21.5) |- (24.5,17.5) |- (27.5,19.5);
\draw[very thick] (16.5,20.5) |- (18.5,20.5);
\draw[very thick] (16,19.5) |- (19,19.5);
\draw[very thick] (16.5,18.5) |- (18.5,18.5);
\draws{16}{21}
\draws{16}{20}
\draws{16}{19}
\draws{16}{18}
\draws{16}{17}
\draws{17}{21}
\draws{17}{20}
\draws{17}{19}
\draws{17}{18}
\draws{17}{17}
\draws{18}{21}
\draws{18}{20}
\draws{18}{19}
\draws{18}{18}
\draws{18}{17}
\draws{19}{21}
\draws{20}{21}
\draws{21}{21}
\draws{21}{22}
\draws{21}{23}
\draws{21}{24}
\draws{21}{25}
\draws{21}{26}
\draws{21}{27}
\draws{21}{28}
\draws{21}{29}
\draws{21}{30}
\draws{21}{31}
\draws{19}{17}
\draws{20}{17}
\draws{21}{17}
\draws{22}{17}
\draws{23}{17}
\draws{24}{17}
\draws{24}{18}
\draws{24}{19}
\draws{25}{19}
\draws{26}{19}
\draws{27}{19}

\draw (19.7,19.5) node {$g_1$};
\draw (15.2,19.5) node {$g_2$};

%chemin de liaison 1
\draw[very thick] (10,19.5) -| (8.5,24.5) -| (6.5,14.5) -| (6,14.5);
\draw[very thick] (6.5,21.5) -| (4.5,27.5) -| (10.5,23.5) -| (12.5,25) ;
\drawt{9}{19}
\drawt{8}{19}
\drawt{8}{20}
\drawt{8}{21}
\drawt{8}{22}
\drawt{8}{23}
\drawt{5}{21}
\drawt{4}{21}
\drawt{4}{22}
\drawt{4}{23}
\drawt{4}{24}
\drawt{4}{25}
\drawt{4}{26}
\drawt{4}{27}
\drawt{5}{27}
\drawt{6}{27}
\drawt{7}{27}
\drawt{8}{27}
\drawt{9}{27}
\drawt{10}{27}
\drawt{10}{26}
\drawt{10}{25}
\drawt{10}{24}
\drawt{10}{23}
\drawt{11}{23}
\drawt{12}{23}
\drawt{12}{24}
\drawt{8}{24}
\drawt{7}{24}
\drawt{6}{24}
\drawt{6}{23}
\drawt{6}{22}
\drawt{6}{21}
\drawt{6}{20}
\drawt{6}{19}
\drawt{6}{18}
\drawt{6}{17}
\drawt{6}{16}
\drawt{6}{15}
\drawt{6}{14}

\draw (10.8,19.5) node {$g_2$};
\draw (12.5,25.5) node {$g_3$};
\draw (5.5,14.5) node {$g_7$};

%chemin de liaison 2
\draw[very thick] (12,13.5) -| (15.5,15.5)  -| (19,15.5);
\drawt{12}{13}
\drawt{13}{13}
\drawt{14}{13}
\drawt{15}{13}
\drawt{15}{14}
\drawt{15}{15}
\drawt{16}{15}
\drawt{17}{15}
\drawt{18}{15}

%chemin fini sur peigne
\draw[very thick] (7.5,33) |- (5.5,34.5)  -| (5.5,33.5);
\drawt{7}{33}
\drawt{7}{34}
\drawt{6}{34}
\drawt{5}{34}
\drawt{5}{33}

\draw (7.5,32.4) node {$g_6$};

%pompage rouge principale
\draw[red,very thick] (12.5,28) |- (12.5,31);
\draw[red,very thick] (12,29.5) |- (13,29.5);
\drawred{12}{28}
\drawred{12}{29}
\drawred{12}{30}
\draw (12.5,27.5) node {$g_3$};
\draw (12.5,31.5) node {$g_3$};
\draw (13.7,29.5) node {$g_4$};
\draw (11.5,29.5) node {$g_5$};

%peigne rouge 1
\draw[red!50!white,very thick] (6,29.5) |- (9,29.5);
\draw[red!50!white,very thick] (7.5,29.5) |- (7.5,30);
\drawr{8}{29}
\drawr{7}{29}
\drawr{6}{29}

\draw (9.7,29.5) node {$g_5$};
\draw (5.5,29.5) node {$g_5$};
\draw (7.5,30.5) node {$g_6$};

%peigne rouge 2
\draw[red!50!white,very thick] (16,29.5) -| (17.5,31.5) -| (18,31.5);
\drawr{16}{29}
\drawr{17}{29}
\drawr{17}{30}
\drawr{17}{31}

\draw (15.5,29.5) node {$g_4$};
\draw (18.7,31.5) node {$g_4$};

%pompage orange
\draw[orange!50!white,very thick] (25,26.5) -| (29.5,31.5) -| (30,31.5);
\drawo{25}{26}
\drawo{26}{26}
\drawo{27}{26}
\drawo{28}{26}
\drawo{29}{26}
\drawo{29}{27}
\drawo{29}{28}
\drawo{29}{29}
\drawo{29}{30}
\drawo{29}{31}
\draw (24.5,26.5) node {$g_1$};
\draw (30.7,31.5) node {$g_1$};

%pompage bleu
\draw[blue,very thick] (3,14.5) |- (1.5,14.5) |- (6.5,12.5) |- (3,10.5);
\draw[blue,very thick] (1.5,13.5) |- (2,13.5);
\drawblue{2}{14}
\drawblue{1}{14}
\drawblue{1}{13}
\drawblue{1}{12}
\drawblue{2}{12}
\drawblue{3}{12}
\drawblue{4}{12}
\drawblue{5}{12}
\drawblue{6}{12}
\drawblue{6}{11}
\drawblue{6}{10}
\drawblue{5}{10}
\drawblue{4}{10}
\drawblue{3}{10}

\draw (3.7,14.5) node {$g_7$};
\draw (2.5,10.5) node {$g_7$};
\draw (2.7,13.5) node {$g_8$};

%pompage vert 
\draw[green,very thick] (22,5.5) |- (28,5.5);
\draw[green,very thick] (23.5,5.5) |- (23.5,6);
\drawgreen{27}{5}
\drawgreen{26}{5}
\drawgreen{25}{5}
\drawgreen{24}{5}
\drawgreen{23}{5}
\drawgreen{22}{5}

\draw (21.2,5.5) node {$g_9$};
\draw (28.7,5.5) node {$g_9$};
\draw (23.5,6.5) node {$g_{10}$};

%chemin de liaison vert
\draw[very thick] (23.5,9) |- (26,10.5);
\drawt{23}{9}
\drawt{23}{10}
\drawt{24}{10}
\drawt{25}{10}

\draw (11.5,13.5) node {$g_8$};
\draw (19.7,15.5) node {$g_9$};

\draw (23.5,8.5) node {$g_{10}$};
\draw (26.8,10.5) node {$g_{11}$};

%peigne vert 1
\draw[green!75!black,very thick] (30,10.5) -| (31.5,14.5)  -| (32.5,13.5) -| (33,13.5);
\draw[green!75!black,very thick] (30.5,10.5) -| (30.5,11);
\drawgre{30}{10}
\drawgre{31}{10}
\drawgre{31}{11}
\drawgre{31}{12}
\drawgre{31}{13}
\drawgre{31}{14}
\drawgre{32}{14}
\drawgre{32}{13}

\draw (29.2,10.5) node {$g_{11}$};
\draw (30.5,11.5) node {$g_{12}$};
\draw (33.8,13.5) node {$g_{11}$};

%chemin de liaison vert-vert
\draw[very thick] (30.5,16) |- (29.5,16.5) |- (30,17.5);
\drawt{30}{16}
\drawt{29}{16}
\drawt{29}{17}

\draw (30.5,15.5) node {$g_{12}$};
\draw (30.8,17.5) node {$g_{13}$};

%peigne vert 2 
\draw[green!50!black,very thick] (30,20.5) -| (30.5,23.5) |- (33,23.5);
\drawg{30}{20}
\drawg{30}{21}
\drawg{30}{22}
\drawg{30}{23}
\drawg{31}{23}
\drawg{32}{23}

\draw (29.2,20.5) node {$g_{13}$};
\draw (33.7,23.5) node {$g_{13}$};

\path [dotted, draw, thin] (0,0) grid[step=0.35cm] (34,36);
\end{tikzpicture}
\caption{A tile set: the seed is in black and the colored paths are simply pumpable. Each tile drawn here is unique.}
\label{fig:appC:step1}
\end{figure}

\begin{figure}
\centering
\begin{tikzpicture}[x=0.35cm,y=0.35cm]

%graine
\draw[very thick] (7.5,21.5) -| (12.5,31.5);
\draw[very thick] (7.5,21.5) |- (15.5,17.5) |- (18.5,19.5);
\draw[very thick] (7.5,20.5) |- (9.5,20.5);
\draw[very thick] (7,19.5) |- (10,19.5);
\draw[very thick] (7.5,18.5) |- (9.5,18.5);
\draws{7}{21}
\draws{7}{20}
\draws{7}{19}
\draws{7}{18}
\draws{7}{17}
\draws{8}{21}
\draws{8}{20}
\draws{8}{19}
\draws{8}{18}
\draws{8}{17}
\draws{9}{21}
\draws{9}{20}
\draws{9}{19}
\draws{9}{18}
\draws{9}{17}
\draws{10}{21}
\draws{11}{21}
\draws{12}{21}
\draws{12}{22}
\draws{12}{23}
\draws{12}{24}
\draws{12}{25}
\draws{12}{26}
\draws{12}{27}
\draws{12}{28}
\draws{12}{29}
\draws{12}{30}
\draws{12}{31}
\draws{10}{17}
\draws{11}{17}
\draws{12}{17}
\draws{13}{17}
\draws{14}{17}
\draws{15}{17}
\draws{15}{18}
\draws{15}{19}
\draws{16}{19}
\draws{17}{19}
\draws{18}{19}

%chemin de liaison 1
\draw[very thick] (7,19.5) -| (5.5,24.5) -| (3.5,14.5) -| (3,14.5);
\draw[very thick] (3.5,21.5) -| (1.5,27.5) -| (7.5,23.5) -| (9.5,25) ;
\drawt{6}{19}
\drawt{5}{19}
\drawt{5}{20}
\drawt{5}{21}
\drawt{5}{22}
\drawt{5}{23}
\drawt{2}{21}
\drawt{1}{21}
\drawt{1}{22}
\drawt{1}{23}
\drawt{1}{24}
\drawt{1}{25}
\drawt{1}{26}
\drawt{1}{27}
\drawt{2}{27}
\drawt{3}{27}
\drawt{4}{27}
\drawt{5}{27}
\drawt{6}{27}
\drawt{7}{27}
\drawt{7}{26}
\drawt{7}{25}
\drawt{7}{24}
\drawt{7}{23}
\drawt{8}{23}
\drawt{9}{23}
\drawt{9}{24}
\drawt{5}{24}
\drawt{4}{24}
\drawt{3}{24}
\drawt{3}{23}
\drawt{3}{22}
\drawt{3}{21}
\drawt{3}{20}
\drawt{3}{19}
\drawt{3}{18}
\drawt{3}{17}
\drawt{3}{16}
\drawt{3}{15}
\drawt{3}{14}

%chemin de liaison 2
\draw[very thick] (2,13.5) -| (5.5,15.5)  -| (9,15.5);
\drawt{2}{13}
\drawt{3}{13}
\drawt{4}{13}
\drawt{5}{13}
\drawt{5}{14}
\drawt{5}{15}
\drawt{6}{15}
\drawt{7}{15}
\drawt{8}{15}

\draw[very thick] (2,9.5) -| (5.5,10);
\drawt{2}{9}
\drawt{3}{9}
\drawt{4}{9}
\drawt{5}{9}

\draw[very thick] (2,5.5) -| (5.5,6);
\drawt{2}{5}
\drawt{3}{5}
\drawt{4}{5}
\drawt{5}{5}

\draw[very thick] (2,1.5) -| (5.5,2);
\drawt{2}{1}
\drawt{3}{1}
\drawt{4}{1}
\drawt{5}{1}

%chemin fini sur peigne
\draw[very thick] (7.5,30) |- (5.5,31.5)  -| (5.5,30.5);
\drawt{7}{30}
\drawt{7}{31}
\drawt{6}{31}
\drawt{5}{31}
\drawt{5}{30}

\draw[very thick] (4.5,30) |- (2.5,31.5)  -| (2.5,30.5);
\drawt{4}{30}
\drawt{4}{31}
\drawt{3}{31}
\drawt{2}{31}
\drawt{2}{30}

\draw[very thick] (1.5,30) |- (0,31.5);
\drawt{1}{30}
\drawt{1}{31}
\drawt{0}{31}

\draw[very thick] (7.5,33) |- (5.5,34.5)  -| (5.5,33.5);
\drawt{7}{33}
\drawt{7}{34}
\drawt{6}{34}
\drawt{5}{34}
\drawt{5}{33}

\draw[very thick] (4.5,33) |- (2.5,34.5)  -| (2.5,33.5);
\drawt{4}{33}
\drawt{4}{34}
\drawt{3}{34}
\drawt{2}{34}
\drawt{2}{33}

\draw[very thick] (1.5,33) |- (0,34.5);
\drawt{1}{33}
\drawt{1}{34}
\drawt{0}{34}

%pompage rouge principale
\draw[red,very thick] (9.5,25) |- (9.5,36);
\draw[red,very thick] (9,26.5) |- (10,26.5);
\draw[red,very thick] (9,29.5) |- (10,29.5);
\draw[red,very thick] (9,32.5) |- (10,32.5);
\draw[red,very thick] (9,35.5) |- (10,35.5);
\drawred{9}{25}
\drawred{9}{26}
\drawred{9}{27}
\drawred{9}{28}
\drawred{9}{29}
\drawred{9}{30}
\drawred{9}{31}
\drawred{9}{32}
\drawred{9}{33}
\drawred{9}{34}
\drawred{9}{35}

%peigne rouge 1
\draw[red!50!white,very thick] (8,26.5) |- (9,26.5);
\drawr{8}{26}

\draw[red!50!white,very thick] (0,29.5) |- (9,29.5);
\draw[red!50!white,very thick] (7.5,29.5) |- (7.5,30);
\draw[red!50!white,very thick] (4.5,29.5) |- (4.5,30);
\draw[red!50!white,very thick] (1.5,29.5) |- (1.5,30);
\drawr{8}{29}
\drawr{7}{29}
\drawr{6}{29}
\drawr{5}{29}
\drawr{4}{29}
\drawr{3}{29}
\drawr{2}{29}
\drawr{1}{29}
\drawr{0}{29}

\draw[red!50!white,very thick] (0,32.5) |- (9,32.5);
\draw[red!50!white,very thick] (7.5,32.5) |- (7.5,33);
\draw[red!50!white,very thick] (4.5,32.5) |- (4.5,33);
\draw[red!50!white,very thick] (1.5,32.5) |- (1.5,33);
\drawr{8}{32}
\drawr{7}{32}
\drawr{6}{32}
\drawr{5}{32}
\drawr{4}{32}
\drawr{3}{32}
\drawr{2}{32}
\drawr{1}{32}
\drawr{0}{32}

\draw[red!50!white,very thick] (0,35.5) |- (9,35.5);
\draw[red!50!white,very thick] (7.5,35.5) |- (7.5,36);
\draw[red!50!white,very thick] (4.5,35.5) |- (4.5,36);
\draw[red!50!white,very thick] (1.5,35.5) |- (1.5,36);
\drawr{8}{35}
\drawr{7}{35}
\drawr{6}{35}
\drawr{5}{35}
\drawr{4}{35}
\drawr{3}{35}
\drawr{2}{35}
\drawr{1}{35}
\drawr{0}{35}

%peigne rouge 2
\draw[red!50!white,very thick] (10,35.5) -| (11.5,36);
\drawr{10}{35}
\drawr{11}{35}

\draw[red!50!white,very thick] (10,32.5) -| (11.5,34.5) -| (13.5,36);
\drawr{10}{32}
\drawr{11}{32}
\drawr{11}{33}
\drawr{11}{34}
\drawr{12}{34}
\drawr{13}{34}
\drawr{13}{35}

\draw[red!50!white,very thick] (10,29.5) -| (11.5,31.5) -| (12,31.5);
\drawr{10}{29}
\drawr{11}{29}
\drawr{11}{30}
\drawr{11}{31}

\draw[red!50!white,very thick] (10,26.5) -| (11.5,28.5) -| (12,28.5);
\drawr{10}{26}
\drawr{11}{26}
\drawr{11}{27}
\drawr{11}{28}

%pompage orange
\draw[orange!50!white,very thick] (10,19.5) -| (14.5,24.5) -| (15,24.5);
\drawo{10}{19}
\drawo{11}{19}
\drawo{12}{19}
\drawo{13}{19}
\drawo{14}{19}
\drawo{14}{20}
\drawo{14}{21}
\drawo{14}{22}
\drawo{14}{23}
\drawo{14}{24}

\draw[orange!50!white,very thick] (15,24.5) -| (19.5,29.5) -| (20,29.5);
\drawo{15}{24}
\drawo{16}{24}
\drawo{17}{24}
\drawo{18}{24}
\drawo{19}{24}
\drawo{19}{25}
\drawo{19}{26}
\drawo{19}{27}
\drawo{19}{28}
\drawo{19}{29}

\draw[orange!50!white,very thick] (20,29.5) -| (24.5,34.5) -| (25,34.5);
\drawo{20}{29}
\drawo{21}{29}
\drawo{22}{29}
\drawo{23}{29}
\drawo{24}{29}
\drawo{24}{30}
\drawo{24}{31}
\drawo{24}{32}
\drawo{24}{33}
\drawo{24}{34}

\draw[orange!50!white,very thick] (25,34.5) -| (29.5,36);
\drawo{25}{34}
\drawo{26}{34}
\drawo{27}{34}
\drawo{28}{34}
\drawo{29}{34}
\drawo{29}{35}

%pompage bleu
\draw[blue,very thick] (3,14.5) |- (1.5,14.5) |- (6.5,12.5) |- (3,10.5);
\draw[blue,very thick] (1.5,13.5) |- (2,13.5);
\drawblue{2}{14}
\drawblue{1}{14}
\drawblue{1}{13}
\drawblue{1}{12}
\drawblue{2}{12}
\drawblue{3}{12}
\drawblue{4}{12}
\drawblue{5}{12}
\drawblue{6}{12}
\drawblue{6}{11}
\drawblue{6}{10}
\drawblue{5}{10}
\drawblue{4}{10}
\drawblue{3}{10}

\draw[blue,very thick] (3,10.5) |- (1.5,10.5) |- (6.5,8.5) |- (3,6.5);
\draw[blue,very thick] (1.5,9.5) |- (2,9.5);
\drawblue{2}{10}
\drawblue{1}{10}
\drawblue{1}{9}
\drawblue{1}{8}
\drawblue{2}{8}
\drawblue{3}{8}
\drawblue{4}{8}
\drawblue{5}{8}
\drawblue{6}{8}
\drawblue{6}{7}
\drawblue{6}{6}
\drawblue{5}{6}
\drawblue{4}{6}
\drawblue{3}{6}

\draw[blue,very thick] (3,6.5) |- (1.5,6.5) |- (6.5,4.5) |- (3,2.5);
\draw[blue,very thick] (1.5,5.5) |- (2,5.5);
\drawblue{2}{6}
\drawblue{1}{6}
\drawblue{1}{5}
\drawblue{1}{4}
\drawblue{2}{4}
\drawblue{3}{4}
\drawblue{4}{4}
\drawblue{5}{4}
\drawblue{6}{4}
\drawblue{6}{3}
\drawblue{6}{2}
\drawblue{5}{2}
\drawblue{4}{2}
\drawblue{3}{2}

\draw[blue,very thick] (3,2.5) |- (1.5,2.5) |- (6.5,0.5) |- (6.5,0);
\draw[blue,very thick] (1.5,1.5) |- (2,1.5);
\drawblue{2}{2}
\drawblue{1}{2}
\drawblue{1}{1}
\drawblue{1}{0}
\drawblue{2}{0}
\drawblue{3}{0}
\drawblue{4}{0}
\drawblue{5}{0}
\drawblue{6}{0}

%pompage vert 
\draw[green,very thick] (9,15.5) |- (15,15.5);
\draw[green,very thick] (10.5,15.5) |- (10.5,16);
\drawgreen{14}{15}
\drawgreen{13}{15}
\drawgreen{12}{15}
\drawgreen{11}{15}
\drawgreen{10}{15}
\drawgreen{9}{15}

\draw[green,very thick] (15,15.5) |- (21,15.5);
\draw[green,very thick] (16.5,15.5) |- (16.5,16);
\drawgreen{20}{15}
\drawgreen{19}{15}
\drawgreen{18}{15}
\drawgreen{17}{15}
\drawgreen{16}{15}
\drawgreen{15}{15}

\draw[green,very thick] (21,15.5) |- (27,15.5);
\draw[green,very thick] (22.5,15.5) |- (22.5,16);
\drawgreen{26}{15}
\drawgreen{25}{15}
\drawgreen{24}{15}
\drawgreen{23}{15}
\drawgreen{22}{15}
\drawgreen{21}{15}

\draw[green,very thick] (27,15.5) |- (33,15.5);
\draw[green,very thick] (28.5,15.5) |- (28.5,16);
\drawgreen{32}{15}
\drawgreen{31}{15}
\drawgreen{30}{15}
\drawgreen{29}{15}
\drawgreen{28}{15}
\drawgreen{27}{15}

\draw[green,very thick] (33,15.5) |- (34,15.5);
\drawgreen{33}{15}

%chemin de liaison vert
\draw[very thick] (10.5,16) |- (10.5,17);
\drawt{10}{16}

\draw[very thick] (16.5,16) |- (19,17.5);
\drawt{16}{16}
\drawt{16}{17}
\drawt{17}{17}
\drawt{18}{17}

\draw[very thick] (22.5,16) |- (25,17.5);
\drawt{22}{16}
\drawt{22}{17}
\drawt{23}{17}
\drawt{24}{17}

\draw[very thick] (28.5,16) |- (31,17.5);
\drawt{28}{16}
\drawt{28}{17}
\drawt{29}{17}
\drawt{30}{17}

%peigne vert 1
\draw[green!75!black,very thick] (19,17.5) -| (20.5,21.5)  -| (21.5,20.5) -| (22,20.5);
\draw[green!75!black,very thick] (19.5,17.5) -| (19.5,18);
\drawgre{19}{17}
\drawgre{20}{17}
\drawgre{20}{18}
\drawgre{20}{19}
\drawgre{20}{20}
\drawgre{20}{21}
\drawgre{21}{21}
\drawgre{21}{20}

\draw[green!75!black,very thick] (22,20.5) -| (23.5,24.5)  -| (24.5,23.5) -| (25,23.5);
\draw[green!75!black,very thick] (22.5,20.5) -| (22.5,21);
\drawgre{22}{20}
\drawgre{23}{20}
\drawgre{23}{21}
\drawgre{23}{22}
\drawgre{23}{23}
\drawgre{23}{24}
\drawgre{24}{24}
\drawgre{24}{23}

\draw[green!75!black,very thick] (25,23.5) -| (26.5,27.5)  -| (27.5,26.5) -| (28,26.5);
\draw[green!75!black,very thick] (25.5,23.5) -| (25.5,24);
\drawgre{25}{23}
\drawgre{26}{23}
\drawgre{26}{24}
\drawgre{26}{25}
\drawgre{26}{26}
\drawgre{26}{27}
\drawgre{27}{27}
\drawgre{27}{26}

\draw[green!75!black,very thick] (28,26.5) -| (29.5,30.5)  -| (30.5,29.5) -| (31,29.5);
\draw[green!75!black,very thick] (28.5,26.5) -| (28.5,27);
\drawgre{28}{26}
\drawgre{29}{26}
\drawgre{29}{27}
\drawgre{29}{28}
\drawgre{29}{29}
\drawgre{29}{30}
\drawgre{30}{30}
\drawgre{30}{29}

\draw[green!75!black,very thick] (31,29.5) -| (32.5,33.5)  -| (33.5,32.5) -| (34,32.5);
\draw[green!75!black,very thick] (31.5,29.5) -| (31.5,30);
\drawgre{31}{29}
\drawgre{32}{29}
\drawgre{32}{30}
\drawgre{32}{31}
\drawgre{32}{32}
\drawgre{32}{33}
\drawgre{33}{33}
\drawgre{33}{32}

\draw[green!75!black,very thick] (25,17.5) -| (26.5,21.5)  -| (27.5,20.5) -| (28,20.5);
\draw[green!75!black,very thick] (25.5,17.5) -| (25.5,18);
\drawgre{25}{17}
\drawgre{26}{17}
\drawgre{26}{18}
\drawgre{26}{19}
\drawgre{26}{20}
\drawgre{26}{21}
\drawgre{27}{21}
\drawgre{27}{20}

\draw[green!75!black,very thick] (28,20.5) -| (29.5,24.5)  -| (30.5,23.5) -| (31,23.5);
\draw[green!75!black,very thick] (28.5,20.5) -| (28.5,21);
\drawgre{28}{20}
\drawgre{29}{20}
\drawgre{29}{21}
\drawgre{29}{22}
\drawgre{29}{23}
\drawgre{29}{24}
\drawgre{30}{24}
\drawgre{30}{23}

\draw[green!75!black,very thick] (31,23.5) -| (32.5,27.5)  -| (33.5,26.5) -| (34,26.5);
\draw[green!75!black,very thick] (31.5,23.5) -| (31.5,24);
\drawgre{31}{23}
\drawgre{32}{23}
\drawgre{32}{24}
\drawgre{32}{25}
\drawgre{32}{26}
\drawgre{32}{27}
\drawgre{33}{27}
\drawgre{33}{26}

\draw[green!75!black,very thick] (31,17.5) -| (32.5,21.5)  -| (33.5,20.5) -| (34,20.5);
\draw[green!75!black,very thick] (31.5,17.5) -| (31.5,18);
\drawgre{31}{17}
\drawgre{32}{17}
\drawgre{32}{18}
\drawgre{32}{19}
\drawgre{32}{20}
\drawgre{32}{21}
\drawgre{33}{21}
\drawgre{33}{20}

%chemin de liaison vert-vert
\draw[very thick] (19.5,18) |- (18.5,18.5) |- (18.5,19);
\drawt{19}{18}
\drawt{18}{18}

\draw[very thick] (22.5,21) |- (22,21.5);
\drawt{22}{21}

\draw[very thick] (25.5,24) |- (25,24.5);
\drawt{25}{24}

\draw[very thick] (28.5,27) |- (28,27.5);
\drawt{28}{27}

\draw[very thick] (31.5,30) |- (31,30.5);
\drawt{31}{30}

\draw[very thick] (25.5,18) |- (24.5,18.5) |- (25,19.5);
\drawt{25}{18}
\drawt{24}{18}
\drawt{24}{19}

\draw[very thick] (28.5,21) |- (28,21.5);
\drawt{28}{21}

\draw[very thick] (31.5,24) |- (31,24.5);
\drawt{31}{24}

\draw[very thick] (31.5,18) |- (30.5,18.5) |- (31,19.5);
\drawt{31}{18}
\drawt{30}{18}
\drawt{30}{19}

%peigne vert 2 
\draw[green!50!black,very thick] (25,19.5) -| (25.5,22.5) |- (28,22.5);
\drawg{25}{19}
\drawg{25}{20}
\drawg{25}{21}
\drawg{25}{22}
\drawg{26}{22}
\drawg{27}{22}

\draw[green!50!black,very thick] (28,22.5) -| (28.5,25.5) |- (31,25.5);
\drawg{28}{22}
\drawg{28}{23}
\drawg{28}{24}
\drawg{28}{25}
\drawg{29}{25}
\drawg{30}{25}

\draw[green!50!black,very thick] (31,25.5) -| (31.5,28.5) |- (34,28.5);
\drawg{31}{25}
\drawg{31}{26}
\drawg{31}{27}
\drawg{31}{28}
\drawg{32}{28}
\drawg{33}{28}

\draw[green!50!black,very thick] (31,19.5) -| (31.5,22.5) |- (34,22.5);
\drawg{31}{19}
\drawg{31}{20}
\drawg{31}{21}
\drawg{31}{22}
\drawg{32}{22}
\drawg{33}{22}

\path [dotted, draw, thin] (0,0) grid[step=0.35cm] (34,36);
\end{tikzpicture}
\caption{the terminal assembly of the tile set described in Figure \ref{fig:appC:step1}.}
\label{fig:appC:step2}
\end{figure}
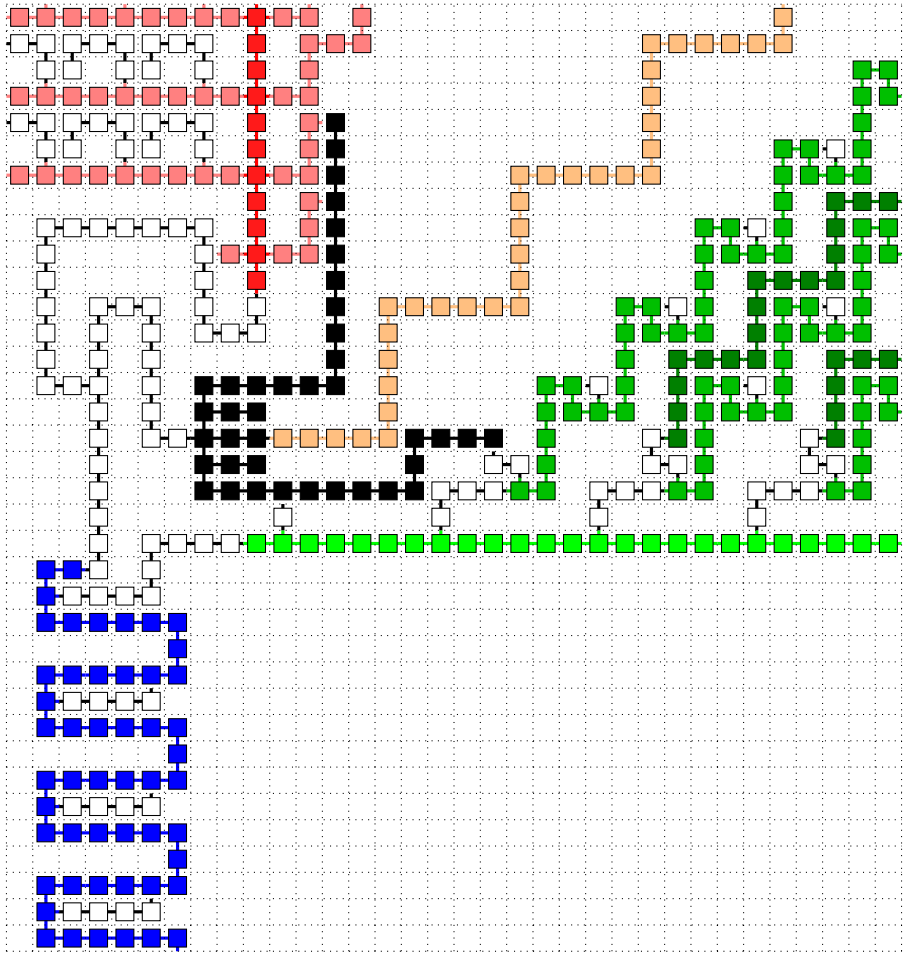

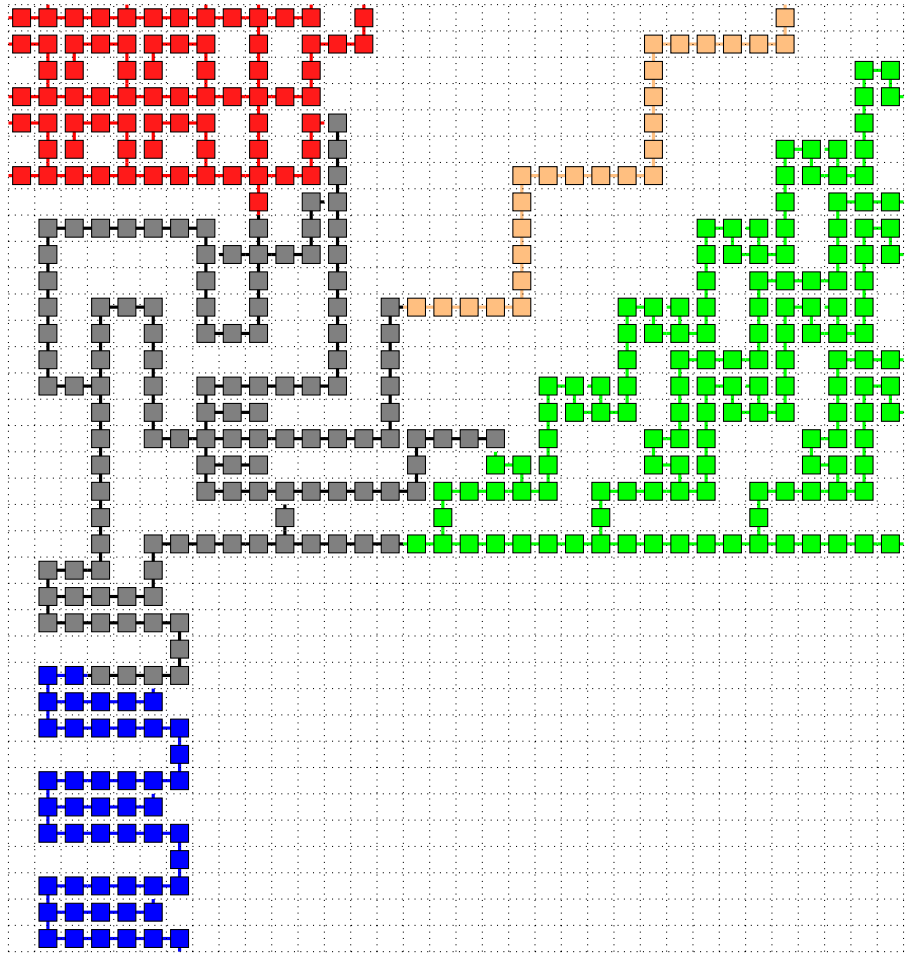
\begin{figure}
\centering
\begin{tikzpicture}[x=0.35cm,y=0.35cm]

%graine
\draw[very thick] (7.5,21.5) -| (12.5,31.5);
\draw[very thick] (7.5,21.5) |- (15.5,17.5) |- (18.5,19.5);
\draw[very thick] (7.5,20.5) |- (9.5,20.5);
\draw[very thick] (7,19.5) |- (10,19.5);
\draw[very thick] (7.5,18.5) |- (9.5,18.5);
\drawgray{7}{21}
\drawgray{7}{20}
\drawgray{7}{19}
\drawgray{7}{18}
\drawgray{7}{17}
\drawgray{8}{21}
\drawgray{8}{20}
\drawgray{8}{19}
\drawgray{8}{18}
\drawgray{8}{17}
\drawgray{9}{21}
\drawgray{9}{20}
\drawgray{9}{19}
\drawgray{9}{18}
\drawgray{9}{17}
\drawgray{10}{21}
\drawgray{11}{21}
\drawgray{12}{21}
\drawgray{12}{22}
\drawgray{12}{23}
\drawgray{12}{24}
\drawgray{12}{25}
\drawgray{12}{26}
\drawgray{12}{27}
\drawgray{12}{28}
\drawgray{12}{29}
\drawgray{12}{30}
\drawgray{12}{31}
\drawgray{10}{17}
\drawgray{11}{17}
\drawgray{12}{17}
\drawgray{13}{17}
\drawgray{14}{17}
\drawgray{15}{17}
\drawgray{15}{18}
\drawgray{15}{19}
\drawgray{16}{19}
\drawgray{17}{19}
\drawgray{18}{19}

%chemin de liaison 1
\draw[very thick] (7,19.5) -| (5.5,24.5) -| (3.5,14.5) -| (3,14.5);
\draw[very thick] (3.5,21.5) -| (1.5,27.5) -| (7.5,23.5) -| (9.5,25) ;
\drawgray{6}{19}
\drawgray{5}{19}
\drawgray{5}{20}
\drawgray{5}{21}
\drawgray{5}{22}
\drawgray{5}{23}
\drawgray{2}{21}
\drawgray{1}{21}
\drawgray{1}{22}
\drawgray{1}{23}
\drawgray{1}{24}
\drawgray{1}{25}
\drawgray{1}{26}
\drawgray{1}{27}
\drawgray{2}{27}
\drawgray{3}{27}
\drawgray{4}{27}
\drawgray{5}{27}
\drawgray{6}{27}
\drawgray{7}{27}
\drawgray{7}{26}
\drawgray{7}{25}
\drawgray{7}{24}
\drawgray{7}{23}
\drawgray{8}{23}
\drawgray{9}{23}
\drawgray{9}{24}
\drawgray{5}{24}
\drawgray{4}{24}
\drawgray{3}{24}
\drawgray{3}{23}
\drawgray{3}{22}
\drawgray{3}{21}
\drawgray{3}{20}
\drawgray{3}{19}
\drawgray{3}{18}
\drawgray{3}{17}
\drawgray{3}{16}
\drawgray{3}{15}
\drawgray{3}{14}

%chemin de liaison 2
\draw[very thick] (2,13.5) -| (5.5,15.5)  -| (9,15.5);
\drawgray{2}{13}
\drawgray{3}{13}
\drawgray{4}{13}
\drawgray{5}{13}
\drawgray{5}{14}
\drawgray{5}{15}
\drawgray{6}{15}
\drawgray{7}{15}
\drawgray{8}{15}

\draw[blue,very thick] (2,9.5) -| (5.5,10);
\drawblue{2}{9}
\drawblue{3}{9}
\drawblue{4}{9}
\drawblue{5}{9}

\draw[blue,very thick] (2,5.5) -| (5.5,6);
\drawblue{2}{5}
\drawblue{3}{5}
\drawblue{4}{5}
\drawblue{5}{5}

\draw[blue,very thick] (2,1.5) -| (5.5,2);
\drawblue{2}{1}
\drawblue{3}{1}
\drawblue{4}{1}
\drawblue{5}{1}

%chemin fini sur peigne
\draw[red,very thick] (7.5,30) |- (5.5,31.5)  -| (5.5,30.5);
\drawred{7}{30}
\drawred{7}{31}
\drawred{6}{31}
\drawred{5}{31}
\drawred{5}{30}

\draw[red,very thick] (4.5,30) |- (2.5,31.5)  -| (2.5,30.5);
\drawred{4}{30}
\drawred{4}{31}
\drawred{3}{31}
\drawred{2}{31}
\drawred{2}{30}

\draw[red,very thick] (1.5,30) |- (0,31.5);
\drawred{1}{30}
\drawred{1}{31}
\drawred{0}{31}

\draw[red,very thick] (7.5,33) |- (5.5,34.5)  -| (5.5,33.5);
\drawred{7}{33}
\drawred{7}{34}
\drawred{6}{34}
\drawred{5}{34}
\drawred{5}{33}

\draw[red,very thick] (4.5,33) |- (2.5,34.5)  -| (2.5,33.5);
\drawred{4}{33}
\drawred{4}{34}
\drawred{3}{34}
\drawred{2}{34}
\drawred{2}{33}

\draw[red,very thick] (1.5,33) |- (0,34.5);
\drawred{1}{33}
\drawred{1}{34}
\drawred{0}{34}

%pompage rouge principale
\draw[red,very thick] (9.5,28) |- (9.5,36);
\draw[very thick] (9.5,25) |- (9.5,28);
\draw[very thick] (9,26.5) |- (10,26.5);
\draw[red,very thick] (9,29.5) |- (10,29.5);
\draw[red,very thick] (9,32.5) |- (10,32.5);
\draw[red,very thick] (9,35.5) |- (10,35.5);
\drawgray{9}{25}
\drawgray{9}{26}
\drawgray{9}{27}
\drawred{9}{28}
\drawred{9}{29}
\drawred{9}{30}
\drawred{9}{31}
\drawred{9}{32}
\drawred{9}{33}
\drawred{9}{34}
\drawred{9}{35}

%peigne rouge 1
\draw[very thick] (8,26.5) |- (9,26.5);
\drawgray{8}{26}

\draw[red,very thick] (0,29.5) |- (9,29.5);
\draw[red,very thick] (7.5,29.5) |- (7.5,30);
\draw[red,very thick] (4.5,29.5) |- (4.5,30);
\draw[red,very thick] (1.5,29.5) |- (1.5,30);
\drawred{8}{29}
\drawred{7}{29}
\drawred{6}{29}
\drawred{5}{29}
\drawred{4}{29}
\drawred{3}{29}
\drawred{2}{29}
\drawred{1}{29}
\drawred{0}{29}

\draw[red,very thick] (0,32.5) |- (9,32.5);
\draw[red,very thick] (7.5,32.5) |- (7.5,33);
\draw[red,very thick] (4.5,32.5) |- (4.5,33);
\draw[red,very thick] (1.5,32.5) |- (1.5,33);
\drawred{8}{32}
\drawred{7}{32}
\drawred{6}{32}
\drawred{5}{32}
\drawred{4}{32}
\drawred{3}{32}
\drawred{2}{32}
\drawred{1}{32}
\drawred{0}{32}

\draw[red,very thick] (0,35.5) |- (9,35.5);
\draw[red,very thick] (7.5,35.5) |- (7.5,36);
\draw[red,very thick] (4.5,35.5) |- (4.5,36);
\draw[red,very thick] (1.5,35.5) |- (1.5,36);
\drawred{8}{35}
\drawred{7}{35}
\drawred{6}{35}
\drawred{5}{35}
\drawred{4}{35}
\drawred{3}{35}
\drawred{2}{35}
\drawred{1}{35}
\drawred{0}{35}

%peigne rouge 2
\draw[red,very thick] (10,35.5) -| (11.5,36);
\drawred{10}{35}
\drawred{11}{35}

\draw[red,very thick] (10,32.5) -| (11.5,34.5) -| (13.5,36);
\drawred{10}{32}
\drawred{11}{32}
\drawred{11}{33}
\drawred{11}{34}
\drawred{12}{34}
\drawred{13}{34}
\drawred{13}{35}

\draw[red,very thick] (10,29.5) -| (11.5,31.5) -| (12,31.5);
\drawred{10}{29}
\drawred{11}{29}
\drawred{11}{30}
\drawred{11}{31}

\draw[very thick] (10,26.5) -| (11.5,28.5) -| (12,28.5);
\drawgray{10}{26}
\drawgray{11}{26}
\drawgray{11}{27}
\drawgray{11}{28}

%pompage orange
\draw[very thick] (10,19.5) -| (14.5,24.5) -| (15,24.5);
\drawgray{10}{19}
\drawgray{11}{19}
\drawgray{12}{19}
\drawgray{13}{19}
\drawgray{14}{19}
\drawgray{14}{20}
\drawgray{14}{21}
\drawgray{14}{22}
\drawgray{14}{23}
\drawgray{14}{24}

\draw[orange!50!white,very thick] (15,24.5) -| (19.5,29.5) -| (20,29.5);
\drawo{15}{24}
\drawo{16}{24}
\drawo{17}{24}
\drawo{18}{24}
\drawo{19}{24}
\drawo{19}{25}
\drawo{19}{26}
\drawo{19}{27}
\drawo{19}{28}
\drawo{19}{29}

\draw[orange!50!white,very thick] (20,29.5) -| (24.5,34.5) -| (25,34.5);
\drawo{20}{29}
\drawo{21}{29}
\drawo{22}{29}
\drawo{23}{29}
\drawo{24}{29}
\drawo{24}{30}
\drawo{24}{31}
\drawo{24}{32}
\drawo{24}{33}
\drawo{24}{34}

\draw[orange!50!white,very thick] (25,34.5) -| (29.5,36);
\drawo{25}{34}
\drawo{26}{34}
\drawo{27}{34}
\drawo{28}{34}
\drawo{29}{34}
\drawo{29}{35}

%pompage bleu
\draw[very thick] (3,14.5) |- (1.5,14.5) |- (6.5,12.5) |- (3,10.5);
\draw[very thick] (1.5,13.5) |- (2,13.5);
\drawgray{2}{14}
\drawgray{1}{14}
\drawgray{1}{13}
\drawgray{1}{12}
\drawgray{2}{12}
\drawgray{3}{12}
\drawgray{4}{12}
\drawgray{5}{12}
\drawgray{6}{12}
\drawgray{6}{11}
\drawgray{6}{10}
\drawgray{5}{10}
\drawgray{4}{10}
\drawgray{3}{10}

\draw[blue,very thick] (3,10.5) |- (1.5,10.5) |- (6.5,8.5) |- (3,6.5);
\draw[blue,very thick] (1.5,9.5) |- (2,9.5);
\drawblue{2}{10}
\drawblue{1}{10}
\drawblue{1}{9}
\drawblue{1}{8}
\drawblue{2}{8}
\drawblue{3}{8}
\drawblue{4}{8}
\drawblue{5}{8}
\drawblue{6}{8}
\drawblue{6}{7}
\drawblue{6}{6}
\drawblue{5}{6}
\drawblue{4}{6}
\drawblue{3}{6}

\draw[blue,very thick] (3,6.5) |- (1.5,6.5) |- (6.5,4.5) |- (3,2.5);
\draw[blue,very thick] (1.5,5.5) |- (2,5.5);
\drawblue{2}{6}
\drawblue{1}{6}
\drawblue{1}{5}
\drawblue{1}{4}
\drawblue{2}{4}
\drawblue{3}{4}
\drawblue{4}{4}
\drawblue{5}{4}
\drawblue{6}{4}
\drawblue{6}{3}
\drawblue{6}{2}
\drawblue{5}{2}
\drawblue{4}{2}
\drawblue{3}{2}

\draw[blue,very thick] (3,2.5) |- (1.5,2.5) |- (6.5,0.5) |- (6.5,0);
\draw[blue,very thick] (1.5,1.5) |- (2,1.5);
\drawblue{2}{2}
\drawblue{1}{2}
\drawblue{1}{1}
\drawblue{1}{0}
\drawblue{2}{0}
\drawblue{3}{0}
\drawblue{4}{0}
\drawblue{5}{0}
\drawblue{6}{0}

%pompage vert 
\draw[very thick] (9,15.5) |- (15,15.5);
\draw[very thick] (10.5,15.5) |- (10.5,16);
\drawgray{14}{15}
\drawgray{13}{15}
\drawgray{12}{15}
\drawgray{11}{15}
\drawgray{10}{15}
\drawgray{9}{15}

\draw[green,very thick] (15,15.5) |- (21,15.5);
\draw[green,very thick] (16.5,15.5) |- (16.5,16);
\drawgreen{20}{15}
\drawgreen{19}{15}
\drawgreen{18}{15}
\drawgreen{17}{15}
\drawgreen{16}{15}
\drawgreen{15}{15}

\draw[green,very thick] (21,15.5) |- (27,15.5);
\draw[green,very thick] (22.5,15.5) |- (22.5,16);
\drawgreen{26}{15}
\drawgreen{25}{15}
\drawgreen{24}{15}
\drawgreen{23}{15}
\drawgreen{22}{15}
\drawgreen{21}{15}

\draw[green,very thick] (27,15.5) |- (33,15.5);
\draw[green,very thick] (28.5,15.5) |- (28.5,16);
\drawgreen{32}{15}
\drawgreen{31}{15}
\drawgreen{30}{15}
\drawgreen{29}{15}
\drawgreen{28}{15}
\drawgreen{27}{15}

\draw[green,very thick] (33,15.5) |- (34,15.5);
\drawgreen{33}{15}

%chemin de liaison vert
\draw[very thick] (10.5,16) |- (10.5,17);
\drawgray{10}{16}

\draw[green,very thick] (16.5,16) |- (19,17.5);
\drawgreen{16}{16}
\drawgreen{16}{17}
\drawgreen{17}{17}
\drawgreen{18}{17}

\draw[green,very thick] (22.5,16) |- (25,17.5);
\drawgreen{22}{16}
\drawgreen{22}{17}
\drawgreen{23}{17}
\drawgreen{24}{17}

\draw[green,very thick] (28.5,16) |- (31,17.5);
\drawgreen{28}{16}
\drawgreen{28}{17}
\drawgreen{29}{17}
\drawgreen{30}{17}

%peigne vert 1
\draw[green,very thick] (19,17.5) -| (20.5,21.5)  -| (21.5,20.5) -| (22,20.5);
\draw[green,very thick] (19.5,17.5) -| (19.5,18);
\drawgreen{19}{17}
\drawgreen{20}{17}
\drawgreen{20}{18}
\drawgreen{20}{19}
\drawgreen{20}{20}
\drawgreen{20}{21}
\drawgreen{21}{21}
\drawgreen{21}{20}

\draw[green,very thick] (22,20.5) -| (23.5,24.5)  -| (24.5,23.5) -| (25,23.5);
\draw[green,very thick] (22.5,20.5) -| (22.5,21);
\drawgreen{22}{20}
\drawgreen{23}{20}
\drawgreen{23}{21}
\drawgreen{23}{22}
\drawgreen{23}{23}
\drawgreen{23}{24}
\drawgreen{24}{24}
\drawgreen{24}{23}

\draw[green,very thick] (25,23.5) -| (26.5,27.5)  -| (27.5,26.5) -| (28,26.5);
\draw[green,very thick] (25.5,23.5) -| (25.5,24);
\drawgreen{25}{23}
\drawgreen{26}{23}
\drawgreen{26}{24}
\drawgreen{26}{25}
\drawgreen{26}{26}
\drawgreen{26}{27}
\drawgreen{27}{27}
\drawgreen{27}{26}

\draw[green,very thick] (28,26.5) -| (29.5,30.5)  -| (30.5,29.5) -| (31,29.5);
\draw[green,very thick] (28.5,26.5) -| (28.5,27);
\drawgreen{28}{26}
\drawgreen{29}{26}
\drawgreen{29}{27}
\drawgreen{29}{28}
\drawgreen{29}{29}
\drawgreen{29}{30}
\drawgreen{30}{30}
\drawgreen{30}{29}

\draw[green,very thick] (31,29.5) -| (32.5,33.5)  -| (33.5,32.5) -| (34,32.5);
\draw[green,very thick] (31.5,29.5) -| (31.5,30);
\drawgreen{31}{29}
\drawgreen{32}{29}
\drawgreen{32}{30}
\drawgreen{32}{31}
\drawgreen{32}{32}
\drawgreen{32}{33}
\drawgreen{33}{33}
\drawgreen{33}{32}

\draw[green,very thick] (25,17.5) -| (26.5,21.5)  -| (27.5,20.5) -| (28,20.5);
\draw[green,very thick] (25.5,17.5) -| (25.5,18);
\drawgreen{25}{17}
\drawgreen{26}{17}
\drawgreen{26}{18}
\drawgreen{26}{19}
\drawgreen{26}{20}
\drawgreen{26}{21}
\drawgreen{27}{21}
\drawgreen{27}{20}

\draw[green,very thick] (28,20.5) -| (29.5,24.5)  -| (30.5,23.5) -| (31,23.5);
\draw[green,very thick] (28.5,20.5) -| (28.5,21);
\drawgreen{28}{20}
\drawgreen{29}{20}
\drawgreen{29}{21}
\drawgreen{29}{22}
\drawgreen{29}{23}
\drawgreen{29}{24}
\drawgreen{30}{24}
\drawgreen{30}{23}

\draw[green,very thick] (31,23.5) -| (32.5,27.5)  -| (33.5,26.5) -| (34,26.5);
\draw[green,very thick] (31.5,23.5) -| (31.5,24);
\drawgreen{31}{23}
\drawgreen{32}{23}
\drawgreen{32}{24}
\drawgreen{32}{25}
\drawgreen{32}{26}
\drawgreen{32}{27}
\drawgreen{33}{27}
\drawgreen{33}{26}

\draw[green,very thick] (31,17.5) -| (32.5,21.5)  -| (33.5,20.5) -| (34,20.5);
\draw[green,very thick] (31.5,17.5) -| (31.5,18);
\drawgreen{31}{17}
\drawgreen{32}{17}
\drawgreen{32}{18}
\drawgreen{32}{19}
\drawgreen{32}{20}
\drawgreen{32}{21}
\drawgreen{33}{21}
\drawgreen{33}{20}

%chemin de liaison vert-vert
\draw[green,very thick] (19.5,18) |- (18.5,18.5) |- (18.5,19);
\drawgreen{19}{18}
\drawgreen{18}{18}

\draw[green,very thick] (22.5,21) |- (22,21.5);
\drawgreen{22}{21}

\draw[green,very thick] (25.5,24) |- (25,24.5);
\drawgreen{25}{24}

\draw[green,very thick] (28.5,27) |- (28,27.5);
\drawgreen{28}{27}

\draw[green,very thick] (31.5,30) |- (31,30.5);
\drawgreen{31}{30}

\draw[green,very thick] (25.5,18) |- (24.5,18.5) |- (25,19.5);
\drawgreen{25}{18}
\drawgreen{24}{18}
\drawgreen{24}{19}

\draw[green,very thick] (28.5,21) |- (28,21.5);
\drawgreen{28}{21}

\draw[green,very thick] (31.5,24) |- (31,24.5);
\drawgreen{31}{24}

\draw[green,very thick] (31.5,18) |- (30.5,18.5) |- (31,19.5);
\drawgreen{31}{18}
\drawgreen{30}{18}
\drawgreen{30}{19}

%peigne vert 2 
\draw[green,very thick] (25,19.5) -| (25.5,22.5) |- (28,22.5);
\drawgreen{25}{19}
\drawgreen{25}{20}
\drawgreen{25}{21}
\drawgreen{25}{22}
\drawgreen{26}{22}
\drawgreen{27}{22}

\draw[green,very thick] (28,22.5) -| (28.5,25.5) |- (31,25.5);
\drawgreen{28}{22}
\drawgreen{28}{23}
\drawgreen{28}{24}
\drawgreen{28}{25}
\drawgreen{29}{25}
\drawgreen{30}{25}

\draw[green,very thick] (31,25.5) -| (31.5,28.5) |- (34,28.5);
\drawgreen{31}{25}
\drawgreen{31}{26}
\drawgreen{31}{27}
\drawgreen{31}{28}
\drawgreen{32}{28}
\drawgreen{33}{28}

\draw[green,very thick] (31,19.5) -| (31.5,22.5) |- (34,22.5);
\drawgreen{31}{19}
\drawgreen{31}{20}
\drawgreen{31}{21}
\drawgreen{31}{22}
\drawgreen{32}{22}
\drawgreen{33}{22}

\path [dotted, draw, thin] (0,0) grid[step=0.35cm] (34,36);
\end{tikzpicture}
\caption{the terminal assembly is composed of a finite assembly in gray and four assemblies of complexity less than $2$ (in blue, red, orange, and green).}
\label{fig:appC:step3}
\end{figure}

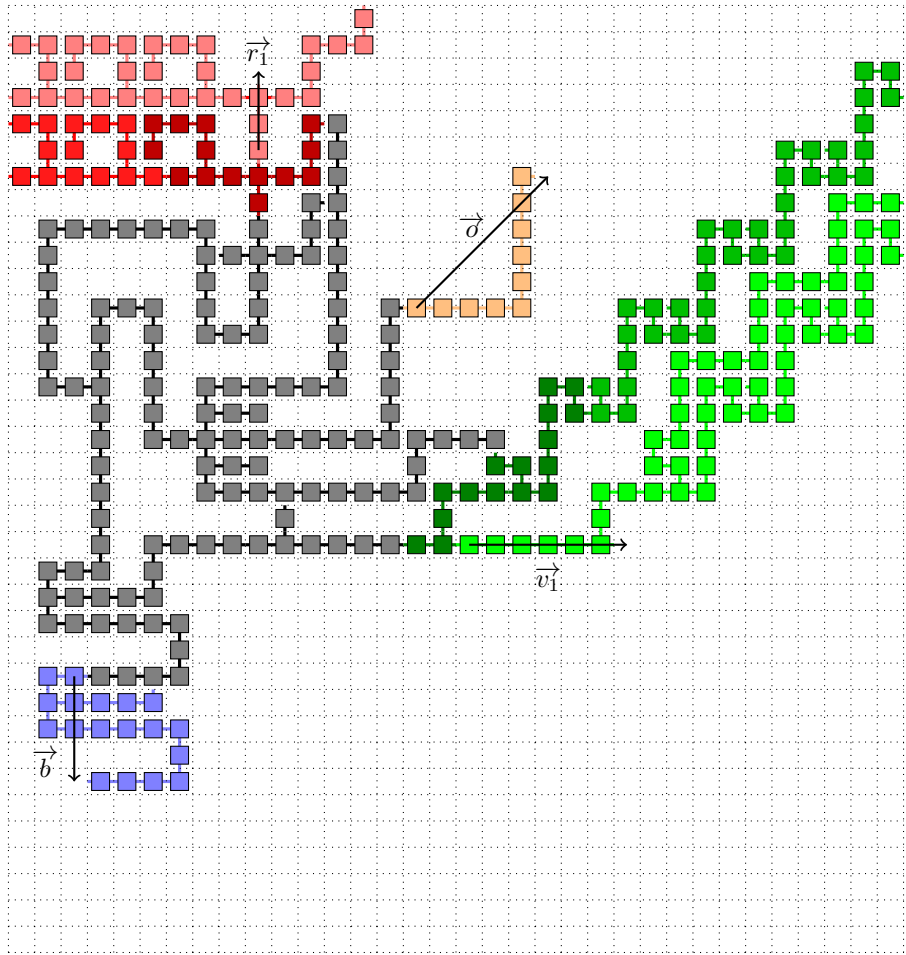
\begin{figure}
\centering
\begin{tikzpicture}[x=0.35cm,y=0.35cm]

%graine
\draw[very thick] (7.5,21.5) -| (12.5,31.5);
\draw[very thick] (7.5,21.5) |- (15.5,17.5) |- (18.5,19.5);
\draw[very thick] (7.5,20.5) |- (9.5,20.5);
\draw[very thick] (7,19.5) |- (10,19.5);
\draw[very thick] (7.5,18.5) |- (9.5,18.5);
\drawgray{7}{21}
\drawgray{7}{20}
\drawgray{7}{19}
\drawgray{7}{18}
\drawgray{7}{17}
\drawgray{8}{21}
\drawgray{8}{20}
\drawgray{8}{19}
\drawgray{8}{18}
\drawgray{8}{17}
\drawgray{9}{21}
\drawgray{9}{20}
\drawgray{9}{19}
\drawgray{9}{18}
\drawgray{9}{17}
\drawgray{10}{21}
\drawgray{11}{21}
\drawgray{12}{21}
\drawgray{12}{22}
\drawgray{12}{23}
\drawgray{12}{24}
\drawgray{12}{25}
\drawgray{12}{26}
\drawgray{12}{27}
\drawgray{12}{28}
\drawgray{12}{29}
\drawgray{12}{30}
\drawgray{12}{31}
\drawgray{10}{17}
\drawgray{11}{17}
\drawgray{12}{17}
\drawgray{13}{17}
\drawgray{14}{17}
\drawgray{15}{17}
\drawgray{15}{18}
\drawgray{15}{19}
\drawgray{16}{19}
\drawgray{17}{19}
\drawgray{18}{19}

%chemin de liaison 1
\draw[very thick] (7,19.5) -| (5.5,24.5) -| (3.5,14.5) -| (3,14.5);
\draw[very thick] (3.5,21.5) -| (1.5,27.5) -| (7.5,23.5) -| (9.5,25) ;
\drawgray{6}{19}
\drawgray{5}{19}
\drawgray{5}{20}
\drawgray{5}{21}
\drawgray{5}{22}
\drawgray{5}{23}
\drawgray{2}{21}
\drawgray{1}{21}
\drawgray{1}{22}
\drawgray{1}{23}
\drawgray{1}{24}
\drawgray{1}{25}
\drawgray{1}{26}
\drawgray{1}{27}
\drawgray{2}{27}
\drawgray{3}{27}
\drawgray{4}{27}
\drawgray{5}{27}
\drawgray{6}{27}
\drawgray{7}{27}
\drawgray{7}{26}
\drawgray{7}{25}
\drawgray{7}{24}
\drawgray{7}{23}
\drawgray{8}{23}
\drawgray{9}{23}
\drawgray{9}{24}
\drawgray{5}{24}
\drawgray{4}{24}
\drawgray{3}{24}
\drawgray{3}{23}
\drawgray{3}{22}
\drawgray{3}{21}
\drawgray{3}{20}
\drawgray{3}{19}
\drawgray{3}{18}
\drawgray{3}{17}
\drawgray{3}{16}
\drawgray{3}{15}
\drawgray{3}{14}

%chemin de liaison 2
\draw[very thick] (2,13.5) -| (5.5,15.5)  -| (9,15.5);
\drawgray{2}{13}
\drawgray{3}{13}
\drawgray{4}{13}
\drawgray{5}{13}
\drawgray{5}{14}
\drawgray{5}{15}
\drawgray{6}{15}
\drawgray{7}{15}
\drawgray{8}{15}

\draw[blue!50!white,very thick] (2,9.5) -| (5.5,10);
\drawb{2}{9}
\drawb{3}{9}
\drawb{4}{9}
\drawb{5}{9}

%chemin fini sur peigne
\draw[red!75!black,very thick] (7.5,30) |- (5.5,31.5)  -| (5.5,30.5);
\drawre{7}{30}
\drawre{7}{31}
\drawre{6}{31}
\drawre{5}{31}
\drawre{5}{30}

\draw[red,very thick] (4.5,30) |- (2.5,31.5)  -| (2.5,30.5);
\drawred{4}{30}
\drawred{4}{31}
\drawred{3}{31}
\drawred{2}{31}
\drawred{2}{30}

\draw[red,very thick] (1.5,30) |- (0,31.5);
\drawred{1}{30}
\drawred{1}{31}
\drawred{0}{31}

\draw[red!50!white,very thick] (7.5,33) |- (5.5,34.5)  -| (5.5,33.5);
\drawr{7}{33}
\drawr{7}{34}
\drawr{6}{34}
\drawr{5}{34}
\drawr{5}{33}

\draw[red!50!white,very thick] (4.5,33) |- (2.5,34.5)  -| (2.5,33.5);
\drawr{4}{33}
\drawr{4}{34}
\drawr{3}{34}
\drawr{2}{34}
\drawr{2}{33}

\draw[red!50!white,very thick] (1.5,33) |- (0,34.5);
\drawr{1}{33}
\drawr{1}{34}
\drawr{0}{34}

%pompage rouge principale
\draw(9.5,34.2)node {$\vect{r_1}$};
\draw[very thick] (9.5,25) |- (9.5,28);
\draw[very thick] (9,26.5) |- (10,26.5);
\draw[red!75!black,very thick] (9.5,28) |- (9.5,30);
\draw[red!50!white,very thick] (9.5,30) |- (9.5,33);
\draw[red!75!black,very thick] (9,29.5) |- (10,29.5);
\draw[red,very thick] (9,32.5) |- (10,32.5);
\drawgray{9}{25}
\drawgray{9}{26}
\drawgray{9}{27}
\drawre{9}{28}
\drawre{9}{29}
\drawr{9}{30}
\drawr{9}{31}
\drawr{9}{32}
\draw[thick,->] (9.5,30.5) -- (9.5,33.5);

%peigne rouge 1
\draw[very thick] (8,26.5) |- (9,26.5);
\drawgray{8}{26}

\draw[red,very thick] (0,29.5) |- (6,29.5);
\draw[red!75!black,very thick] (6,29.5) |- (9,29.5);
\draw[red!75!black,very thick] (7.5,29.5) |- (7.5,30);
\draw[red,very thick] (4.5,29.5) |- (4.5,30);
\draw[red,very thick] (1.5,29.5) |- (1.5,30);
\drawre{8}{29}
\drawre{7}{29}
\drawre{6}{29}
\drawred{5}{29}
\drawred{4}{29}
\drawred{3}{29}
\drawred{2}{29}
\drawred{1}{29}
\drawred{0}{29}

\draw[red!50!white,very thick] (0,32.5) |- (9,32.5);
\draw[red!50!white,very thick] (7.5,32.5) |- (7.5,33);
\draw[red!50!white,very thick] (4.5,32.5) |- (4.5,33);
\draw[red!50!white,very thick] (1.5,32.5) |- (1.5,33);
\drawr{8}{32}
\drawr{7}{32}
\drawr{6}{32}
\drawr{5}{32}
\drawr{4}{32}
\drawr{3}{32}
\drawr{2}{32}
\drawr{1}{32}
\drawr{0}{32}

%peigne rouge 2
\draw[red!50!white,very thick] (10,32.5) -| (11.5,34.5) -| (13.5,36);
\drawr{10}{32}
\drawr{11}{32}
\drawr{11}{33}
\drawr{11}{34}
\drawr{12}{34}
\drawr{13}{34}
\drawr{13}{35}

\draw[red!75!black,very thick] (10,29.5) -| (11.5,31.5) -| (12,31.5);
\drawre{10}{29}
\drawre{11}{29}
\drawre{11}{30}
\drawre{11}{31}

\draw[very thick] (10,26.5) -| (11.5,28.5) -| (12,28.5);
\drawgray{10}{26}
\drawgray{11}{26}
\drawgray{11}{27}
\drawgray{11}{28}

%pompage orange
\draw[very thick] (10,19.5) -| (14.5,24.5) -| (15,24.5);
\drawgray{10}{19}
\drawgray{11}{19}
\drawgray{12}{19}
\drawgray{13}{19}
\drawgray{14}{19}
\drawgray{14}{20}
\drawgray{14}{21}
\drawgray{14}{22}
\drawgray{14}{23}
\drawgray{14}{24}

\draw[orange!50!white,very thick] (15,24.5) -| (19.5,29.5) -| (20,29.5);
\drawo{15}{24}
\drawo{16}{24}
\drawo{17}{24}
\drawo{18}{24}
\drawo{19}{24}
\drawo{19}{25}
\drawo{19}{26}
\drawo{19}{27}
\drawo{19}{28}
\drawo{19}{29}

\draw(17.6,27.6) node {$\vect{o}$};
\draw[thick,->] (15.5,24.5) -- (20.5,29.5);

%pompage bleu
\draw[very thick] (3,14.5) |- (1.5,14.5) |- (6.5,12.5) |- (3,10.5);
\draw[very thick] (1.5,13.5) |- (2,13.5);
\drawgray{2}{14}
\drawgray{1}{14}
\drawgray{1}{13}
\drawgray{1}{12}
\drawgray{2}{12}
\drawgray{3}{12}
\drawgray{4}{12}
\drawgray{5}{12}
\drawgray{6}{12}
\drawgray{6}{11}
\drawgray{6}{10}
\drawgray{5}{10}
\drawgray{4}{10}
\drawgray{3}{10}

\draw[blue!50!white,very thick] (3,10.5) |- (1.5,10.5) |- (6.5,8.5) |- (3,6.5);
\draw[blue!50!white,very thick] (1.5,9.5) |- (2,9.5);
\drawb{2}{10}
\drawb{1}{10}
\drawb{1}{9}
\drawb{1}{8}
\drawb{2}{8}
\drawb{3}{8}
\drawb{4}{8}
\drawb{5}{8}
\drawb{6}{8}
\drawb{6}{7}
\drawb{6}{6}
\drawb{5}{6}
\drawb{4}{6}
\drawb{3}{6}

\draw(1.4,7.2) node {$\vect{b}$};
\draw[thick,->] (2.5,10.5) -- (2.5,6.5);

%pompage vert 
\draw[very thick] (9,15.5) |- (15,15.5);
\draw[very thick] (10.5,15.5) |- (10.5,16);
\drawgray{14}{15}
\drawgray{13}{15}
\drawgray{12}{15}
\drawgray{11}{15}
\drawgray{10}{15}
\drawgray{9}{15}

\draw[green!50!black,very thick] (15,15.5) |- (17,15.5);
\draw[green,very thick] (17,15.5) |- (21,15.5);
\draw[green!50!black,very thick] (16.5,15.5) |- (16.5,16);
\drawgreen{20}{15}
\drawgreen{19}{15}
\drawgreen{18}{15}
\drawgreen{17}{15}
\drawg{16}{15}
\drawg{15}{15}

\draw[green,very thick] (21,15.5) |- (23,15.5);
\draw[green,very thick] (22.5,15.5) |- (22.5,16);
\drawgreen{22}{15}
\drawgreen{21}{15}

\draw(20.5,14.2) node {$\vect{v_1}$};
\draw[thick,->] (17.5,15.5) -- (23.5,15.5);

%chemin de liaison vert
\draw[very thick] (10.5,16) |- (10.5,17);
\drawgray{10}{16}

\draw[green!50!black,very thick] (16.5,16) |- (19,17.5);
\drawg{16}{16}
\drawg{16}{17}
\drawg{17}{17}
\drawg{18}{17}

\draw[green,very thick] (22.5,16) |- (25,17.5);
\drawgreen{22}{16}
\drawgreen{22}{17}
\drawgreen{23}{17}
\drawgreen{24}{17}

%peigne vert 1
\draw[green!50!black,very thick] (19,17.5) -| (20.5,21.5)  -| (21.5,20.5) -| (22,20.5);
\draw[green!50!black,very thick] (19.5,17.5) -| (19.5,18);
\drawg{19}{17}
\drawg{20}{17}
\drawg{20}{18}
\drawg{20}{19}
\drawg{20}{20}
\drawg{20}{21}
\drawg{21}{21}
\drawg{21}{20}

\draw[green!75!black,very thick] (22,20.5) -| (23.5,24.5)  -| (24.5,23.5) -| (25,23.5);
\draw[green!75!black,very thick] (22.5,20.5) -| (22.5,21);
\drawgre{22}{20}
\drawgre{23}{20}
\drawgre{23}{21}
\drawgre{23}{22}
\drawgre{23}{23}
\drawgre{23}{24}
\drawgre{24}{24}
\drawgre{24}{23}

\draw[green!75!black,very thick] (25,23.5) -| (26.5,27.5)  -| (27.5,26.5) -| (28,26.5);
\draw[green!75!black,very thick] (25.5,23.5) -| (25.5,24);
\drawgre{25}{23}
\drawgre{26}{23}
\drawgre{26}{24}
\drawgre{26}{25}
\drawgre{26}{26}
\drawgre{26}{27}
\drawgre{27}{27}
\drawgre{27}{26}

\draw[green!75!black,very thick] (28,26.5) -| (29.5,30.5)  -| (30.5,29.5) -| (31,29.5);
\draw[green!75!black,very thick] (28.5,26.5) -| (28.5,27);
\drawgre{28}{26}
\drawgre{29}{26}
\drawgre{29}{27}
\drawgre{29}{28}
\drawgre{29}{29}
\drawgre{29}{30}
\drawgre{30}{30}
\drawgre{30}{29}

\draw[green!75!black,very thick] (31,29.5) -| (32.5,33.5)  -| (33.5,32.5) -| (34,32.5);
\draw[green!75!black,very thick] (31.5,29.5) -| (31.5,30);
\drawgre{31}{29}
\drawgre{32}{29}
\drawgre{32}{30}
\drawgre{32}{31}
\drawgre{32}{32}
\drawgre{32}{33}
\drawgre{33}{33}
\drawgre{33}{32}

\draw[green,very thick] (25,17.5) -| (26.5,21.5)  -| (27.5,20.5) -| (28,20.5);
\draw[green,very thick] (25.5,17.5) -| (25.5,18);
\drawgreen{25}{17}
\drawgreen{26}{17}
\drawgreen{26}{18}
\drawgreen{26}{19}
\drawgreen{26}{20}
\drawgreen{26}{21}
\drawgreen{27}{21}
\drawgreen{27}{20}

\draw[green,very thick] (28,20.5) -| (29.5,24.5)  -| (30.5,23.5) -| (31,23.5);
\draw[green,very thick] (28.5,20.5) -| (28.5,21);
\drawgreen{28}{20}
\drawgreen{29}{20}
\drawgreen{29}{21}
\drawgreen{29}{22}
\drawgreen{29}{23}
\drawgreen{29}{24}
\drawgreen{30}{24}
\drawgreen{30}{23}

\draw[green,very thick] (31,23.5) -| (32.5,27.5)  -| (33.5,26.5) -| (34,26.5);
\draw[green,very thick] (31.5,23.5) -| (31.5,24);
\drawgreen{31}{23}
\drawgreen{32}{23}
\drawgreen{32}{24}
\drawgreen{32}{25}
\drawgreen{32}{26}
\drawgreen{32}{27}
\drawgreen{33}{27}
\drawgreen{33}{26}

%chemin de liaison vert-vert
\draw[green!50!black,very thick] (19.5,18) |- (18.5,18.5) |- (18.5,19);
\drawg{19}{18}
\drawg{18}{18}

\draw[green!75!black,very thick] (22.5,21) |- (22,21.5);
\drawgre{22}{21}

\draw[green!75!black,very thick] (25.5,24) |- (25,24.5);
\drawgre{25}{24}

\draw[green!75!black,very thick] (28.5,27) |- (28,27.5);
\drawgre{28}{27}

\draw[green!75!black,very thick] (31.5,30) |- (31,30.5);
\drawgre{31}{30}

\draw[green,very thick] (25.5,18) |- (24.5,18.5) |- (25,19.5);
\drawgreen{25}{18}
\drawgreen{24}{18}
\drawgreen{24}{19}

\draw[green,very thick] (28.5,21) |- (28,21.5);
\drawgreen{28}{21}

\draw[green,very thick] (31.5,24) |- (31,24.5);
\drawgreen{31}{24}

%peigne vert 2 
\draw[green,very thick] (25,19.5) -| (25.5,22.5) |- (28,22.5);
\drawgreen{25}{19}
\drawgreen{25}{20}
\drawgreen{25}{21}
\drawgreen{25}{22}
\drawgreen{26}{22}
\drawgreen{27}{22}

\draw[green,very thick] (28,22.5) -| (28.5,25.5) |- (31,25.5);
\drawgreen{28}{22}
\drawgreen{28}{23}
\drawgreen{28}{24}
\drawgreen{28}{25}
\drawgreen{29}{25}
\drawgreen{30}{25}

\draw[green,very thick] (31,25.5) -| (31.5,28.5) |- (34,28.5);
\drawgreen{31}{25}
\drawgreen{31}{26}
\drawgreen{31}{27}
\drawgreen{31}{28}
\drawgreen{32}{28}
\drawgreen{33}{28}

\path [dotted, draw, thin] (0,0) grid[step=0.35cm] (34,36);
\end{tikzpicture}
\caption{An assembly of complexity less than $2$ can be decomposed as the union of a finite assembly (dark colors), some assemblies of complexity $1$ (medium colors) and all the translations of an assembly of complexity less than $1$ by a given vector (light color).}
\label{fig:appC:step4}
\end{figure}

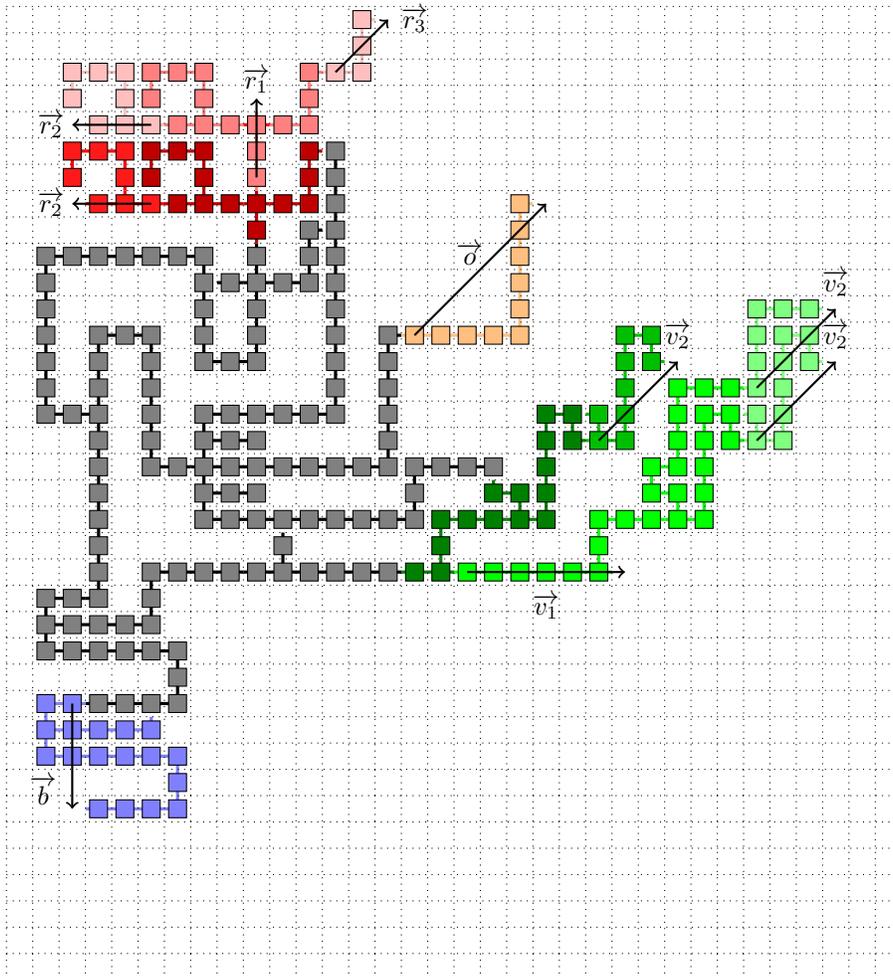
\begin{figure}
\centering
\begin{tikzpicture}[x=0.35cm,y=0.35cm]

%graine
\draw[very thick] (7.5,21.5) -| (12.5,31.5);
\draw[very thick] (7.5,21.5) |- (15.5,17.5) |- (18.5,19.5);
\draw[very thick] (7.5,20.5) |- (9.5,20.5);
\draw[very thick] (7,19.5) |- (10,19.5);
\draw[very thick] (7.5,18.5) |- (9.5,18.5);
\drawgray{7}{21}
\drawgray{7}{20}
\drawgray{7}{19}
\drawgray{7}{18}
\drawgray{7}{17}
\drawgray{8}{21}
\drawgray{8}{20}
\drawgray{8}{19}
\drawgray{8}{18}
\drawgray{8}{17}
\drawgray{9}{21}
\drawgray{9}{20}
\drawgray{9}{19}
\drawgray{9}{18}
\drawgray{9}{17}
\drawgray{10}{21}
\drawgray{11}{21}
\drawgray{12}{21}
\drawgray{12}{22}
\drawgray{12}{23}
\drawgray{12}{24}
\drawgray{12}{25}
\drawgray{12}{26}
\drawgray{12}{27}
\drawgray{12}{28}
\drawgray{12}{29}
\drawgray{12}{30}
\drawgray{12}{31}
\drawgray{10}{17}
\drawgray{11}{17}
\drawgray{12}{17}
\drawgray{13}{17}
\drawgray{14}{17}
\drawgray{15}{17}
\drawgray{15}{18}
\drawgray{15}{19}
\drawgray{16}{19}
\drawgray{17}{19}
\drawgray{18}{19}

%chemin de liaison 1
\draw[very thick] (7,19.5) -| (5.5,24.5) -| (3.5,14.5) -| (3,14.5);
\draw[very thick] (3.5,21.5) -| (1.5,27.5) -| (7.5,23.5) -| (9.5,25) ;
\drawgray{6}{19}
\drawgray{5}{19}
\drawgray{5}{20}
\drawgray{5}{21}
\drawgray{5}{22}
\drawgray{5}{23}
\drawgray{2}{21}
\drawgray{1}{21}
\drawgray{1}{22}
\drawgray{1}{23}
\drawgray{1}{24}
\drawgray{1}{25}
\drawgray{1}{26}
\drawgray{1}{27}
\drawgray{2}{27}
\drawgray{3}{27}
\drawgray{4}{27}
\drawgray{5}{27}
\drawgray{6}{27}
\drawgray{7}{27}
\drawgray{7}{26}
\drawgray{7}{25}
\drawgray{7}{24}
\drawgray{7}{23}
\drawgray{8}{23}
\drawgray{9}{23}
\drawgray{9}{24}
\drawgray{5}{24}
\drawgray{4}{24}
\drawgray{3}{24}
\drawgray{3}{23}
\drawgray{3}{22}
\drawgray{3}{21}
\drawgray{3}{20}
\drawgray{3}{19}
\drawgray{3}{18}
\drawgray{3}{17}
\drawgray{3}{16}
\drawgray{3}{15}
\drawgray{3}{14}

%chemin de liaison 2
\draw[very thick] (2,13.5) -| (5.5,15.5)  -| (9,15.5);
\drawgray{2}{13}
\drawgray{3}{13}
\drawgray{4}{13}
\drawgray{5}{13}
\drawgray{5}{14}
\drawgray{5}{15}
\drawgray{6}{15}
\drawgray{7}{15}
\drawgray{8}{15}

\draw[blue!50!white,very thick] (2,9.5) -| (5.5,10);
\drawb{2}{9}
\drawb{3}{9}
\drawb{4}{9}
\drawb{5}{9}

%chemin fini sur peigne
\draw[red!75!black,very thick] (7.5,30) |- (5.5,31.5)  -| (5.5,30.5);
\drawre{7}{30}
\drawre{7}{31}
\drawre{6}{31}
\drawre{5}{31}
\drawre{5}{30}

\draw[red,very thick] (4.5,30) |- (2.5,31.5)  -| (2.5,30.5);
\drawred{4}{30}
\drawred{4}{31}
\drawred{3}{31}
\drawred{2}{31}
\drawred{2}{30}

\draw[red!50!white,very thick] (7.5,33) |- (5.5,34.5)  -| (5.5,33.5);
\drawr{7}{33}
\drawr{7}{34}
\drawr{6}{34}
\drawr{5}{34}
\drawr{5}{33}

\draw[red!25!white,very thick] (4.5,33) |- (2.5,34.5)  -| (2.5,33.5);
\drawlr{4}{33}
\drawlr{4}{34}
\drawlr{3}{34}
\drawlr{2}{34}
\drawlr{2}{33}

%pompage rouge principale
\draw(9.5,34.2)node {$\vect{r_1}$};
\draw[very thick] (9.5,25) |- (9.5,28);
\draw[very thick] (9,26.5) |- (10,26.5);
\draw[red!75!black,very thick] (9.5,28) |- (9.5,30);
\draw[red!50!white,very thick] (9.5,30) |- (9.5,33);
\draw[red!75!black,very thick] (9,29.5) |- (10,29.5);
\draw[red,very thick] (9,32.5) |- (10,32.5);
\drawgray{9}{25}
\drawgray{9}{26}
\drawgray{9}{27}
\drawre{9}{28}
\drawre{9}{29}
\drawr{9}{30}
\drawr{9}{31}
\drawr{9}{32}
\draw[thick,->] (9.5,30.5) -- (9.5,33.5);

%peigne rouge 1
\draw[very thick] (8,26.5) |- (9,26.5);
\drawgray{8}{26}

\draw[red,very thick] (3,29.5) |- (6,29.5);
\draw[red!75!black,very thick] (6,29.5) |- (9,29.5);
\draw[red!75!black,very thick] (7.5,29.5) |- (7.5,30);
\draw[red,very thick] (4.5,29.5) |- (4.5,30);
\drawre{8}{29}
\drawre{7}{29}
\drawre{6}{29}
\drawred{5}{29}
\drawred{4}{29}
\drawred{3}{29}

\draw[red!50!white,very thick] (6,32.5) |- (9,32.5);
\draw[red!50!white,very thick] (7.5,32.5) |- (7.5,33);
\drawr{8}{32}
\drawr{7}{32}
\drawr{6}{32}

\draw[red!25!white,very thick] (3,32.5) |- (6,32.5);
\draw[red!25!white,very thick] (4.5,32.5) |- (4.5,33);
\drawlr{5}{32}
\drawlr{4}{32}
\drawlr{3}{32}

\draw(1.7,29.5) node {$\vect{r_2}$};
\draw[thick,->] (5.5,29.5) -- (2.5,29.5);

\draw(1.7,32.5) node {$\vect{r_2}$};
\draw[thick,->] (5.5,32.5) -- (2.5,32.5);

%peigne rouge 2
\draw[red!50!white,very thick] (10,32.5) -| (11.5,34.5) -| (12,34.5);
\drawr{10}{32}
\drawr{11}{32}
\drawr{11}{33}
\drawr{11}{34}

\draw[red!25!white,very thick] (12,34.5) -| (13.5,36.5) -| (14,36.5);
\drawlr{12}{34}
\drawlr{13}{34}
\drawlr{13}{35}
\drawlr{13}{36}

\draw[red!75!black,very thick] (10,29.5) -| (11.5,31.5) -| (12,31.5);
\drawre{10}{29}
\drawre{11}{29}
\drawre{11}{30}
\drawre{11}{31}

\draw[very thick] (10,26.5) -| (11.5,28.5) -| (12,28.5);
\drawgray{10}{26}
\drawgray{11}{26}
\drawgray{11}{27}
\drawgray{11}{28}

%pompage orange
\draw[very thick] (10,19.5) -| (14.5,24.5) -| (15,24.5);
\drawgray{10}{19}
\drawgray{11}{19}
\drawgray{12}{19}
\drawgray{13}{19}
\drawgray{14}{19}
\drawgray{14}{20}
\drawgray{14}{21}
\drawgray{14}{22}
\drawgray{14}{23}
\drawgray{14}{24}

\draw[orange!50!white,very thick] (15,24.5) -| (19.5,29.5) -| (20,29.5);
\drawo{15}{24}
\drawo{16}{24}
\drawo{17}{24}
\drawo{18}{24}
\drawo{19}{24}
\drawo{19}{25}
\drawo{19}{26}
\drawo{19}{27}
\drawo{19}{28}
\drawo{19}{29}

\draw(17.6,27.6) node {$\vect{o}$};
\draw[thick,->] (15.5,24.5) -- (20.5,29.5);

%pompage bleu
\draw[very thick] (3,14.5) |- (1.5,14.5) |- (6.5,12.5) |- (3,10.5);
\draw[very thick] (1.5,13.5) |- (2,13.5);
\drawgray{2}{14}
\drawgray{1}{14}
\drawgray{1}{13}
\drawgray{1}{12}
\drawgray{2}{12}
\drawgray{3}{12}
\drawgray{4}{12}
\drawgray{5}{12}
\drawgray{6}{12}
\drawgray{6}{11}
\drawgray{6}{10}
\drawgray{5}{10}
\drawgray{4}{10}
\drawgray{3}{10}

\draw[blue!50!white,very thick] (3,10.5) |- (1.5,10.5) |- (6.5,8.5) |- (3,6.5);
\draw[blue!50!white,very thick] (1.5,9.5) |- (2,9.5);
\drawb{2}{10}
\drawb{1}{10}
\drawb{1}{9}
\drawb{1}{8}
\drawb{2}{8}
\drawb{3}{8}
\drawb{4}{8}
\drawb{5}{8}
\drawb{6}{8}
\drawb{6}{7}
\drawb{6}{6}
\drawb{5}{6}
\drawb{4}{6}
\drawb{3}{6}

\draw(1.4,7.2) node {$\vect{b}$};
\draw[thick,->] (2.5,10.5) -- (2.5,6.5);

%pompage vert 
\draw[very thick] (9,15.5) |- (15,15.5);
\draw[very thick] (10.5,15.5) |- (10.5,16);
\drawgray{14}{15}
\drawgray{13}{15}
\drawgray{12}{15}
\drawgray{11}{15}
\drawgray{10}{15}
\drawgray{9}{15}

\draw[green!50!black,very thick] (15,15.5) |- (17,15.5);
\draw[green,very thick] (17,15.5) |- (21,15.5);
\draw[green!50!black,very thick] (16.5,15.5) |- (16.5,16);
\drawgreen{20}{15}
\drawgreen{19}{15}
\drawgreen{18}{15}
\drawgreen{17}{15}
\drawg{16}{15}
\drawg{15}{15}

\draw[green,very thick] (21,15.5) |- (23,15.5);
\draw[green,very thick] (22.5,15.5) |- (22.5,16);
\drawgreen{22}{15}
\drawgreen{21}{15}

\draw(20.5,14.2) node {$\vect{v_1}$};
\draw[thick,->] (17.5,15.5) -- (23.5,15.5);

%chemin de liaison vert
\draw[very thick] (10.5,16) |- (10.5,17);
\drawgray{10}{16}

\draw[green!50!black,very thick] (16.5,16) |- (19,17.5);
\drawg{16}{16}
\drawg{16}{17}
\drawg{17}{17}
\drawg{18}{17}

\draw[green,very thick] (22.5,16) |- (25,17.5);
\drawgreen{22}{16}
\drawgreen{22}{17}
\drawgreen{23}{17}
\drawgreen{24}{17}

%peigne vert 1
\draw[green!50!black,very thick] (19,17.5) -| (20.5,21.5)  -| (21.5,20.5) -| (22,20.5);
\draw[green!50!black,very thick] (19.5,17.5) -| (19.5,18);
\drawg{19}{17}
\drawg{20}{17}
\drawg{20}{18}
\drawg{20}{19}
\drawg{20}{20}
\drawg{20}{21}
\drawg{21}{21}
\drawg{21}{20}

\draw[green!75!black,very thick] (22,20.5) -| (23.5,24.5)  -| (24.5,23.5) -| (25,23.5);
\draw[green!75!black,very thick] (22.5,20.5) -| (22.5,21);
\drawgre{22}{20}
\drawgre{23}{20}
\drawgre{23}{21}
\drawgre{23}{22}
\drawgre{23}{23}
\drawgre{23}{24}
\drawgre{24}{24}
\drawgre{24}{23}

\draw(25.5,24.5) node {$\vect{v_2}$};
\draw[thick,->] (22.5,20.5) -- (25.5,23.5);

\draw[green,very thick] (25,17.5) -| (26.5,21.5)  -| (27.5,20.5) -| (28,20.5);
\draw[green,very thick] (25.5,17.5) -| (25.5,18);
\drawgreen{25}{17}
\drawgreen{26}{17}
\drawgreen{26}{18}
\drawgreen{26}{19}
\drawgreen{26}{20}
\drawgreen{26}{21}
\drawgreen{27}{21}
\drawgreen{27}{20}

\draw[green!50!white,very thick] (28,20.5) -| (29.5,24.5)  -| (30.5,23.5) -| (31,23.5);
\draw[green!50!white,very thick] (28.5,20.5) -| (28.5,21);
\drawlg{28}{20}
\drawlg{29}{20}
\drawlg{29}{21}
\drawlg{29}{22}
\drawlg{29}{23}
\drawlg{29}{24}
\drawlg{30}{24}
\drawlg{30}{23}

%chemin de liaison vert-vert
\draw[green!50!black,very thick] (19.5,18) |- (18.5,18.5) |- (18.5,19);
\drawg{19}{18}
\drawg{18}{18}

\draw[green!75!black,very thick] (22.5,21) |- (22,21.5);
\drawgre{22}{21}

\draw(31.5,24.5) node {$\vect{v_2}$};
\draw[thick,->] (28.5,20.5) -- (31.5,23.5);

\draw[green,very thick] (25.5,18) |- (24.5,18.5) |- (25,19.5);
\drawgreen{25}{18}
\drawgreen{24}{18}
\drawgreen{24}{19}

\draw[green!50!white,very thick] (28.5,21) |- (28,21.5);
\drawlg{28}{21}

%peigne vert 2 
\draw[green,very thick] (25,19.5) -| (25.5,22.5) |- (28,22.5);
\drawgreen{25}{19}
\drawgreen{25}{20}
\drawgreen{25}{21}
\drawgreen{25}{22}
\drawgreen{26}{22}
\drawgreen{27}{22}

\draw[green!50!white,very thick] (28,22.5) -| (28.5,25.5) |- (31,25.5);
\drawlg{28}{22}
\drawlg{28}{23}
\drawlg{28}{24}
\drawlg{28}{25}
\drawlg{29}{25}
\drawlg{30}{25}

\draw(31.5,26.5) node {$\vect{v_2}$};
\draw[thick,->] (28.5,22.5) -- (31.5,25.5);

\draw(15.5,36.5) node {$\vect{r_3}$};
\draw[thick,->] (12.5,34.5) -- (14.5,36.5);

\path [dotted, draw, thin] (0,0) grid[step=0.35cm] (34,37);
\end{tikzpicture}
\caption{An assembly of complexity less than $1$ can be decomposed as the union of a finite assembly (light color) and all the translations of a finite assembly by a given vector (very light color).}
\label{fig:appC:step5}
\end{figure}

\thanks{We thank Damien Woods for his support and advices.}

\bibliographystyle{plain}
\bibliography{dirige}

\begin{thebibliography}{1}

\bibitem{OneTile}
Erik~D. Demaine, Martin~L. Demaine, S{\'a}ndor~P. Fekete, Matthew~J. Patitz,
  Robert~T. Schweller, Andrew Winslow, and Damien Woods.
\newblock One tile to rule them all: Simulating any tile assembly system with a
  single universal tile.
\newblock In {\em ICALP: Proceedings of the 41st International Colloquium on
  Automata, Languages, and Programming}, volume 8572 of {\em LNCS}, pages
  368--379. Springer, 2014.
\newblock Arxiv preprint: \href{http://arxiv.org/abs/1212.4756}{\tt
  arXiv:1212.4756}.

\bibitem{Doty-2011}
David Doty, Matthew~J. Patitz, and Scott~M. Summers.
\newblock Limitations of self-assembly at temperature 1.
\newblock {\em Theoretical Computer Science}, 412(1--2):145--158, 2011.
\newblock Arxiv preprint: \href{http://arxiv.org/abs/0903.1857v1}{\tt
  arXiv:0903.1857v1}.

\bibitem{pumpabilityLargeBound}
Pierre-{\'E}tienne Meunier and Damien Regnault.
\newblock A pumping lemma for noncooperative self-assembly.
\newblock 2015.
\newblock Arxiv preprint: \href{http://arxiv.org/abs/1312.6668v4}{\tt
  arXiv:1312.6668v4} [cs.CC].

\bibitem{STOC2019}
Pierre-Étienne Meunier, Damien Regnault, and Damien Woods.
\newblock The program-size complexity of self-assembled paths.
\newblock In {\em STOC: Proceedings of the 52nd Annual ACM SIGACT Symposium on
  Theory of Computing}, pages 727--737. ACM, 2020.
\newblock Arxiv preprint with full proofs:
  \href{https://arxiv.org/abs/2002.04012}{\tt arXiv:2002.04012} [cs.CC].

\bibitem{STOC2017}
Pierre-Étienne Meunier and Damien Woods.
\newblock The non-cooperative tile assembly model is not intrinsically
  universal or capable of bounded {T}uring machine simulation.
\newblock In {\em STOC: Proceedings of the 49th Annual ACM SIGACT Symposium on
  Theory of Computing}, pages 328--341, Montreal, Canada, 2017. ACM.
\newblock Arxiv preprint with full proofs:
  \href{https://arxiv.org/abs/1702.00353v2}{\tt arXiv:1702.00353v2} [cs.CC].

\bibitem{Roth01}
Paul W.~K. Rothemund.
\newblock {\em Theory and Experiments in Algorithmic Self-Assembly}.
\newblock PhD thesis, University of Southern California, December 2001.

\bibitem{RotWin00}
Paul W.~K. Rothemund and Erik Winfree.
\newblock The program-size complexity of self-assembled squares (extended
  abstract).
\newblock In {\em STOC: Proceedings of the thirty-second annual ACM Symposium
  on Theory of Computing}, pages 459--468, Portland, Oregon, 2000. ACM.

\bibitem{Winf98}
Erik Winfree.
\newblock {\em Algorithmic Self-Assembly of {D}{N}{A}}.
\newblock PhD thesis, California Institute of Technology, June 1998.

\end{thebibliography}

\end{document}